\documentclass[letterpaper,11pt]{article}
\usepackage[utf8]{inputenc}
\usepackage[english]{babel}

\usepackage{booktabs}
\usepackage{multirow}

\usepackage{xspace}
\usepackage{amsmath,amsfonts,amsthm,amssymb}
\usepackage{mathtools}
\usepackage{thmtools}
\usepackage{bm}
\usepackage{verbatim}
\usepackage{algorithm}
\usepackage{algorithmicx}
\usepackage{subcaption}

\usepackage[numbers, sort]{natbib}

\usepackage[margin=1in]{geometry}

\usepackage[dvipsnames,usenames]{xcolor}
\usepackage[colorlinks=true,pdfpagemode=UseNone,urlcolor=RoyalBlue,linkcolor=RoyalBlue,citecolor=OliveGreen,pdfstartview=FitH]{hyperref}

\usepackage{tcolorbox}

\usepackage{todonotes}

\declaretheorem{definition}

\declaretheorem{theorem}
\declaretheorem[sibling=theorem]{lemma}
\declaretheorem[sibling=theorem]{claim}
\declaretheorem[sibling=theorem]{corollary}

\newcommand{\zhiyi}[1]{\todo[inline,color=green!40]{ZH: #1}}

\newcommand{\gaorq}[1]{\todo[inline,color=blue!40]{Ruiquan: #1}}

\newcommand{\eset}{\mathcal{E}}
\newcommand{\uset}{\mathcal{U}}

\newcommand{\head}{\mathsf{H}}
\newcommand{\tail}{\mathsf{T}}
\newcommand{\origin}{\mathsf{O}}

\newcommand{\yes}{\mathsf{M}}
\newcommand{\no}{\mathsf{U}}
\newcommand{\matched}{q_\yes}
\newcommand{\unmatched}{q_\no}
\newcommand{\ready}{q_{\no^2}}
\newcommand{\state}{q}
\newcommand{\choice}{d}

\renewcommand{\vec}[1]{\bm{#1}}

\newcommand*{\defeq}{\stackrel{\text{def}}{=}}

\newcommand{\E}{\mathbf{E}}
\renewcommand{\Pr}{\mathbf{Pr}}

\newcommand{\exante}{\text{\rm ex-ante}}
\newcommand{\forest}{\text{\rm forest}}

\newcommand{\indexset}{\mathcal{I}}

\newcommand{\elabel}{\ell}

\newcommand{\matchinglp}{\textsc{Matching LP}\xspace}
\newcommand{\balancelp}{\textsc{Balance LP}\xspace}

\newcommand{\alg}{\textsc{Alg}}

\newcommand{\balance}{\textsc{Balance}\xspace}
\newcommand{\balanceocs}{\textsc{Balance-OCS}\xspace}
\newcommand{\ranking}{\textsc{Ranking}\xspace}

\title{
    Improved Online Correlated Selection
    \thanks{This is the second version on arXiv. Compared to the first version, this one adds a discussion on two concurrent works on the same topic, gives a more accurate description of previous results, and improves the presentation based on the feedbacks by anonymous reviewers. The conference version appears in FOCS 2021.}
}

\author{
Ruiquan Gao
\thanks{IIIS, Tsinghua University. Email: grq18@mails.tsinghua.edu.cn, hezt18@mails.tsinghua.edu.cn.}
\and
Zhongtian He
\footnotemark[2]
\and
Zhiyi Huang
\thanks{The University of Hong Kong. Email: zhiyi@cs.hku.hk, bisjack@connect.hku.hk, yanzhong.cs@gmail.com.}
\and
Zipei Nie
\thanks{Nine-Chapter Lab, Huawei. Email: niezipei@huawei.com.}
\and
Bijun Yuan
\footnotemark[3]
\and
Yan Zhong
\footnotemark[3]
}

\date{December 2021}

\begin{document}

\begin{titlepage}
    \thispagestyle{empty}
    \maketitle
    \begin{abstract}
        \thispagestyle{empty}
        This paper studies the online correlated selection (OCS) problem.
It was introduced by Fahrbach, Huang, Tao, and Zadimoghaddam (2020) to obtain the first edge-weighted online bipartite matching algorithm that breaks the $0.5$ barrier.
Suppose that we receive a pair of elements in each round and immediately select one of them.
Can we select with negative correlation to be more effective than independent random selections?
Our contributions are threefold.
For semi-OCS, which considers the probability that an element remains unselected after appearing in $k$ rounds, we give an optimal algorithm that minimizes this probability for all $k$.
It leads to $0.536$-competitive unweighted and vertex-weighted online bipartite matching algorithms that randomize over only two options in each round, improving the $0.508$-competitive ratio by Fahrbach et al.~(2020).
Further, we develop the first multi-way semi-OCS that allows an arbitrary number of elements with arbitrary masses in each round.
As an application, it rounds the \balance algorithm in unweighted and vertex-weighted online bipartite matching and is $0.593$-competitive.
Finally, we study OCS, which further considers the probability that an element is unselected in \emph{an arbitrary subset of rounds}.
We prove that the optimal ``level of negative correlation'' is between $0.167$ and $0.25$, improving the previous bounds of $0.109$ and $1$ by Fahrbach et al.~(2020).
Our OCS gives a $0.519$-competitive edge-weighted online bipartite matching algorithm, improving the previous $0.508$-competitive ratio by Fahrbach et al.~(2020).

    \end{abstract}
\end{titlepage}

\tableofcontents
\thispagestyle{empty}
\newpage
\addtocontents{toc}{\protect\thispagestyle{empty}} 
\setcounter{page}{1}

\section{Introduction}
\label{sec:introduction}

Real-life optimization problems often need to make decisions based on the information at hand instead of the full picture in hindsight.
Online advertising platforms show advertisements within milliseconds after receiving each user query.
Ride hailing applications match riders and drivers without full knowledge of future ride requests.
Cloud service providers assign computational tasks to physical servers not knowing what tasks the users may submit later.
Due to the broad applications, the design and analysis of online algorithms for these optimization problems are a central topic in computer science and operations research.

Lacking accurate knowledge of the full picture, there is usually no universally good decision for all possible future input in these online optimization problems.
As a result, online algorithms need to hedge against different possibilities through randomized decisions.
Consider the online bipartite matching problem by \citet*{KarpVV:STOC:1990} as a running example.
We want to find a matching in a bipartite graph and maximize its size.
Initially, we only know the left-hand-side of the bipartite graph, a.k.a., the offline vertices.
Online vertices on the right-hand-side arrive one at a time.
We must immediately and irrevocably match each of them upon arrival.
Any deterministic greedy algorithm gives a maximal matching, and therefore its size is at least half of the maximum matching in hindsight.
Beating this trivial bound of half, however, necessitates randomization even for bipartite graphs with only two vertices on each side.

\bigskip

\begin{tcolorbox}[colback=lightgray!40,arc=0pt,outer arc=0pt,width=\textwidth,boxrule=.5pt]
    \textbf{Example.~}
    Consider vertices $1$ and $2$ on the left and $3$ and $4$ on the right.
    Vertex $3$ arrives first with edges to both $1$ and $2$.
    A deterministic algorithm then immediately matches $3$, e.g., to $1$.
    If vertex $4$ only has an edge to $1$, however, the algorithm cannot match it even though a perfect matching exists in hindsight.
    A randomized algorithm that matches $3$ to $1$ and $2$ with equal probability, on the other hand, matches $\frac{3}{2}$ edges in expectation.
\end{tcolorbox}

\medskip

Would it suffice to independently randomize over two offline neighbors?
Unfortunately, the answer is negative (e.g., \citet*{FahrbachHTZ:FOCS:2020}).
We need to correlate different rounds' selections to break the $\frac{1}{2}$ barrier.
One can introduce correlation through problem specific methods.
The \ranking algorithm of \citet{KarpVV:STOC:1990}, for example, samples a random order of the offline vertices at the beginning, and then matches each online vertex to the first unmatched offline neighbor by that order.
\ranking and its variants achieve the optimal $1-\frac{1}{e}$ competitive ratio for unweighted~\cite{KarpVV:STOC:1990} and vertex-weighted online bipartite matching~\cite{AggarwalGKM:SODA:2011}, but have difficulties extending to the more general edge-weighted problem (a.k.a., Display Ads)~\cite{FeldmanKMMP:WINE:2009} and AdWords~\cite{MehtaSVV:JACM:2007}.

\citet{FahrbachHTZ:FOCS:2020}, on the other hand, formulate a generic online selection problem and design online correlated selection (OCS) algorithms that lead to the first edge-weighted online bipartite matching algorithm that breaks the $\frac{1}{2}$ barrier.
Subsequently, \citet*{HuangZZ:FOCS:2020} break the $\frac{1}{2}$ barrier in AdWords using a similar approach.
The online selection problem considers a set of ground elements (e.g., the offline vertices) and a sequence of pairs of these elements (e.g., a pair of offline neighbors for each online vertex).
The algorithm immediately selects an element upon receiving each pair.
If it independently selects a random element from each pair, then with probability $2^{-k}$ an element remains unselected in a subset of $k$ pairs involving it.
Can we be more effective than independent random selections?

\citet{FahrbachHTZ:FOCS:2020} study two versions of online selection called semi-OCS and OCS.
Semi-OCS focuses on the probability that an element is unselected at the end when it is in $k$ pairs.
They give a semi-OCS that upper bounds this probability by $2^{-k}(1-\gamma)^{k-1}$ for $\gamma = 0.109$, and call it a $\gamma$-semi-OCS.
They also prove that $1$-semi-OCS is impossible.
OCS further considers the probability that an element is unselected in \emph{an arbitrary subset} of pairs involving it.
When the subset is the union of $m$ consecutive subsequences of the pairs involving the element, with lengths $k_1, k_2, \dots, k_m$, their OCS bounds the unselected probability by $\prod_{i=1}^m 2^{-k_i} (1-\gamma)^{k_i-1}$ also for $\gamma = 0.109$.
They call it a $\gamma$-OCS.
The weaker guarantee of semi-OCS is sufficient for unweighted and vertex-weighted online bipartite matching, while the stronger guarantee of OCS is sufficient for the edge-weighted problem.
The main idea is to randomly match the pairs so that two matched pairs share a common element that is not in any pair in between.
Their semi-OCS and OCS then select oppositely from the matched pairs with respect to the common element.

\citet{FahrbachHTZ:FOCS:2020} explicitly leave two open questions:
(1) \emph{What is the best possible $\gamma$ for which $\gamma$-semi-OCS and $\gamma$-OCS exist?}
(2) \emph{Are there multi-way online selection algorithms that select from multiple elements in each round, with sufficient negative correlation such that the resulting online matching algorithms are better than the two-way counterparts?}
We remark that the matching-based approach fails fundamentally in the multi-way extension.
If each round has $n$ elements, it can then be matched to $\Omega(n)$ other rounds. For large $n$, any matching is sparse and the resulting negative correlation is negligible.

\subsection{Our Contributions}

\paragraph{Semi-OCS and Weighted Sampling without Replacement.}
This paper gives a complete answer to the first open question for semi-OCS.
In fact, we not only show that the optimal $\gamma$ equals $\frac{1}{2}$ for semi-OCS, but also find that the unselected probability converges to zero much faster than the guanrantee of $\gamma$-semi-OCS when the element is in $k \ge 3$ rounds.

\paragraph{Informal Theorem 1.}
\emph{
    There is a polynomial-time semi-OCS such that an element that is in $k$ pairs is selected with probability at least $1-2^{-2^k+1}$.
    This is the best possible for all $k \ge 1$.
}

\bigskip

In each round the optimal semi-OCS selects the element that appears more in previous rounds, and is unselected thus far, breaking ties randomly.
It is the limit case of weighted sampling without replacement, when an element's weight is exponential in the number of previous rounds with the element, and when the base of the exponential tends to infinity.
The main lemma in our analysis, which may be of independent interest, shows that the selections of different elements are negatively correlated in weighted sampling without replacement with two elements per round.
See Section~\ref{sec:semi-ocs}.

This paper further answers the second open question affirmatively for semi-OCS by studying a multi-way online selection problem that allows an arbitrary marginal distribution over the elements in each round.
We will refer to the marginal probability of an element as its mass in that round.
Consider the probability that an element is unselected at the end when its total mass is $y$.
On the one hand, sampling independently from the marginal distributions (with replacement) bounds the probability by $\exp(-y)$.
On the other hand, for an overly idealized algorithm, which samples from the marginal distributions and ensures that each element is sampled at most once, this probability is $\max \{1-y, 0\}$.
For $y \in [0, 1]$, it is $\exp(-y-\frac{y^2}{2}-\frac{y^3}{3}-\dots)$ by the Taylor series of $\ln(1-y)$.
We match the overly idealized bound up to the quadratic term, with a smaller cubic coefficient.

\paragraph{Informal Theorem 2.}
\emph{
    There is a polynomial-time multi-way semi-OCS such that an element with total mass $y$ is selected with probability at least $1-\exp(-y-\frac{y^2}{2}-\frac{4-2\sqrt{3}}{3}y^3)$.
}

\bigskip

It may be tempting to conjecture that sampling independently from the marginal distributions \emph{without replacement} already improves the trivial bound of $\exp(-y)$.
Unfortunately, this is false.
Consider an element $e$ that is in all $T$ rounds each with mass $\epsilon = \frac{y}{T}$, and let there be a distinct element other than $e$ with mass $1-\epsilon$ in each round.
Then element $e$ remains unselected at the end with probability $\exp(-y)$ when $T$ tends to infinity.

This example suggests that if an element has accumulated some mass and is still unselected, we shall give it a higher priority than new elements that are not in any previous rounds.
It motivates weighted sampling without replacement where the weight is a function of the total mass of the element in previous rounds.
We choose the weight to be the inverse of the upper bound on the unselected probability so that \emph{the expected sampling weight of any element is at most its mass in the round}, an invariant that is the key to our analysis.
See Section~\ref{sec:multi-way}.

\begin{table}[t]
    \centering
    \caption{A summary of the results in this paper on online correlated selection and their applications in online bipartite matching, with a comparison to those by \citet{FahrbachHTZ:FOCS:2020}.}
    \label{tab:summary}
    \renewcommand{\arraystretch}{1.25}
    \begin{tabular}{lcc}
        \toprule
        & \citet{FahrbachHTZ:FOCS:2020} & \textbf{This Paper} \\
        \midrule
        Semi-OCS\footnotemark[1] & $2^{-k} (1-0.109)^{k-1}$ & $2^{-2^k+1}$ \\
        Multi-way Semi-OCS\footnotemark[2] & - & $\exp(-y-\frac{y^2}{2}-\frac{4-2\sqrt{3}}{3} y^3)$ \\
        $\gamma$-OCS & $0.109 \le \gamma < 1$ & $0.167 \le \gamma \le \frac{1}{4}$ \\
        Unweighted/Vertex-weighted ($2$-Way) & $0.508$ & $0.536$ \\
        Unweighted/Vertex-weighted (Multi-way)\footnotemark[3] & - & $0.593$ \\
        Edge-weighted & $0.508$ & $0.519$ \\
        \bottomrule%
    \end{tabular}
    
    \medskip

    \begin{minipage}{.95\textwidth}
        \footnotesize
        \footnotemark[1]
        The table presents upper bounds on the probability that an element is unselected when it is in $k$ pairs.%
        \\
        \footnotemark[2]
        The table presents upper bounds on the probability that an element is unselected when its total mass is $y$.%
        \\
        \footnotemark[3]
        \ranking by \citet{KarpVV:STOC:1990} is $1-\frac{1}{e}$-competitive which is optimal.
        Nonetheless, the algorithms in this paper are the first ones other than \ranking whose competitive ratios are beyond the $0.5+\epsilon$ regime.
    \end{minipage}
\end{table}

\paragraph{OCS and Probabilistic Automata.}
This paper also contributes to the first open question for OCS by narrowing the gap between the upper and lower bounds on the best possible $\gamma$.

\paragraph{Informal Theorem 3.}
\emph{
    There is a polynomial-time $0.167$-OCS.
    Further, there is no $\gamma$-OCS for any $\gamma > \frac{1}{4}$, even with unlimited computational power.
}

\bigskip

The improved OCS also abandons the matching-based approach and instead introduces an automata-based approach.
Informally, it picks an element from each round to probe the element's state.
If the element was selected last time, select the other element this time.
If the element has not been selected in the last two appearances, select it this time.
Finally, only when the element was not selected last time but was selected before that, the OCS selects with fresh randomness.
The actual algorithm is more involved.
For example, we cannot pick an element to probe independently in each round in general.
See Section~\ref{sec:ocs} for detail.

\paragraph{Applications in Online Bipartite Matching.}
The new results on online correlated selection from this paper lead to better online bipartite matching algorithms.
For the unweighted and vertex-weighted problems, we get a $0.536$-competitive two-way algorithm, improving the $0.508$-competitive algorithm by \citet{FahrbachHTZ:FOCS:2020}.
We further show that the multi-way semi-OCS can round the (fractional) \balance algorithm (e.g., \cite{KalyanasundaramP:TCS:2000, MehtaSVV:JACM:2007}), and be $0.593$-competitive.
For the edge-weighted problem, the $0.167$-OCS gives a $0.512$-competitive algorithm.
In the process, we refine the reductions from online matching problems to online correlated selection so that the competitive ratios admit close-formed expressions, and a guarantee strictly weaker than $\gamma$-OCS already suffices for the edge-weighted problem.
Motivated by the relaxed guarantee, we design a variant of OCS that further improves the edge-weighted competitive ratio.
See Section~\ref{sec:matching}.

\paragraph{Informal Theorem 4.}
\emph{
    There is a polynomial-time $0.519$-competitive algorithm for edge-weighted online bipartite matching.%
    \footnote{We assume free disposals, which is standard in the edge-weighted problem under worst-case analysis.}
}

\subsection{Other Related Works}

\paragraph{Online Rounding.}
Online correlated selection is related to online rounding algorithms.
It is common to first design online algorithms for an easier fractional online optimization problem, and to round it using online rounding algorithms to solve the original integral problem.
Independent rounding is the simplest and most general online rounding;
it corresponds to independent random selections in the online selection problem in this paper.
For instance, \citet{BuchbinderN:MOR:2009} first design fractional online covering and packing algorithms under the online primal-dual framework, and then round them with independent rounding.
More involved online rounding algorithms are usually designed on a problem-by-problem basis in the literature, e.g., for $k$-server~\cite{BansalBMN:JACM:2015}, online submodular maximization~\cite{ChanHJKT:TALG:2018}, online edge coloring~\cite{CohenPW:FOCS:2019}, etc.

To our knowledge, the only general online rounding method other than independent rounding is the online contention resolution schemes initiated by \citet{FeldmanSZ:SODA:2016} and further developed by \citet{AdamczykW:FOCS:2018}, \citet{LeeS:ESA:2018}, and \citet{Dughmi:ICALP:2020, Dughmi:arXiv:2021}.
It has found applications mainly in online problems with stochastic information, such as prophet inequality~\cite{FeldmanSZ:SODA:2016, EzraFGT:EC:2020}, posted pricing~\cite{FeldmanSZ:SODA:2016}, stochastic probing~\cite{FeldmanSZ:SODA:2016}, and stochastic matching~\cite{GamlathKS:SODA:2019, EzraFGT:EC:2020, FuTWWZ:ICALP:2021}.

There is an important difference between the above usages of online rounding algorithms and the applications of OCS in online bipartite matching.
The above online rounding algorithms and the corresponding fractional online algorithms are designed separately;
the final competitive ratio is the product of their ratios.%
\footnote{Using this two-step approach to analyze the applications of the multi-way semi-OCS in this paper in unweighted and vertex-weighted online bipartite matching would lead to a much worse competitive ratio of about $0.514$. Using it with the (two-way) semi-OCS and OCS even gives a ratio strictly smaller than $0.5$!}
By contrast, this paper and previous works on OCS~\cite{FahrbachHTZ:FOCS:2020, HuangZZ:FOCS:2020} take an end-to-end approach:
the online matching algorithms make fractional decisions based on the guarantee of OCS to directly optimize the expected objective of the rounded matching.
For example, the algorithms for vertex-weighted and edge-weighted matching in Section~\ref{sec:matching} rely on discount functions derived from optimization problems that take the OCS guarantees as parameters.
Another example with a similar spirit is the convex rounding technique by \citet{DughmiRY:JACM:2016} and \citet{Dughmi:EC:2011} from the algorithmic game theory literature.

\paragraph{Online Matching.}
We refer readers to \citet{Mehta:FTTCS:2013} for a survey on online matching problems.
The unweighted, vertex-weighted, and edge-weighted online bipartite matching problems are first studied by \citet{KarpVV:STOC:1990}, \citet{AggarwalGKM:SODA:2011}, and \citet{FeldmanKMMP:WINE:2009}.
Later, \citet{DevanurJK:SODA:2013} and \citet{DevanurHKMY:TEAC:2016} simplify the analyses under the online primal-dual framework.
In particular, \citet{DevanurHKMY:TEAC:2016} view the expected maximal edge-weight matched to an offline vertex as an integral of the complementary cumulative distribution function, a key ingredient of the application of OCS in edge-weighted matching.
Finally, \citet{BuchbinderNW:Arxiv:2021}, \citet{CohenW:SODA:2018}, \citet{GamlathKMSW:FOCS:2019}, \citet{PapadimitriouPSW:EC:2021}, and \citet{SaberiW:ICALP:2021} also build on negative correlation properties to analyze their online algorithms, although these negative correlation properties and their usage are orthogonal to those in this paper.

\subsection{Concurrent Works}

Concurrently and independently, \citet{BlancC:FOCS:2021} and \citet{ShinA:Arxiv:2021} also improved the results of \citet{FahrbachHTZ:FOCS:2020}.
They mainly study multi-way OCS, while this paper focuses on $2$-way semi-OCS, $2$-way OCS, and multi-way semi-OCS.
Hence, these two papers are almost orthogonal to ours.
Using a $6$-way OCS, \citet{BlancC:FOCS:2021} obtained a $0.5368$-competitive algorithm for edge-weighted online bipartite matching.
They also gave similar simplifications for the reduction of online matching problems to OCS. 
\citet{ShinA:Arxiv:2021} gave a method for converting $2$-way OCS to $3$-way OCS.
Applying their method to the OCS of \citet{FahrbachHTZ:FOCS:2020} gives a $0.509$-competitive algorithm for edge-weighted online bipartite matching.
Applying it to the improved $2$-way OCS in this paper further improves the ratio to $0.513$.

\section{Preliminaries}
\label{sec:preliminary}

The \emph{online selection problem} considers a set $\eset$ of elements and a selection process that proceeds in $T$ rounds.
For any round $1 \le t \le T$, a pair of elements $\eset^t$ arrive and the \emph{online selection algorithm} needs to immediately select an element from $\eset^t$.
Let $s^t$ denote the selected element in round $t$.

For any subset of rounds $T' \subseteq T$, we say that an element $e$ is \emph{unselected} in $T'$ if the algorithm does not select the element in any round in $T'$, i.e., if $s^t \ne e$ for any $t \in T'$.
If $T' = T$, we simply say that element $e$ is unselected.
For any $0 \le t \le T$, let $\uset^t$ denote the set of elements that are unselected in rounds $1, 2, \dots, t$.
Let $\uset = \uset^T$ denote the set of unselected elements at the end.

Semi-OCS considers the probability that an element $e \in \eset$ is unselected at the end, and seeks to bound it as a function of the number rounds containing element $e$.

\begin{definition}[$\gamma$-semi-OCS, c.f., \citet{FahrbachHTZ:FOCS:2020}]
    An online selection algorithm is a $\gamma$-semi-OCS if for any online selection instance and any element $e$ that appears in $k$ rounds, element $e$ is unselected with probability at most:
    \[
        2^{-k} (1-\gamma)^{k-1}
        ~.
    \]
\end{definition}

Selecting an element in each round independently and uniformly at random is a $0$-semi-OCS.
\citet{FahrbachHTZ:FOCS:2020} give a $0.109$-semi-OCS and prove that there is no $1$-semi-OCS.

OCS further considers the probability that an element $e \in \eset$ is unselected in an \emph{arbitrary subset of rounds} containing the element.
The upper bounds on this probability depend on the structure of the subset of rounds.
A consecutive subsequence of the rounds containing element $e$ is a subset of rounds $\{t_1, t_2, \dots, t_k\}$ such that each round $t_i$ contains $e$, i.e., $e \in \eset^{t_i}$ for any $1 \le i \le k$, and no round in between contains $e$, i.e., $e \notin \eset^t$ for any $1 \le i \le k-1$ and any $t_i < t < t_{i+1}$.

\begin{definition}[$\gamma$-OCS, c.f., \citet{FahrbachHTZ:FOCS:2020}]
    An online selection algorithm is a $\gamma$-OCS if for any online selection instance, any element $e$, and any subset of rounds $T' \subseteq T$ containing $e$ such that $T'$ is the union of $m$ consecutive subsequences of the rounds containing $e$, with lengths $k_1, k_2, \dots, k_m$, element $e$ is unselected in $T'$ with probability at most:
    \[
        \prod_{i=1}^m 2^{-k_i} (1-\gamma)^{k_i-1}
        ~.
    \]
\end{definition}

For example, suppose that rounds $1, 2, 5, 6, 9$ are the ones that contain element $e$, and consider $T' = \{1, 2, 6, 9\}$.
Then, $T'$ is the union of two consecutive subsequences $1, 2$ and $6, 9$, whose lengths are $2$.
\citet{FahrbachHTZ:FOCS:2020} give a $0.109$-OCS.
Since OCS is stronger than semi-OCS, the impossibility result for semi-OCS implies that there is no $1$-OCS.

\section{Optimal Semi-OCS}
\label{sec:semi-ocs}

\subsection{Algorithms}

This paper considers a semi-OCS that remembers the number of rounds involving each element thus far, and selects from each round the element that appears more and is unselected so far, breaking ties uniformly at random and independently in different rounds.
See Algorithm~\ref{alg:optimal-semi-ocs}.

\begin{algorithm}
    \caption{Optimal Semi-OCS}
    \label{alg:optimal-semi-ocs}
    \begin{algorithmic}
        \medskip
        \State \textbf{State variables:} (for each element $e$)
        \begin{itemize}
            \item The number of previous rounds that contain element $e$, denoted as $k_e$.
            \item Whether element $e$ has been selected in any previous rounds.
        \end{itemize}
        \State \textbf{For each round $t$:}
            (suppose $\eset^t = \{e, e'\}$)
        \begin{enumerate}
            \item If both $e$ and $e'$ have been selected, select arbitrarily, e.g., $s^t \in \{ e, e' \}$ uniformly at random.
            \item If only one of $e$ and $e'$ has been selected, select $s^t$ to be the one that has not been selected.
            \item If neither $e$ nor $e'$ has been selected:
            \begin{itemize}
                \item If $k_e \ne k_{e'}$, select $s^t$ to be the one with more previous appearances.
                \item Otherwise, select $s^t \in \{e, e'\}$ uniformly at random.
            \end{itemize}
        \end{enumerate}
    \end{algorithmic}
\end{algorithm}

The main lemma in the analysis of Algorithm~\ref{alg:optimal-semi-ocs} will prove that the (un)selections of elements are negatively correlated.
It is more instructive to prove this lemma for a broader family of weighted sampling algorithms (Algorithm~\ref{alg:weighted-2way-sampling}).
Algorithm~\ref{alg:optimal-semi-ocs} is the special case when we let $w_e^t = 1$ if $e$ appears in previous rounds at least as many times as the other element does, and let $w_e^t = 0$ otherwise.

\begin{algorithm}
    \caption{Weighted $2$-Way Sampling without Replacement}
    \label{alg:weighted-2way-sampling}
    \begin{algorithmic}
        \medskip
        \State \textbf{Parameters:}
        \begin{itemize}
            \item $w_e^t \ge 0$, weight of element $e \in \eset^t$ in round $t$
        \end{itemize}
        \State \textbf{For each round $t$:}
        \begin{enumerate}
            \item If both elements in $\eset^t$ have been selected, select arbitrarily, e.g., uniformly at random.
            \item Otherwise, select each unselected element $e \in \eset^t$ with probability proportional to $w^t_e$.
        \end{enumerate}
    \end{algorithmic}
\end{algorithm}


\subsection{Negative Correlation in Weighted $2$-Way Sampling without Replacement}

Recall that $\uset^t$ denotes the set of unselected elements after the first $t$ rounds.
Hence the event that a subset $S \subseteq \eset$ of elements are all unselected after the first $t$ rounds can be written as $S \subseteq \uset^t$.
We shall establish the negative correlation of such events in the next lemma.

\begin{lemma}
    \label{lem:2way-sampling-negative-correlation}
    For weighted $2$-way sampling (Algorithm~\ref{alg:weighted-2way-sampling}) with any weights, any $0 \le t \le T$, and any disjoint subsets of elements $A, B \subseteq \eset$:
    \[
        \Pr \big[ A \cup B \subseteq \uset^t \big] \le \Pr \big[ A \subseteq \uset^t \big] \Pr \big[ B \subseteq \uset^t \big]
        ~.
    \]
\end{lemma}

\begin{proof}
    We shall prove the lemma by induction on $t$.
    The base case when $t = 0$ is trivial since $\uset^0 = \eset$, and thus $\Pr \big[ A \subseteq \uset^0 \big] = \Pr \big[ B \subseteq \uset^0 \big] = \Pr \big[ A \cup B \subseteq \uset^0 \big] = 1$.
    Next suppose that the lemma holds for round $t-1$ and consider round $t$.
    
    \paragraph{Case 1:}
    $\eset^t \cap (A \cup B) = \emptyset$, i.e., no element in this round belongs to $A$ or $B$.
    Since the selection in round $t$ does not affect the events of concern, 
    %
    %
    the lemma continues to hold after round $t$ by the inductive hypothesis.
     
    \paragraph{Case 2:}
    $\big| \eset^t \cap (A \cup B) \big| = 1$, i.e., exactly one element in this round belongs to $A$ or $B$.
    Denote this element as $e \in \eset^t$ and the other element as $e' \in \eset^t$.
    Further suppose without loss of generality that $e \in A$.
    Since the elements in $B$ are not involved in round $t$, we have:
    %
    %
    \[
        \Pr \big[ B \subseteq \uset^t \big] = \Pr \big[ B \subseteq \uset^{t-1} \big]
        ~.
    \]

    Next consider the elements in $A$.
    If $e'$ has been selected in the first $t-1$ rounds, $e$ would certainly be selected after round $t$.
    Hence, to have $A \subseteq \uset^t$, we need not only $A \subseteq \uset^{t-1}$, but also $e' \in \uset^{t-1}$.
    Further the algorithm must select $e'$ in round $t$.
    Putting together we have:
    \[
        \Pr \big[ A \subseteq \uset^t \big] = \Pr \big[ A \cup \{ e' \} \subseteq \uset^{t-1} \big] \frac{w^t_{e'}}{w^t_e + w^t_{e'}}
        ~.
    \]

    Similarly we have:
    \[
        \Pr \big[ A \cup B \subseteq \uset^T \big] = \Pr \big[ A \cup B \cup \{ e' \} \subseteq \uset^{t-1} \big] \frac{w^t_{e'}}{w^t_e + w^t_{e'}}
        ~.
    \]

    Cancelling the common term $\frac{w^t_{e'}}{w^t_e + w^t_{e'}}$, the inequality in the lemma is equivalent to:
    \[
        \Pr \big[ A \cup B \cup \{ e' \} \subseteq \uset^{t-1} \big] 
        \le
        \Pr \big[ A \cup \{ e' \} \subseteq \uset^{t-1} \big] \Pr \big[ B \subseteq \uset^{t-1} \big]
        ~.
    \]

    This follows by the inductive hypothesis for subsets $A \cup \{ e' \}$ and $B$ in round $t-1$.

    \paragraph{Case 3:}
    $\eset^t \subseteq (A \cup B)$, i.e., both elements in round $t$ belong to $A$ or $B$.
    Since one element in $\eset^t$ is selected in round $t$, we have $\Pr [ A \cup B \subseteq \uset^t ] = 0$.
    %
    %
    Hence the stated inequality trivially holds.
\end{proof}

We remark that the lemma no longer holds if we have $3$ or more elements in each round.
Appendix~\ref{app:positive-correlation} provides a counter-example.
See also \citet{Alexander:AnnStat:1989}.

\subsection{Analysis}

\begin{theorem}
    \label{thm:optimal-semi-ocs}
    For any instance and any element that appears in $k$ rounds, the probability that the element is never selected by Algorithm~\ref{alg:optimal-semi-ocs} is at most:
    \[
        2^{-2^k+1}
        ~.
    \]
\end{theorem}

\begin{proof}
    We shall prove the theorem by induction on the number of rounds $T$ in the instance.
    The base case when $T = 0$ is trivial since $k$ must be $0$ in this case.
    Next suppose that the lemma holds for up to $T-1$ rounds.
    Consider an arbitrary instance with $T$ rounds, and any element $e$ that appears $k$ times.
    Without loss of generality, we may assume that $e$ is in the last round $T$;
    otherwise it follows directly from the inductive hypothesis.
    Further suppose that the other element in round $T$ is $e'$.
    Consider three cases depending on the relation between the number of appearances $k_e$ and $k_{e'}$ before round $T$.
    Observe that $k = k_e + 1$.

    \paragraph{Case 1:}
    $k_e > k_{e'}$.
    %
    By the definition of Algorithm~\ref{alg:optimal-semi-ocs}, element $e$ is selected with certainty after $T$.
    Hence the probability of concern is $0$, and is trivially smaller than the stated bound.

    \paragraph{Case 2:}
    $k_e < k_{e'}$.
    By the definition of Algorithm~\ref{alg:optimal-semi-ocs}, element $e'$ is selected with certainty in round $T$ if it is not yet selected previously.
    Hence, element $e$ is never selected by the algorithm at the end if and only if both $e$ and $e'$ are unselected before round $T$.
    This probability is:
    \begin{align*}
        \Pr \big[ \{ e, e' \} \subseteq \uset^{T-1} \big]
        &
        \le
        \Pr \big[ e \in \uset^{T-1} \big] \Pr \big[ e' \in \uset^{T-1} \big]
        \tag{Lemma~\ref{lem:2way-sampling-negative-correlation}} \\
        &
        \le 2^{-2^{k_e}+1} \cdot 2^{-2^{k_{e'}}+1}
        \tag{Inductive hypothesis} \\
        &
        \le 2^{-2^{k_e}+1} \cdot 2^{-2^{k_e+1}+1}
        \tag{$k_e < k_{e'}$} \\
        &
        \le 2^{-2^k+1}
        \tag{$k = k_e + 1$}
        ~.
    \end{align*}

    \paragraph{Case 3:}
    $k_e = k_{e'}$.
    By the definition of Algorithm~\ref{alg:optimal-semi-ocs}, elements $e$ and $e'$ would be selected with equal probability if neither has been selected before.
    Therefore, element $e$ is never selected by the algorithm at the end if and only if both $e$ and $e'$ are unselected before round $T$, and the algorithm selects $e'$ in round $T$.
    The latter happens with probability half and is independent with the former.
    Hence, this probability is:
    \begin{align*}
        2^{-1} \Pr \big[ \{ e, e' \} \subseteq \uset^{T-1} \big]
        &
        \le
        2^{-1} \Pr \big[ e \in \uset^{T-1} \big] \Pr \big[ e' \in \uset^{T-1} \big]
        \tag{Lemma~\ref{lem:2way-sampling-negative-correlation}} \\
        &
        \le 2^{-1} \cdot 2^{-2^{k_e}+1} \cdot 2^{-2^{k_{e'}}+1}
        \tag{Inductive hypothesis} \\
        &
        \le 2^{-1} \cdot 2^{-2^{k_e}+1} \cdot 2^{-2^{k_e}+1}
        \tag{$k_e = k_{e'}$} \\
        &
        = 2^{-2^k+1}
        \tag{$k = k_e + 1$}
        ~.
    \end{align*}

    Summarizing the three cases completes the inductive step and thus the proof of the theorem.
\end{proof}

Since $2^{-2^k+1} \le 2^{-2k+1} = 2^{-k}(1-\frac{1}{2})^{k-1}$, Theorem~\ref{thm:optimal-semi-ocs} leads to the following corollary in terms of the original definition of semi-OCS.

\begin{corollary}
    \label{cor:optimal-semi-ocs}
    Algorithm~\ref{alg:optimal-semi-ocs} is a $\frac{1}{2}$-semi-OCS.
\end{corollary}

We remark that the guarantee of $\frac{1}{2}$-semi-OCS only gives an $\frac{8}{15} \approx 0.533$-competitive algorithm for unweighted and vertex-weighted matching, while Theorem~\ref{thm:optimal-semi-ocs} leads to least $0.536$.
That is, the tighter analysis in Theorem~\ref{thm:optimal-semi-ocs} indeed results in better competitive ratios in online matching.

\subsection{Hardness}

Finally, we show that the semi-OCS (Algorithm~\ref{alg:optimal-semi-ocs}) and its analysis (Theorem~\ref{thm:optimal-semi-ocs}) are optimal for all $k$ simultaneously.
The proof is deferred to Appendix~\ref{app:semi-ocs-hardness}.

\begin{theorem}
    \label{thm:semi-ocs-hardness}
    For any algorithm and any $k \ge 0$, there is an instance and an element that appears in $k$ rounds, such that with probability at least $2^{-2^k+1}$ the algorithm never selects the element.
\end{theorem}

The special case of $k = 2$ further implies a hardness for the original definition of $\gamma$-semi-OCS.

\begin{corollary}
    \label{cor:semi-ocs-hardness}
    There is no $\gamma$-semi-OCS for $\gamma > \frac{1}{2}$.
\end{corollary}

\section{Multi-way Semi-OCS}
\label{sec:multi-way}

\subsection{Definitions}

The \emph{multi-way online selection problem} considers a set of elements $\eset$ and a selection process that proceeds in $T$ rounds as follows.
Each round $1 \le t \le T$ is associated with a non-negative vector $\vec{x}^t = (x^t_e)_{e \in \eset}$ such that $\sum_{e \in \eset} x^t_e = 1$.
We shall refer to $x^t_e$ as the \emph{mass} of element $e$ in round $t$.
The vectors are unknown at the beginning and are revealed to an \emph{multi-way online selection algorithm} at the corresponding rounds.
Let $\eset^t = \{ e : x^t_e > 0 \}$ be the set of elements with positive masses in round $t$, i.e., those that may be selected in the round.
Upon observing the mass vector $\vec{x}^t$ for round $t$, the algorithm selects an element from $\eset^t$.

We may interpret $x^t_e$ as the probability of selecting element $e$ in the round \emph{if none of the elements have appeared in previous rounds}, although in general the correlation introduced by the multi-way online selection algorithms will complicate the selection probabilities.
For any $0 \le t \le T$, let $y_e^t = \sum_{t' \le t} x_e^{t'}$ be the cumulative mass of element $e$ in the first $t$ rounds.
Let $y_e = y_e^T$ be its total mass in the instance for brevity.

\begin{definition}[$p$-Multi-way Semi-OCS]
    \label{def:mocs}
    A multi-way online selection algorithm is a $p$-multi-way semi-OCS for a non-increasing function $p : [0, +\infty) \to [0, 1]$ if for any multi-way online selection instance and any element $e$, $e$ is unselected with probability at most $p(y_e)$.
\end{definition}

%

\subsection{Algorithm: Weighted Sampling without Replacement}


We consider weighted sampling without replacement, which is parameterized by a weight function $w : [0, +\infty) \to [1, +\infty)$ with $w(0) = 1$.
In each round $t$, the sampling weight of an element $e \in \eset^t$ equals $0$ if the element has already been selected in the previous rounds, and equals $x_e^t w(y_e^{t-1})$ otherwise.
See Algorithm~\ref{alg:multiway}. 

\begin{algorithm}[t]
    \caption{Multi-way Semi-OCS: Weighted Sampling without Replacement}
    \label{alg:multiway}
    \begin{algorithmic}
        \medskip
        \State \textbf{Parameters:}\\
        \smallskip
        \quad Non-decreasing weight function $w : [0, +\infty) \to [1, +\infty)$ such that $w(0) = 1$.\\
        \quad Our result lets $w(y) = \exp \Big( y + \frac{y^2}{2} + c y^3 \Big)$ where $c= \frac{4-2\sqrt{3}}{3}$.
        \medskip
        \State \textbf{State variables:} (for each element $e$)
        \begin{itemize}
            \item Cumulative mass $y_e^t$ of element $e$ up to any round $t$.
            \item Whether element $e$ has been selected in any previous rounds.
        \end{itemize}
        \smallskip
        \State \textbf{For each round $t$:}
        \begin{enumerate}
            \item If all elements in $\eset^t$ have been selected, select arbitrarily, e.g., uniformly at random.
            \item Otherwise, select an unselected $e \in \eset^t$ with probability proportional to $x^t_e \cdot w(y^{t-1}_e)$.
        \end{enumerate}
    \end{algorithmic}
\end{algorithm}

We remark that the optimal ($2$-way) semi-OCS in Section~\ref{sec:semi-ocs} can be interpreted as the limit case when $w(y) = W^{y}$ and $W$ tends to infinity.

\subsection{Analysis}

\begin{theorem}
    \label{thm:multiway-ocs}
    Weighted Sampling without Replacement (Algorithm~\ref{alg:multiway}) with weight function:
    \begin{equation}
        \label{eqn:multiway-weight}
        w(y) = \exp \Big( y + \frac{y^2}{2} + c y^3 \Big)
    \end{equation}
    where $c = \frac{4-2\sqrt{3}}{3} \approx 0.179$ is a $p$-multi-way semi-OCS for:
    \[
        p(y) = \frac{1}{w(y)} = \exp \Big( - y - \frac{y^2}{2} - c y^3 \Big)
        ~.
    \]
\end{theorem}

Consider an overly idealized algorithm which selects each element $e$ in round $t$ with probability exactly $x_e^t$ and never selects any element more than once.
It would be a $p^*$-multiway semi-OCS for $p^*(y) = \max\{1-y, 0\}$.
By the Taylor series of $\log(1-y)$ for $0 \le y < 1$, it can be written as:
\[
    p^*(y) = \exp \Big( - \sum_{i=1}^\infty \frac{y^i}{i} \Big)
    ~.
\]

The guarantee of Theorem~\ref{thm:multiway-ocs} matches the overly idealized bound up to the quadratic term and has a smaller coefficient for the cubic term.

First, we prove some properties about the weight function $w$ in Eqn.~\eqref{eqn:multiway-weight}.

\begin{lemma}\label{lemma:multiway-cubic}
    For any $0 \le x < 1$ and any $y \ge 0$:
    \[
        \frac{w(y+x)}{w(y)}\le  \frac{x }{1-x}w(y)+1
        ~.
    \]
\end{lemma}

The proof of Lemma~\ref{lemma:multiway-cubic} involves tedious calculations and computer-aided numerical verifications that are not insightful.
Hence, we defer it to Appendix~\ref{app:multiway-cubic};
see also Appendix~\ref{app:multiway-cubic-weak} for a proof that does not use computer-aided numerical verifications for a weaker version of the lemma.
We further introduce a generalized version of Lemma~\ref{lemma:multiway-cubic} whose proof is also deferred to Appendix~\ref{app:multiway-condition}.

\begin{lemma}\label{lemma:multiway-condition}
    For any $k \ge 1$, any $x_i, y_i \ge 0$ for $1 \le i \le k$ such that $\sum_{i=1}^k x_i \in [0, 1]$:
    \[
        \frac{1-\sum_{i=1}^k x_i}{\sum_{i=1}^k x_i w(y_i) +1- \sum_{i=1}^k x_i}\le \prod_{i=1}^k \frac{w(y_i)}{w(y_i+x_i)}
        ~.
    \]
\end{lemma}


With these two lemmas, we bound the unselected probability for any subset of elements, which implies Theorem~\ref{thm:multiway-ocs} as a special case.

\begin{theorem}
    \label{thm:multiway-ocs-strong}
    Weighted Sampling without Replacement (Algorithm~\ref{alg:multiway}) with weight function $$w(y)=\exp\left(y+ \frac{y^2}{2}+c y^3\right)$$ with $c=\frac{4-2\sqrt{3}}{3}$ ensures that any subset of elements $\eset' \subseteq \eset$ are unselected with probability at most:
    \[
        \prod_{e \in \eset'} p(y_e)
        ~,
    \]
    where $p(y) = \frac{1}{w(y)}$.
\end{theorem}

\begin{proof}
    Recall that $\uset^t$ denotes the set of unselected elements after round $t$.
    Hence, $\eset' \subseteq \uset^t$ is the event that the elements in $\eset'$ are unselected in the first $t$ rounds.
    We shall prove by induction on $0 \le t \le T$ that:
    \[
        \Pr \big[ \eset' \subseteq \uset^t \big] \le \prod_{e \in \eset'} p(y_e^t)
        ~,
    \]
    which implies Theorem~\ref{thm:multiway-ocs-strong} as a special case when $t = T$.

    The base case when $t = 0$ holds vacuously because both sides of the inequality equal $1$.

    Next suppose that it holds for $t-1$ rounds for some $t > 0$, and consider the case of $t$ rounds.
    Let $\bar{X}_e^t$ be the indicator of whether element $e$ is unselected after round $t$, and define $\bar{X}_{\eset'}^t = \prod_{e\in \eset'} \bar{X}_e^t$ for any $\eset' \subseteq \eset$.
    Finally, we write $\bar{X}^t$ for $( \bar{X}_e^t )_{e \in \eset}$.
    \begin{align*}
        \Pr \big[ \eset'\subseteq \uset^t \big]
        &
        = \E\,\bar{X}_{\eset'}^t \\
        &
        = \E_{\bar{X}^{t-1}} \left[ \bar{X}_{\eset'}^{t-1} \left( 1 - \frac{ \sum_{e\in \eset'} w(y_e^{t-1}) x_e^t \bar{X}_e^{t-1} }{ \sum_{e\in \eset} w(y_e^{t-1}) x_e^t \bar{X}_e^{t-1} } \right)  \right]
        ~.
    \end{align*}
    Here we artificually define $\frac{0}{0} = 0$ for ease of presentation.
    Readers may verify that our argument stays true with this caveat.

    Next, multiply $\bar{X}_{\eset'}^{t-1}$ with both the numerator and denominator in the above fraction.
    Using that $Y^2 = Y$ for $Y \in \{0, 1\}$, we have:
    \begin{align*}
        \Pr \big[ \eset'\subseteq \uset^t \big]
        &
        = \E_{\bar{X}^{t-1}} \left[ \bar{X}_{\eset'}^{t-1} \left( 1 - \frac{ \sum_{e\in \eset'} w(y_e^{t-1}) x_e^t \bar{X}_{\eset'}^{t-1} }{ \sum_{e\in \eset'} w(y_e^{t-1}) x_e^t \bar{X}_{\eset'}^{t-1} + \sum_{e\not\in \eset'} w(y_e^{t-1}) x_e^t \bar{X}_{\eset'\cup\{e\}}^{t-1} } \right) \right] \\
        &
        = \E_{X^{t-1}} \left[  \frac{ \bar{X}_{\eset'}^{t-1} \sum_{e\not\in \eset'} w(y_e^{t-1}) x_e^t \bar{X}_{\eset'\cup\{e\}}^{t-1} }{ \sum_{e\in \eset'} w(y_e^{t-1}) x_e^t \bar{X}_{\eset'}^{t-1} + \sum_{e\not\in \eset'} w(y_e^{t-1}) x_e^t \bar{X}_{\eset'\cup\{e\}}^{t-1} } \right]
        ~.
    \end{align*}
    
    By the concavity of $f(x,y) = \frac{xy}{x+y}$, it follows from Jensen's inequality that:
    \[
        \Pr \big[ \eset'\subseteq \uset^t \big]
        \le \frac{ \E\bar{X}_{\eset'}^{t-1} \sum_{e\not\in \eset'} w(y_e^{t-1}) x_e^t\, \E\bar{X}_{\eset'\cup\{e\}}^{t-1}}{\sum_{e\in \eset'} w(y_e^{t-1}) x_e^t\,\E\bar{X}_{\eset'}^{t-1} + \sum_{e\not\in \eset'} w(y_e^{t-1}) x_e^t\,\E\bar{X}_{\eset'\cup\{e\}}^{t-1} }
        ~.
    \]

    By the inductive hypothesis and the monotonicity of $f(x,y) = \frac{xy}{x+y}$, we further get that:
    \begin{align*}
        \Pr \big[ \eset'\subseteq \uset^t \big]
        &
        \le \frac{ \prod_{e \in \eset'} p(y_e^{t-1})  \sum_{e\not\in \eset'} w(y_e^{t-1}) x_e^t \prod_{e' \in \eset'\cup \{e\}} p(y_{e'}^{t-1})
        }{ \sum_{e\in \eset'} w(y_e^{t-1}) x_e^t \prod_{e \in \eset'} p(y_e^{t-1}) + \sum_{e\not\in \eset'} w(y_e^{t-1}) x_e^t \prod_{e' \in \eset'\cup \{e\}} p(y_{e'}^{t-1}) }
        \\
        &
        = \prod_{e \in \eset'} p(y_e^{t-1}) \frac{  \sum_{e\not\in \eset'} w(y_e^{t-1}) x_e^t p(y_e^{t-1})
        }{ \sum_{e\in \eset'} w(y_e^{t-1}) x_e^t + \sum_{e\not\in \eset'} w(y_e^{t-1}) x_e^t p(y_e^{t-1}) } \\
        &
        = \prod_{e \in \eset'} p(y_e^{t-1}) \frac{  \sum_{e\not\in \eset'} x_e^t
        }{ \sum_{e\in \eset'} w(y_e^{t-1}) x_e^t + \sum_{e\not\in \eset'} x_e^t  }
        ~.
    \end{align*}
    
    %

    Next combine the above with Lemmas \ref{lemma:multiway-cubic} and \ref{lemma:multiway-condition}:
    \begin{align*}
        \Pr \big[ \eset' \subseteq \uset^t \big]
        &
        \le
        \frac{1-\sum_{e \in \eset'} x^t_e}{\sum_{e\in \eset'}x^t_e w(y^{t-1}_e) +1- \sum_{e \in \eset'} x_e^t} \prod_{e \in \eset'} p(y_e^{t-1})
        ~
        \tag{$\sum_{e \in \eset} x_e^t = 1$}\\
        &
        \le \prod_{e\in \eset'} \frac{w(y_{e}^{t-1})}{w(y_{e}^{t-1}+x_e^t)}\prod_{e \in \eset'} p(y_e^{t-1})\tag{Lemmas \ref{lemma:multiway-cubic} and \ref{lemma:multiway-condition}}
        \\[1ex]
        &
        = \prod_{e\in \eset'} p(y_e^{t-1}+x_e^t).
        \tag{$p(y) = \frac{1}{w(y)}$}
    \end{align*}
\end{proof}

\section{Improved Algorithms and Hardness for OCS}
\label{sec:ocs}

\subsection{Definitions}

Recall that an online selection algorithm is a $\gamma$-OCS, if for any ($2$-way) online selection instance, any element $e$, and any disjoint consecutive subsequences of the rounds involving $e$ with lengths $k_1, k_2, \cdots, k_m$ respectively, the probability that $e$ is unselected in these rounds is at most:
\[
    \prod_{\ell=1}^m 2^{-k_{\ell}}(1-\gamma)^{k_{\ell}-1}
    ~.
\]

\begin{definition}[Ex-ante Dependence Graph, c.f., \citet{FahrbachHTZ:FOCS:2020}]
    The \emph{ex-ante dependence graph} $G^\exante = (V, E^\exante)$ is a directed graph defined with respect to an online selection instance.
    We shall refer to its vertices and edges as nodes and arcs to make a distinction with those in online matching problems. 
    The nodes correspond to rounds:
    \[
        V = \big\{ 1, 2, \dots, T \big\}
        ~.
    \]
    The arcs correspond to neighboring appearances of an element (indicated by the subscript%
    \footnote{
        There could be parallel arcs in the ex-ante dependence graph, e.g., when rounds $t$ and $t' = t+1$ have the same two elements.
        The subscript helps distinguish such parallel arcs.
    }%
    ):
    \[
        E^\exante = \Big\{ \big( t, t' \big)_e ~:~ t < t';~ e \in \eset^t;~ e \in \eset^{t'};~ \forall t < t'' < t', e \notin \eset^{t''} \Big\}
        ~.
    \]
\end{definition}

\begin{figure}
    \tikzset{%
        every neuron/.style={
            circle,
            draw,
            minimum size=0.5cm,
            very thick
        },
        neuron/.style={
            circle,
            minimum size=0.5cm
        },
    }
    \centering
    \begin{tikzpicture}[x=0.8cm, y=0.8cm, >=latex]
        \node [every neuron/.try, neuron 1/.try] (1) at (0,0) {1};
        \draw (0,-0.8) node {$\eset^1=\{a,c\}$};
        \node [every neuron/.try, neuron 1/.try] (2) at (3,0) {2};
        \draw (3,-0.8) node {$\eset^2=\{b,d\}$};
        \node [every neuron/.try, neuron 1/.try] (3) at (6,0) {3};
        \draw (6,-0.8) node {$\eset^3=\{a,b\}$};
        \node [every neuron/.try, neuron 1/.try] (4) at (9,0) {4};
        \draw (9,-0.8) node {$\eset^4=\{a,c\}$};
        \node [every neuron/.try, neuron 1/.try] (5) at (12,0) {5};
        \draw (12,-0.8) node {$\eset^5=\{b,c\}$};
        \node [every neuron/.try, neuron 1/.try] (6) at (15,0) {6};
        \draw (15,-0.8) node {$\eset^6=\{b,c\}$};
        \path[->, very thick, dashed] (2) edge (3);
        \path[->, very thick, dashed] (3) edge (4);
        \path[->, very thick, dashed] (4) edge (5);
        \path[->, very thick, dashed] (5) edge (6);
        \path[->, very thick, dashed] (1) edge[bend left] (4);
        \path[->, very thick, dashed] (1) edge[bend left] (3);
        \path[->, very thick, dashed] (3) edge[bend left] (5);
        \path[->, very thick, dashed] (5) edge[bend left] (6);
    \end{tikzpicture}
    \caption{Example of ex-ante dependence graph}
    \label{fig:exante-dependence-graph}
\end{figure}
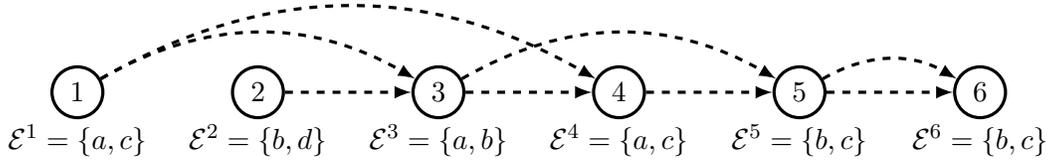

See Figure~\ref{fig:exante-dependence-graph} for an illustrative example of the ex-ante dependence graph. 

\subsection{Roadmap}

\subsubsection{Matching-based Approach versus Automata-based Approach on a Path}
\label{sec:ocs-path}

This subsection reviews the matching-based approach of \citet{FahrbachHTZ:FOCS:2020} and its limitation, and explains the automata-based approach in this paper.
As a running example, consider an instance with the same two elements \emph{head} ($\head$) and \emph{tail} ($\tail$) in every round, and thus the ex-ante dependence graph is a directed path (more precisely, two identical parallel directed paths).

\paragraph{Matching-based Approach.}
\citet{FahrbachHTZ:FOCS:2020} propose to select a matching from the ex-ante graph, and then to select elements in each pair of matched nodes and in each isolated node with independent random bits;
each pair of matched nodes shall select the opposite elements.
We shall select the matching such that 1) the selections of different arcs are negatively dependent (including independent), and 2) the probability of selecting each arc is as high as possible.
If we could select the arcs each with probability at least $\beta$ with negative dependence, we would obtain a $\beta$-OCS because of the following argument.
For any disjoint consecutive subsequences of lengths $k_1, k_2, \dots, k_m$, they contain $\sum_{i=1}^m (k_i-1)$ arcs that could have been selected into the matching.
If we select at least one of them into the matching, the opposite selections in its two nodes ensure selecting both elements.
By the aforementioned properties, the matching has none of these arcs with probability at most $(1-\beta)^{\sum_{i=1}^m (k_i-1)}$.
Even in that case, we still have $\sum_{i=1}^m k_i$ independent selections in these rounds;
the probability of not selecting a given element in them is at most $2^{-\sum_{i=1}^m k_i}$.

For example, \citet{FahrbachHTZ:FOCS:2020} let each node independently pick an incident arc, and then select an arc into the matching if both nodes pick it.
This selects each arc with probability $\frac{1}{4}$ in the special case when the ex-ante graph is a directed path.
It is possible to improve in the special case.
For instance, we could let each arc independently sample a number uniformly from $[0, 1]$ and select an arc if its number is bigger than its neighbors'.
This selects each arc with probability $\frac{1}{3}$.
To our best effort, however, we cannot find any matching-based algorithm that selects each arc with probability more than $\frac{\sqrt{5}-1}{2} \approx 0.382$.
Further, some of these ideas that improve the $\frac{1}{4}$ bound by the algorithm of \citet{FahrbachHTZ:FOCS:2020} fail to generalize beyond the special case.

\paragraph{Automata-based Approach.}
This paper introduces a different approach that selects elements using a probabilistic automaton.
We shall refer to both this automaton and its transition function as $\sigma^*$.
It has five states $q_\origin$, $q_\head$, $q_{\head^2}$, $q_\tail$, and $q_{\tail^2}$.
The original state $q_\origin$ is both the initial state of the automaton and the state it resets to after selecting the same element in two consecutive rounds.
State $q_\head$ (resp., $q_\tail$) means that the automaton selects $\head$ (resp., $\tail$) in the previous round but not twice in a roll;
from this state the automaton selects $\tail$ (resp., $\head$) with a higher chance, and the margin $\beta \in [0, 1]$ will be optimized to be $\beta = \sqrt{2}-1$ in our analysis.
State $q_{\head^2}$ (resp., $q_{\tail^2}$) means that the automaton selects $\head$ (resp, $\tail$) in the last two rounds;
from this state the automaton will select $\tail$ (resp., $\head$) with certainty and resets to the original state $\origin$.
Below is the transition function $\sigma^*$ that takes a state as input and returns a state and an element from $\big\{ \head, \tail \big \}$ (see also Figure~\ref{fig:tree-ocs-automaton}):
\begin{align*}
    \sigma^*(q_\origin) =
    \begin{cases}
        \big( q_\head, \head \big) & \text{w.p.\ $\frac{1}{2}$} \\
        \big( q_\tail, \tail \big) & \text{w.p.\ $\frac{1}{2}$}
    \end{cases}
    ~,
    &
    \quad
    \sigma^*(q_\head) =
    \begin{cases}
        \big( q_{\head^2}, \head \big) & \text{w.p.\ $\frac{1-\beta}{2}$} \\
        \big( q_\tail, \tail \big) & \text{w.p.\ $\frac{1+\beta}{2}$}
    \end{cases}
    ~,
    \quad
    \sigma^*(q_\tail) =
    \begin{cases}
        \big( q_\head, \head \big) & \text{w.p.\ $\frac{1+\beta}{2}$} \\
        \big( q_{\tail^2}, \tail \big) & \text{w.p.\ $\frac{1-\beta}{2}$}
    \end{cases}
    ~,
    \\[2ex]
    &
    \sigma^*(q_{\head^2}) = \big(q_\origin, \tail \big)
    ~,
    \qquad
    \sigma^*(q_{\tail^2}) = \big(q_\origin, \head \big)
    ~.
\end{align*}

\begin{figure}[t]
\begin{center}
\tikzset{%
  every neuron/.style={
    circle,
    draw,
    minimum size=1cm
    ,very thick
  },
  neuron/.style={
    circle,
    minimum size=1cm
  },
}

\begin{tikzpicture}[x=1.5cm, y=1.5cm, >=stealth]

\node [every neuron/.try, neuron 1/.try] (0) at (0,0) {$q_\origin$};
\node [every neuron/.try, neuron 2/.try] (1) at (2,2) {$q_\tail$};
\node [every neuron/.try, neuron 3/.try] (2) at (4,0) {$q_{\tail^2}$};
\node [every neuron/.try, neuron 4/.try] (-1) at (-2,2) {$q_\head$};

\node [every neuron/.try, neuron 5/.try] (-2) at (-4,0) {$q_{\head^2}$};

\draw[->,very thick] (0) -- (1)  ;
\draw[->,very thick] (0) -- (-1) ;
\draw[->,very thick] (1) -- (2) ;
\draw[->,very thick] (1) ..controls (1,3) and (-1,3).. (-1) ;
\draw[->,very thick] (-1) -- (-2) ;
\draw[->,very thick] (-1) -- (1) ;
\draw[->,very thick] (2) -- (0) ;
\draw[->,very thick] (-2) -- (0) ;

\node [rotate=45] (lab1) at (1.17,0.87) {$\tail,\frac{1}{2}$};
\node [rotate=-45] (lab1) at (-1.17,0.87) {$\head,\frac{1}{2}$};
\node [rotate=-45] (lab1) at (3.15,1.15) {$\tail,\frac{1-\beta}{2}$};
\node [rotate=45] (lab1) at (-3.15,1.15) {$\head,\frac{1-\beta}{2}$};
\node [] (lab1) at (2.1,-0.2) {$\head,1$};
\node [] (lab1) at (-2.1,-0.2) {$\tail,1$};
\node [] (lab1) at (0,2.2) {$\tail,\frac{1+\beta}{2}$};
\node [] (lab1) at (0,3.08) {$\head,\frac{1+\beta}{2}$};
\end{tikzpicture}
\end{center}
\caption{The probabilistic automaton $\sigma^*$ that selects an element in each round in the special case whose ex-ante dependence graph is a directed path, and in our $\beta$-tree OCS. The transitions are labeled by the selections from $\big\{ \head, \tail \big\}$, and by the probabilities of transitions.}
\label{fig:tree-ocs-automaton}
\end{figure}
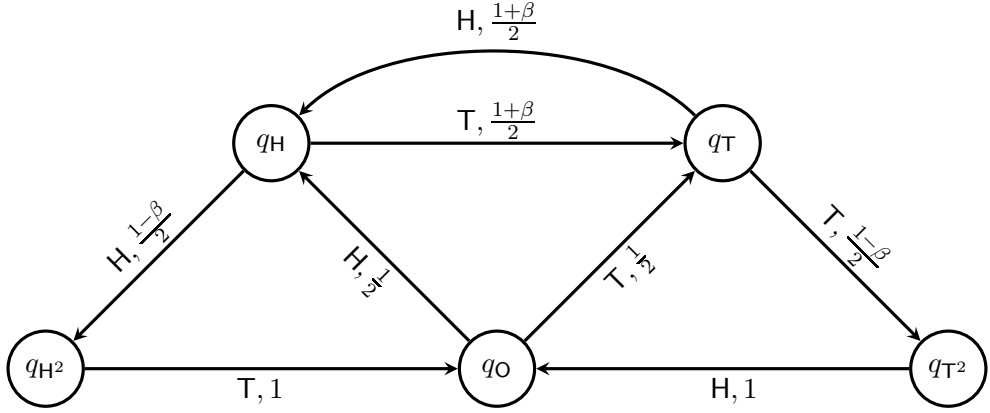

We find that using this automaton to select elements in different rounds is a $(\sqrt{2}-1)$-OCS in the special case.
Readers will find the proof of a stronger claim in Subsection~\ref{sec:forest-ocs}.
This is strictly better than our best effort using the matching-based approach.
More importantly, it generalizes to arbitrary online selection instances using the techniques in the rest of the section.

\subsubsection{Automata-based Approach}

This subsection outlines how to generalize the automata-based approach to general online selection instances and obtain an improvement over the $0.109$-OCS of \citet{FahrbachHTZ:FOCS:2020}.

\begin{theorem}
    \label{thm:ocs}
    There is a polynomial-time $0.167$-OCS for the $2$-way online selection problem.
\end{theorem}

We next explain the ingredients and how to combine them to prove Theorem~\ref{thm:ocs}.
The sequel subsections will substantiate them, with the proofs of some lemmas deferred to Appendix~\ref{app:ocs}.

A main challenge in generalizing the automata-based approach to general instances is deciding from which in-neighbor each node shall inherit the state of the automata.
In other words, we need to select an in-arc for each node to form a directed binary forest.%
\footnote{That is, each node has at most one in-arc from its parent, and at most two out-arcs to its children. The latter is true for the ex-ante dependence graph itself, and therefore also for all its subgraphs.}
We find that the naïve approach of independently and randomly selecting an in-arc for each node does not work unless the instance satisfies additional properties (see Subsection~\ref{sec:forest-constructor-good-instances}), because we need the directed binary forest to satisfy another property defined below.

\begin{definition}[Good Forest]
    A \emph{good forest} $G^\forest = (V, E^\forest)$ with respect to an online selection instance is a subgraph of the ex-ante dependence graph $G^\exante = (V, E^\exante)$ such that:
    \begin{enumerate}
        \item $G^\forest = (V, E^\forest)$ is a directed binary forest;
        \item For any node $p$ with two children $c$ and $c'$ in $G^\forest$, the corresponding rounds have no common element, i.e., $\eset^p \cap \eset^c \cap \eset^{c'} = \emptyset$.
    \end{enumerate}
\end{definition}

In the following definitions, for any subset of nodes $U \subseteq V$, let $E_U^\forest$ denote the subset of arcs induced by $U$ in the forest $G^\forest$.
Further for any element $e$ and any subset of nodes $U \subseteq V$ involving element $e$, let $E_{U, e}^\exante$ denote the subset of arcs induced by $U$ and with subscript $e$:
\[
    E_{U, e}^\exante \defeq \Big\{ (t, t')_e \in E^\exante ~:~ t \in U ;~ t' \in U \Big\}
    ~.
\]

\begin{definition}[Forest Constructor]
    A \emph{forest constructor} takes an online selection instance as input and returns good forest $G^\forest = (V, E^\forest)$.
    On receiving the elements $\eset^t$ of round $t$, it immediately decides whether each in-arc of $t$ belongs to $E^\forest$.
    It is an $\alpha$-forest constructor if for any element $e$, any subset of nodes $U \subseteq V$ involving $e$, and any $\beta \in [0, 1]$:
    \begin{equation}
        \label{eqn:alpha-forest-constructor}
        \E \big(1-\beta\big)^{|E^\forest_U|} \le \big(1-\alpha\beta\big)^{|E^\exante_{U, e}|}
        ~.
    \end{equation}
    The expectation is over the randomness of the forest constructor.
\end{definition}

The next lemma is our main result regarding forest constructors.
Subsection~\ref{sec:forest-constructor} presents the algorithm that proves this lemma.

\begin{lemma}
    \label{lem:forest-constructor}
    There is a polynomial-time $0.404$-forest constructor.
\end{lemma}

\begin{definition}[Forest OCS]
    A \emph{forest OCS} takes both an online selection instance and a good forest $G^\forest = (V, E^\forest)$ as input.
    At each round $t$, it observes the elements $\eset^t$ in the round and whether each in-arc of $t$ is in $E^\forest$, and then selects an element from $\eset^t$.
    It is a $\beta$-forest OCS if for any element $e$ and any subset of nodes $U \subseteq V$ involving $e$, the probability that $e$ is never selected in the corresponding rounds is at most:
    \[
        2^{-|U|} \big(1-\beta\big)^{|E^\forest_U|}
        ~.
    \]
\end{definition}

Our main result regarding forest OCS is the next lemma, whose proof is in Subsection~\ref{sec:forest-ocs}.

\begin{lemma}
    \label{lem:forest-ocs}
    There is a polynomial-time $(\sqrt{2}-1)$-forest OCS.
\end{lemma}

The next lemma combines the two ingredients to get an OCS, and implies Theorem~\ref{thm:ocs} as a corollary using Lemmas~\ref{lem:forest-constructor} and \ref{lem:forest-ocs}.

\begin{lemma}
    Suppose that there is a polynomial-time $\alpha$-forest constructor and a polynomial-time $\beta$-forest OCS.
    Together they form a polynomial-time $\alpha \beta$-OCS.
\end{lemma}

\begin{proof}
    The OCS combines the $\alpha$-forest constructor and the $\beta$-forest OCS as follows.
    On receiving the elements $\eset^t$ of a round $t$, it calls the forest constructor to determine whether each in-arc of $t$ is in $E^\forest$.
    Then, it puts this information together with the elements $\eset^t$ and calls the forest OCS to select an element from $\eset^t$.

    For any element $e$, and any disjoint consecutive subsequences of the rounds involving $e$, let $k_1, k_2, \dots, k_m$ be the lengths of these subsequences, and let $U$ be the subset of nodes that correspond to these rounds.
    By the guarantee of the $\alpha$-forest constructor and the $\beta$-forest OCS, the probability that element $e$ is never selected in these rounds is at most:
    \[
        \E \Big[ 2^{-|U|} \big( 1-\beta \big)^{|E^\forest_U|} \Big] \le 2^{-|U|} \big( 1-\alpha\beta \big)^{|E^\exante_{U,e}|}
        ~.
    \]

    The lemma then follows by $|U| = \sum_{\ell=1}^m k_\ell$ and $|E^\exante_{U,e}| = \sum_{\ell=1}^m \big( k_\ell-1 \big)$.
\end{proof}

\subsection{Forest Constructor}
\label{sec:forest-constructor}

\subsubsection{Warm-up: Good Online Selection Instances}
\label{sec:forest-constructor-good-instances}

We say that an online selection instance is good if its ex-ante graph satisfies the second requirement of good forests. That is, for any node $p$ with two out-neighbors $c$ and $c'$ in $G^\exante$, the corresponding rounds have no common element, i.e., $\eset^p \cap \eset^c \cap \eset^{c'} = \emptyset$.
Such good instances admit a simple forest constructor that for each node keeps one of its in-arcs independently and uniformly at random (Algorithm~\ref{alg:forest-constructor-good-instance}).
The simple forest constructor and its analysis are instructive, and motivate the forest constructor for general instances, so we include them as a warm-up.

\begin{algorithm}[t]
    \caption{$\frac{1}{2}$-Forest constructor for good online selection instances}
    \label{alg:forest-constructor-good-instance}
    \begin{algorithmic}
        \smallskip
        \State \textbf{For each round $t$:} (suppose that $\eset^t=\{e_1,e_2\}$)
        \begin{enumerate}
            \item For $i \in \{1, 2\}$, let $t_i$ be the most recent round that involves $e_i$ (if exists).
            \item Draw $j \in \{1, 2\}$ uniformly at random, and include arc $(t_j, t)_{e_j}$ into $E^\forest$ (if $t_j$ is defined).
        \end{enumerate}
    \end{algorithmic}
\end{algorithm}

\begin{lemma}
    \label{lem:forest-constructor-warmup}
    Algorithm~\ref{alg:forest-constructor-good-instance} is a $\frac{1}{2}$-forest constructor for good online selection instances.
\end{lemma}

\begin{proof}
    For any arc $a \in E^\exante_{U,e}$, let $X_a \in \{0, 1\}$ be the indicator of whether arc $a$ is included into $E^\forest$.
    Since every arc $a \in E^\exante_{U,e}$ with $X_a = 1$ belongs to $E^\forest_U$, we have:
    \[
        \big(1-\beta\big)^{|E^\forest_U|} \le \big(1-\beta\big)^{\sum_{a \in E^\exante_{U,e}} X_a}
        ~.
    \]

    Further, since $E^\exante_{U,e}$ by definition is a subset of arcs connecting neighboring appearances of element $e$, they have distinct in- and out-nodes.
    Hence, by the definition of Algorithm~\ref{alg:forest-constructor-good-instance}, $X_a$'s are independently and uniformly distributed over $\{0, 1\}$ for all arcs $a \in E^\exante_{U,e}$.
    We get that
    \[
        \E \big(1-\beta\big)^{|E^\forest_U|}
        \le \prod_{a \in E^\exante_{U,e}} \E \big(1-\beta\big)^{X_a}
        = \prod_{a \in E^\exante_{U,e}} \Big(1-\frac{\beta}{2} \Big)
        = \Big(1-\frac{\beta}{2} \Big)^{|E^\exante_{U,e}|}
        ~.
    \]
\end{proof}

For good instances, the second property of good forests always holds regardless of which arcs the forest constructor selects.
To satisfy the first property, i.e., to form a directed binary forest, consider a partition of arcs according to their destinations, into groups with one or two arcs each.
Then, selecting a directed binary forest is equivalent to selecting at most one arc from each group.
The above forest constructor indeed independently and randomly selects an arc from each group.
Finally, the independent selections of arcs in $E^\exante_{U,e}$ help show that it is a $\frac{1}{2}$-forest constructor.

\subsubsection{General Online Selection Instances, Pseudo-paths, and Pseudo-matchings}

\begin{figure}
    \centering
    \tikzset{%
        every neuron/.style={
            circle,
            draw,
            minimum size=0.4cm,
            very thick
        },
        neuron/.style={
            circle,
            minimum size=0.4cm
        },
    }
    \begin{subfigure}[b]{0.9\textwidth}
        \centering
        \begin{tikzpicture}[x=0.8cm, y=0.8cm, >=latex]
            \node [every neuron/.try, neuron 1/.try] (11) at (0,1.5) {$3$};
            \node [every neuron/.try, neuron 1/.try] (12) at (2,0.5) {$4$};
            \node [every neuron/.try, neuron 1/.try] (13) at (4,1.5) {$1$};
            \node [every neuron/.try, neuron 1/.try] (14) at (6,0.5) {$3$};
            \node [every neuron/.try, neuron 1/.try] (15) at (8,1.5) {$2$};
            \node [every neuron/.try, neuron 1/.try] (21) at (10,1.5) {$3$};
            \node [every neuron/.try, neuron 1/.try] (22) at (12,0.5) {$5$};
            \node [every neuron/.try, neuron 1/.try] (23) at (14,1.5) {$4$};
            \node [every neuron/.try, neuron 1/.try] (31) at (16,1.9) {$5$};
            \node [every neuron/.try, neuron 1/.try] (32) at (16,0.1) {$6$};
            \path[->, line width=2.2pt] (11) edge (12);
            \path[->, very thick, dashed] (13) edge (12);
            \path[->, very thick, dashed] (13) edge (14);
            \path[->, line width=2.2pt] (15) edge (14);
            \path[->, line width=2.2pt] (21) edge (22);
            \path[->, very thick, dashed] (23) edge (22);
            \path[->, line width=2.2pt] (31) edge[bend left] (32);
            \path[->, very thick, dashed] (31) edge[bend right] (32);
        \end{tikzpicture}
        \caption{Partition of arcs into pseudo-paths, and bolded selected arcs as a pseudo-matching}
        \label{fig:ocs-forest-constructor-pseudo-path}
    \end{subfigure}
    \begin{subfigure}[b]{0.9\textwidth}
        \centering
        \begin{tikzpicture}[x=0.8cm, y=0.8cm, >=latex]
            \node [every neuron/.try, neuron 1/.try] (1) at (0,0) {1};
            \draw (0,-0.8) node {$\eset^1=\{a,c\}$};
            \node [every neuron/.try, neuron 1/.try] (2) at (3,0) {2};
            \draw (3,-0.8) node {$\eset^2=\{b,d\}$};
            \node [every neuron/.try, neuron 1/.try] (3) at (6,0) {3};
            \draw (6,-0.8) node {$\eset^3=\{a,b\}$};
            \node [every neuron/.try, neuron 1/.try] (4) at (9,0) {4};
            \draw (9,-0.8) node {$\eset^4=\{a,c\}$};
            \node [every neuron/.try, neuron 1/.try] (5) at (12,0) {5};
            \draw (12,-0.8) node {$\eset^5=\{b,c\}$};
            \node [every neuron/.try, neuron 1/.try] (6) at (15,0) {6};
            \draw (15,-0.8) node {$\eset^6=\{b,c\}$};
            \path[->, line width=2.2pt] (2) edge (3);
            \path[->, line width=2.2pt] (3) edge (4);
            \path[->, very thick, dashed] (4) edge (5);
            \path[->, line width=2.2pt] (5) edge (6);
            \path[->, very thick, dashed] (1) edge[bend left] (4);
            \path[->, very thick, dashed] (1) edge[bend left] (3);
            \path[->, line width=2.2pt] (3) edge[bend left] (5);
            \path[->, very thick, dashed] (5) edge[bend left] (6);
        \end{tikzpicture}
        \caption{Selected arcs (bolded) form a good forest in ex-ante dependence graph}
        \label{ocs:possible-restricted-dependence-forest}
    \end{subfigure}
    \caption{An illustrative example of the forest constructor for general instances}
    \label{ocs:arc-picking-dependence-graph}
\end{figure}
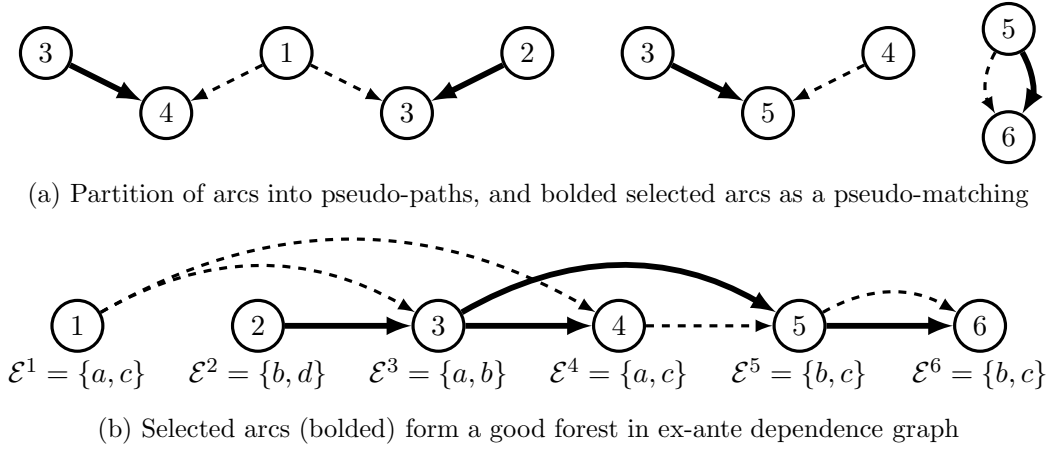

For general instances, a forest constructor needs to also ensure the second property of good forests.
In other words, for any pairs of arcs $a$ and $a'$ with the same origin such that the rounds corresponding to their incident nodes share a common element, a forest constructor must not select both $a$ and $a'$ into the forest.
For example, consider an instance with the ex-ante dependence graph in Figure~\ref{fig:exante-dependence-graph}.
The directed forest property requires that, e.g., arcs $(1,3)_a$ and $(2,3)_b$ cannot be both selected, arcs $(1,4)_c$ and $(3,4)_a$ cannot be both selected, etc., due to having the same destinations.
The second property of good forests further requires that, e.g., arcs $(1,3)_a$ and $(1,4)_c$ cannot be both selected.
Dropping their directions, the above arcs $(3,4)_a, (1,4)_c, (1,3)_a, (2,3)_b$ form an undirected path $3-4-1-3-2$.
More importantly, we can succinctly describe the aforementioned requirements of good forests as not selecting neighboring arcs with respect to the path.
Driven by this observation, we define pseudo-paths and the pseudo-matchings below.

\begin{definition}[Pseudo-path]
    Given any online selection instance and its ex-ante dependence graph, a \emph{pseudo-path} is a maximal ordered subset of arcs $P = \big( (t_i, t_i')_{e_i} \big)_{1 \le i \le \ell}$ such that for any $1 \le i < \ell$:
    \begin{itemize}
        \item Either $t_i' = t_{i+1}'$, i.e., the $i$-th and $(i+1)$-th arcs in $P$ have the same destination;
        \item Or $t_i = t_{i+1}$, i.e., the $i$-th and $(i+1)$-th arcs in $P$ have the same origin, and rounds $t_i = t_{i+1}$, $t_i'$, and $t_{i+1}'$ have a common element.
    \end{itemize}
\end{definition}


\begin{lemma}
    \label{lem:pseudo-path}
    The pseudo-paths partition the arcs of the ex-ante graph.%
    \footnote{We consider two pseudo-paths with the same subset of arcs but in opposite orders as the same pseudo-path.}
\end{lemma}

Figure~\ref{fig:ocs-forest-constructor-pseudo-path} shows the partition of arcs into pseudo-paths in the aforementioned example whose ex-ante dependence graph is Figure~\ref{fig:exante-dependence-graph}.
We defer other structural properties about pseudo-paths to sequel subsections where we use them to design and analyze the forest constructor.

\begin{definition}[Pseudo-matching]
    For any pseudo-path $P = \big( (t_i, t_i')_{e_i} \big)_{1 \le i \le \ell}$, a subset of its arcs $M$ is a \emph{pseudo-matching} if it has no adjacent arcs with respect to $P$, i.e., for any $1 \le i < \ell$, either $(t_i, t_i')_{e_i} \notin M$ or $(t_{i+1}, t_{i+1}')_{e_{i+1}} \notin M$.
\end{definition}

We remark that a pseudo-matching may not be a matching of the ex-ante dependence graph.
For example, arcs $(3,4)_a$ and $(2,3)_b$ form a pseudo-matching of the left-most pseudo-path in Figure~\ref{fig:ocs-forest-constructor-pseudo-path} even though they share node $3$.

\begin{lemma}
    \label{lem:pseudo-matching}
    A subgraph of the ex-ante dependence graph is a good forest if and only if it is a union of pseudo-matchings, one for each pseudo-path.
\end{lemma}

Therefore, a forest constructor needs to pick a pseudo-matching from each pseudo-path.
Further, it must do so in an online fashion.
On observing the elements $\eset^t$ of round $t$, it either appends $t$'s in-arcs to an existing pseudo-path and lets them start a new pseudo-path on their own, according to the definition of pseudo-paths.
It also immediately decides if to include each in-arc into the pseudo-matching.
To make it an $\alpha$-forest constructor for the largest possible $\alpha$, we want to select as many arcs into the pseudo-matchings as possible, and at the same time to keep the selections sufficiently independent so that an analysis similar to Lemma~\ref{lem:forest-constructor-warmup} applies.
The latter refutes selecting either all odd arcs or all even arcs from each pseudo-path with equal probability.

\subsubsection{Forest Constructor for General Instances}

\begin{figure}
\centering
\tikzset{%
  every neuron/.style={
    circle,
    draw,
    minimum size=1cm
    ,very thick
  },
  neuron/.style={
    circle,
    minimum size=1cm
  },
}
\begin{subfigure}[b]{.49\textwidth}
\centering
\begin{tikzpicture}[x=1.5cm, y=1.5cm, >=stealth]
\node [every neuron/.try, neuron 1/.try] (u) at (-1,0) {$\unmatched$};
\node [every neuron/.try, neuron 2/.try] (r) at (0,1.732) {$\ready$};
\node [every neuron/.try, neuron 3/.try] (m) at (1,0) {$\matched$};
\path[->, very thick] (m) [bend left] edge (u);
\path[->, very thick] (r) edge (m);
\path[->, very thick] (u) edge (m);
\path[->, very thick] (u) edge (r);
\node [] (lab1) at (0,-0.55) {$\no,1$};
\node [rotate=-60] (lab1) at (0.7,1) {$\yes,1$};
\node [] (lab1) at (0,0.2) {$\yes,p$};
\node [rotate=60] (lab1) at (-0.7,1) {$\no,1\!-\!p$};
\end{tikzpicture}
\caption{Automaton $\sigma^+$ for the positive end}
\end{subfigure}
\begin{subfigure}[b]{.49\textwidth}
\centering
\begin{tikzpicture}[x=1.5cm, y=1.5cm, >=stealth]
\node [every neuron/.try, neuron 1/.try] (u) at (-1,0) {$\unmatched$};
\node [every neuron/.try, neuron 2/.try] (r) at (0,1.732) {$\ready$};
\node [every neuron/.try, neuron 3/.try] (m) at (1,0) {$\matched$};
\path[->, very thick] (m) edge (u);
\path[->, very thick] (m) edge (r);
\path[->, very thick] (u) [bend right] edge (m);
\path[->, very thick] (r) edge (u);
\node [] (lab1) at (0,-0.55) {\small$\no,1$};
\node [rotate=-60] (lab1) at (0.7,1) {\small$\yes,1\!-\!p$};
\node [] (lab1) at (0,0.2) {\small$\yes,p$};
\node [rotate=60] (lab1) at (-0.7,1) {\small$\no,1$};
\end{tikzpicture}
\caption{Automaton $\sigma^-$ for the negative end}
\end{subfigure}
    \caption{
        The probabilistic automata in our forest constructor (Algorithm~\ref{alg:forest-constructor}).
        The transitions are labeled by the binary decisions and the probabilities of the transitions.
    }
    \label{fig:ocs-forest-constructor-automata}
\end{figure}
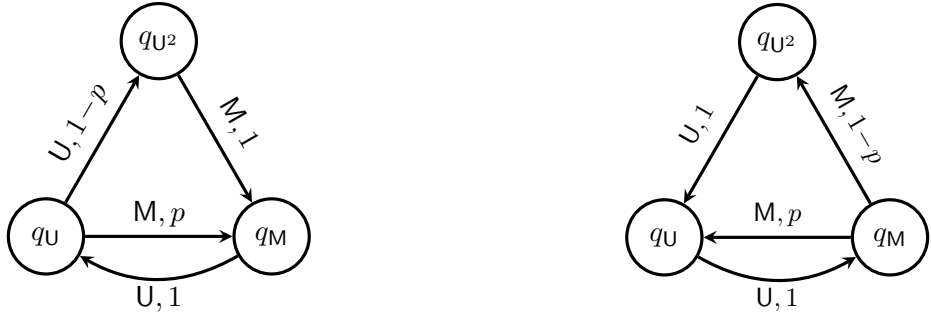

To explain our forest constructor for general instances, we need a structural lemma about the arrival order of arcs in any pseudo-path.
We notice that arcs usually arrive in pairs since the in-arcs of $t$ both arrive in round $t$.
We will artificially break ties to be consistent with the lemma.

\begin{lemma}
    \label{lem:pseudo-path-arrival}
    For any pseudo-path $P = \big( (t_i, t_i')_{e_i} \big)_{1 \le i \le \ell}$ and any $0 \le t \le T$, the subset of arcs that arrive in the first $t$ rounds is a sub-pseudo-path, i.e., either it is an empty set, or there exists $1 \le i_{\min} \le i_{\max} \le \ell$ such that the arrived arcs are $\big( (t_i, t_i')_{e_i} \big)_{i_{\min} \le i \le i_{\max}}$.
\end{lemma}

That is, after the first arc of a pseudo-path arrives, the arrival of any arc of the pseudo-path appends to the existing sub-pseudo-path; we never need to merge two pseudo-paths.
Our argument lets $i_0$ denote the index of the earliest arc.
We say that the arcs with indices $i \ge i_0$ are on the positive end, and those with indices $i < i_0$ are on the negative end.%
\footnote{
    The algorithm does not need to know the index $i_0$ upfront.
    Instead, it could let the earliest arc have index $0$;
    arcs on the positive ends have indices $0, 1, 2$, etc., and those on the negative end have indices $-1, -2, -3$, etc.
    Nonetheless, the choice of indices in the main text admits cleaner notations in the analysis.
}
Finally, changing the roles of the positive and negative ends will not affect our conclusion, so we will without loss of generality let the second earliest arc of any pseudo-path be on the positive end.




The forest constructor uses a probabilistic automaton $\sigma^+$ and its inverse $\sigma^-$.
The automata have states $\unmatched$, $\ready$, and $\matched$.
Intuitively, state $\unmatched$ means that automaton $\sigma^+$ leaves the last arc unmatched, but matches the arc before that;
state $\ready$ means that automaton $\sigma^+$ leaves the last two arcs unmatched;
and state $\matched$ means that automaton $\sigma^+$ matches the last arc.
The transition functions, which we denote also as $\sigma^+$ and $\sigma^-$ abusing notations, take a state as input and returns the next state and also a binary decision.
They are parameterized by $p \in [0, 1]$, the transition probability from $\unmatched$ to $\matched$ in automaton $\sigma^+$.
We will let $p = 0.6616$ in the analysis to optimize the result.
Formally, the transition functions are (see also Figure~\ref{fig:ocs-forest-constructor-automata}):
\begin{align*}
    \sigma^+ \big( \unmatched \big)
    &
    =
    \begin{cases}
        (\unmatched, \yes) & \text{w.p.\ $p$} \\
        (\ready, \no) & \text{w.p.\ $1-p$}
    \end{cases}
    ~,
    &
    \sigma^+ \big( \ready \big)
    &
    = (\matched, \yes)
    ~,
    &
    \sigma^+ \big( \matched \big)
    &
    = (\unmatched, \no)
    ~;
    &
    \\
    \sigma^- \big( \unmatched \big)
    &
    = (\matched, \no)
    ~,
    &
    \sigma^- \big( \ready \big)
    &
    = (\unmatched, \no)
    ~,
    &
    \sigma^- \big( \matched \big)
    &
    =
    \begin{cases}
        (\unmatched, \yes) & \text{w.p.\ $p$} \\
        (\ready, \yes) & \text{w.p.\ $1-p$}
    \end{cases}
    ~.
\end{align*}

For each pseudo-path, our forest constructor draws an initial state from the common stationary distribution of the automata.
Then, when an arc arrives on the positive end, it calls $\sigma^+$ to update the state and to decides whether to include the arc into the pseudo-matching;
similarly, when an arc arrives on the negative end, it calls $\sigma^-$.
See Algorithm~\ref{alg:forest-constructor}.



\begin{algorithm}[t]
    \caption{$0.404$-Forest constructor for general instances (when $p = 0.6616$)}
    \label{alg:forest-constructor}
    \begin{algorithmic}
        %
        \smallskip
        \State \textbf{State variables:} (for each pseudo-path $P$)    
        \begin{itemize}
            \item $\state^+, \state^- \in \{ \unmatched, \ready, \matched \}$ of automata $\sigma^+, \sigma^-$ respectively.
            \item Initialize (when the first arc in the pseudo-path arrives):
                \[
                    \state^+ = \state^-= \begin{cases}
                        \unmatched & \text{w.p.\ $\frac{1}{3-p}$;} \\[1ex]
                        \ready & \text{w.p.\ $\frac{1-p}{3-p}$;} \\[1ex]
                        \matched & \text{w.p.\ $\frac{1}{3-p}$.}
                    \end{cases}
                \]
        \end{itemize}
        \State \textbf{For each arc:} (of pseudo-path $P$)
            \begin{enumerate}
                \item Let $\tau = +$ if the arc is on the positive end, and $-$ otherwise.
                \item Let $(\state^\tau, \choice) = \sigma^\tau(state^\tau)$.
                \item Include the arc into the pseudo-matching and thus the forest if $\choice = \yes$.
            \end{enumerate}
    \end{algorithmic}
\end{algorithm}


\subsubsection{Properties of the Automata}

\begin{lemma}
    \label{lem:forest-constructor-automata-stationary}
    The stationary distribution of the states of $\sigma^+$ and $\sigma^-$ is:
    \[
        \vec{\pi} = \big( \pi_\no, \pi_{\no^2}, \pi_\yes \big) = \Big( \frac{1}{3-p}, \frac{1-p}{3-p}, \frac{1}{3-p} \Big)
        ~.
    \]
\end{lemma}

The next lemma formalizes the claim that $\sigma^-$ is the inverse of $\sigma^+$.


\begin{lemma}
    \label{lem:forest-constructor-automata-inverse}
    Consider two sequences of random variables:
    \begin{enumerate}
        \item Sample $\state^0$ from the stationary distribution $\vec{\pi}$.
            Then recursively let:
            \[
                \hspace{8pt}
                \big( \state^i, \choice^i \big) = \sigma^+ \big( \state^{i-1} \big) ~,
                \hspace{92pt}
                1 \le i \le \ell
                ~.
            \]
        \item Sample $\hat{\state}^{i_0-1}$ from the stationary distribution $\vec{\pi}$.
            Then recursively let:
            \begin{align*}
                \big( \hat{\state}^i, \hat{\choice}^i \big) & = \sigma^+ \big( \hat{\state}^{i-1} \big)
                ~,
                &&
                i_0 \le i \le \ell
                ~;
                \\
                \big( \hat{\state}^{i-1}, \hat{\choice}^i \big) & = \sigma^- \big( \hat{\state}^i \big)
                ~,
                &&
                1 \le i < i_0
                ~.
            \end{align*}
    \end{enumerate}
    The two sequences are identically distributed.
    %
    %
\end{lemma}

The second sequence corresponds to the states and decisions of automata $\sigma^+, \sigma^-$ in Algorithm~\ref{alg:forest-constructor}:
$\hat{q}^{i_0-1}$ is the common initial state of $\sigma^+$ and $\sigma^-$;
$\hat{q}^i$ and $d^i$ for $i \ge i_0$ are the states and decisions on the positive end;
$\hat{q}^i$ and $d^i$ for $i < i_0$ are the states and decisions on the negative end.
Lemma~\ref{lem:forest-constructor-automata-inverse} allows us to analyze each pseudo-path as if the arcs' arrival order is from one end to the other.

Next we adopt the viewpoint of selecting arcs with only automaton $\sigma^+$, and develop several properties of the corresponding sequence $\big( \state^0, \choice^1, \state^1, \dots, \choice^\ell, \state^\ell \big)$.
We start with three lemmas that follow by the definition of $\sigma^+$.

\begin{lemma}
    \label{lem:forest-constructor-automata-matched-state}
    For any $1 \le i \le \ell$, $\state^i = \matched$ if and only if $\choice^i = \yes$.
\end{lemma}

\begin{lemma}
    \label{lem:forest-constructor-automata-marginal}
    For any $1 \le i \le \ell$:
    \[
        \Pr \big[ \choice^i = \yes \big] = \frac{1}{3-p}
        ~.
    \]
\end{lemma}

\begin{lemma}
    \label{lem:forest-constructor-automata-matched-in-three}
    For any $1 \le i \le \ell-2$, at least one of $\choice^i$, $\choice^{i+1}$, and $\choice^{i+2}$ equals $\yes$.
\end{lemma}

Lemma~\ref{lem:forest-constructor-automata-matched-state} asserts that every time $\sigma^+$ selects an arc into the pseudo-matching (and thus the forest), it resets to state $\matched$.
Combining with Lemma~\ref{lem:forest-constructor-automata-stationary}, we get the marginal selection probability in Lemma~\ref{lem:forest-constructor-automata-marginal}.
Lemma~\ref{lem:forest-constructor-automata-matched-in-three} further claims that it resets to state $\matched$ at least once every three rounds. 
Hence, we focus on how the probabilistic automaton transitions starting from state $\matched$, and in particular how likely the $i$-th arc after that would be selected,
which in turns characterizes the transition from the other two states.
These probabilities are characterized by a recurrence:
\begin{equation}
    \label{eqn:forest-constructor-automata-selection-prob}
    f_i =
    \begin{cases}
        1 & i=0 ~; \\
        0 & i=1 ~; \\
        p & i=2 ~; \\
        p f_{i-2} + (1-p) f_{i-3} & i \ge 3 ~.
    \end{cases}
\end{equation}

\begin{lemma}
    \label{lem:forest-constructor-automata-selection-prob}
    For any $i \le j$:
    \[
        \Pr \big[ \choice^j = \yes \mid \state^i = \matched \big]
        =
        \Pr \big[ \choice^j = \yes \mid \choice^{i} = \yes \big]
        =
        f_{j-i}
        ~.
    \]
    Further, for any $i \le j-1$:
    \begin{align*}
        \Pr \big[ \choice^j = \yes \mid \state^i = \ready \big]
        &
        =
        f_{j-i-1}
        ~,
        \\
        \Pr \big[ \choice^j = \yes \mid \state^i = \unmatched \big]
        &
        =
        f_{j-i+1}
        ~.
    \end{align*}
\end{lemma}


\begin{lemma}
    \label{lem:forest-constructor-automata-selection-prob-bound}
    Suppose that $\frac{\sqrt{5}-1}{2}\leq p\leq \frac{2}{3}$.
    Then:
    \begin{align*}
        f_i & \ge 1-p ~,
        && \forall i \ge 2 ~; \\
        f_i & \ge p^3+\big(1-p\big)^2
        && \forall i \ge 4 ~.
    \end{align*}
    %
\end{lemma}

\subsubsection{Analysis of Forest Constructor for General Instances: Proof of Lemma~\ref{lem:forest-constructor}}

This subsection proves Lemma~\ref{lem:forest-constructor} by showing that Algorithm~\ref{alg:forest-constructor} with $p = 0.6616$ is a $0.404$-forest constructor.
Below summarizes some properties that either follow by the definition of Algorithm~\ref{alg:forest-constructor}, or have been established in the previous subsections:
\begin{itemize}
    \item Constructing a good forest is the same as selecting a pseudo-matching from each pseudo-path.
        \hspace*{\fill}
        (Lemma~\ref{lem:pseudo-matching})
    \item The selections of arcs in different pseudo-paths are independent.
        \hspace*{\fill}
        (Definition of Algorithm~\ref{alg:forest-constructor})
    \item The selections of arcs on a pseudo-path is equivalent to sampling a state from the stationary distribution $\vec{\pi}$, and applying $\sigma^+$ to decide for each arc from one end to the other.
        \hspace*{\fill}
        (Lemma~\ref{lem:forest-constructor-automata-inverse})
    \item Automaton $\sigma^+$ selects an arc and resets to state $\matched$ at least once every three rounds.\\
        \hspace*{\fill}
        (Lemmas~\ref{lem:forest-constructor-automata-matched-state} and \ref{lem:forest-constructor-automata-matched-in-three})
    \item The probability of selecting the $i$-th arc after any state is characterized by $f_i$, which can be lower bounded.
        \hspace*{\fill}
        (Lemmas~\ref{lem:forest-constructor-automata-selection-prob} and \ref{lem:forest-constructor-automata-selection-prob-bound})
\end{itemize}

\paragraph{Proving that Algorithm~\ref{alg:forest-constructor} Constructs a Good Forest.}
Consider an arbitrary pseudo-path.
By the definition of $\sigma^+$, it resets its state to $\matched$ when the decision is $\yes$ (Lemma~\ref{lem:pseudo-matching}), i.e., when an arc is selected, after which the next decision will be $\no$.
Therefore, the arcs selected from each pseudo-path are a pseudo-matching.
By Lemma~\ref{lem:pseudo-matching} this is a good forest.

\paragraph{Proof of Equation~\eqref{eqn:alpha-forest-constructor}.}
To show the guarantee of a $\alpha$-forest constructor, which we restate below:
\[
    \forall 0 \le \beta \le 1 ~: \qquad \E \big(1-\beta\big)^{|E^\forest_U|} \le \big(1-\alpha\beta\big)^{|E^\exante_{U, e}|}
    ~,
    \tag{Eqn.~\eqref{eqn:alpha-forest-constructor} restated}
\]
it suffices to consider each pseudo-path $P$ separately and to show that:
\begin{equation}
    \label{eqn:alpha-forest-constructor-by-path}
    \forall 0 \le \beta \le 1 ~: \qquad \E \big(1-\beta\big)^{|E^\forest_U \cap P|} \le \big(1-\alpha\beta\big)^{|E^\exante_{U, e} \cap P|}
    ~,
    \hspace{65pt}
\end{equation}
after which Eqn.~\eqref{eqn:alpha-forest-constructor} follows by taking the product of Eqn.~\eqref{eqn:alpha-forest-constructor-by-path} over all pseudo-paths, and by the independence of arc selections in different pseudo-paths.

The rest of the argument considers an arbitrary pseudo-path $P$ and proves Eqn.~\eqref{eqn:alpha-forest-constructor-by-path}.
We start by establishing the last structural lemma about pseudo-paths, characterizing the subset of arcs that could contribute to the inequality by being counted in $E^\forest_U$.


\begin{lemma}
    \label{lem:pseudo-path-induced}
    For any element $e$, and any subset of nodes $U \subseteq V$ involving $e$, there is a subset of arcs in $P$ with both nodes inside $U$ such that:
    \begin{enumerate}
        \item It is a superset of $E^{\text{ex-ante}}_{U,e}\cap P$;
        \item It is the union of odd-length sub-pseudo-paths;
        \item Any two of these sub-pseudo-paths are at least $3$ arcs apart; and
        \item Each sub-pseudo-path alternates between arcs with subscript $e$, i.e., arcs that also contribute to the right-hand-side of Eqn.~\eqref{eqn:alpha-forest-constructor-by-path}, and arcs with other subscripts, i.e., arcs that only contribute to the left-hand-side.
            The arcs on the two ends have subscript $e$.
    \end{enumerate}
\end{lemma}

Given Lemma~\ref{lem:pseudo-path-induced}, we may assume that the subset of arcs in $P$ with both nodes involving $e$ are sub-pseudo-paths starting from arc indices $i_1, i_2, \dots, i_m$ and lengths $2k_1+1, 2k_2+1, \dots, 2k_m+1$.
Let $\indexset_j$ denote the set of indices of the $j$-th sub-pseudo-path: 
\[
    \indexset_j = \big\{ i_j, i_j+1, \dots, i_j+2k_j \big\}
    ~.
\]

Let $\indexset = \cup_{j=1}^m \indexset_j$ denote the set of arc indices in these pseudo-paths.
For each arc $i \in \indexset$, consider the indicator of if the $i$-th arc on the pseudo-path is selected into the pseudo-matching:
\[
    X_i \defeq \mathbf{1} \big( \choice^i = \yes \big)
    ~.
\]

Since any arc in $\indexset$ have both nodes in $U$ by definition, we have:
\[
    |E^\forest_U \cap P| \ge \sum_{i \in \indexset} X_i
    ~.
\]

Further, there are $k_j+1$ arcs with subscripts $e$ in the $j$-th sub-pseudo-path (i.e., contributing to $E^\exante_{U,e} \cap P$) by Lemma~\ref{lem:pseudo-path-induced}.
Hence, to prove Eqn.~\eqref{eqn:alpha-forest-constructor-by-path} it is sufficient to show:
\[
    \E \big(1-\beta\big)^{\sum_{i \in \indexset} X_i} \le \big(1-\alpha\beta\big)^{\sum_{j=1}^m (k_j+1)}
    ~.
\]

We next argue that it suffices to consider the case when all sub-pseudo-paths have unit lengths,
because we can reduce the general case to it.
Suppose that the $j$-th pseudo-path has length $2k_j+1 > 1$.
When $2k_j+1 = 3$ or $2k_j+1 \ge 7$, by Lemma~\ref{lem:forest-constructor-automata-matched-in-three} we have:
\[
    \sum_{i \in \indexset_j} X_i \ge \Big\lfloor \frac{2k_j+1}{3} \Big\rfloor \ge \frac{k_j+1}{2}
    ~.
\]

Hence, regardless of the realization of randomness in the forest constructor (Algorithm~\ref{alg:forest-constructor}), for $\alpha = 0.404$ and for any $\beta \in [0, 1]$ we always have:
\[
    \big(1-\beta\big)^{\sum_{i \in \indexset_j} X_i} \le \big(1-\beta\big)^{\frac{k_j+1}{2}} \le \Big( 1 - \frac{\beta}{2} \Big)^{k_j+1} \le \big( 1 - \alpha \beta \big)^{k_j+1}
    ~.
\]

In other words, we can without loss of generality remove the $j$-th sub-pseudo-path and prove Eqn.~\eqref{eqn:alpha-forest-constructor-by-path} for the remaining instance.

When $2 k_j+1 = 5$, i.e., $k_j = 2$, also by Lemma~\ref{lem:forest-constructor-automata-matched-in-three} we have:
\[
    \sum_{i \in \indexset_j, i \ne i_j} X_i \ge 1
    ~.
\]

Hence, regardless of the realization of randomness in the forest constructor (Algorithm~\ref{alg:forest-constructor}), for $\alpha = 0.404$ and for any $\beta \in [0, 1]$ we always have:
\[
    \big(1-\beta\big)^{\sum_{i \in \indexset_j, i \ne i_j} X_i} \le \big(1-\beta\big) \le \Big( 1 - \frac{\beta}{2} \Big)^2 \le \big( 1 - \alpha \beta \big)^2
    ~.
\]

That is, we can without loss of generality remove the arcs \emph{other than $i_j$} from $j$-th sub-pseudo-path and prove Eqn.~\eqref{eqn:alpha-forest-constructor-by-path} for the remaining instance.

Finally, consider Eqn.~\eqref{eqn:alpha-forest-constructor-by-path} when all sub-pseudo-paths have unit lengths.
In other words, for a subset of indices $\indexset = \{ i_1, i_2, \dots, i_m \}$ such that any two indices differ by at least $4$ (Lemma~\ref{lem:pseudo-path-induced}), we shall prove that:
\[
    \forall 0 \le \beta \le 1 ~: \qquad \E \big(1-\beta\big)^{\sum_{i \in \indexset} X_i} \le \big(1-\alpha\beta\big)^m
    ~.
\]

Our proof is an induction on $m$.
The base case when $m = 1$ follows by:
\begin{align*}
    \E \big(1-\beta\big)^{X_{i_1}}
    &
    \le \frac{1}{3-p} \big(1-\beta\big) + \frac{2-p}{3-p}
    \tag{Lemma~\ref{lem:forest-constructor-automata-marginal}}
    \\
    &
    = 1 - \frac{1}{3-p} \beta
    \\
    &
    \le 1 - \alpha \beta
    ~.
    \tag{$\alpha=0.404$, $p=0.6616$}
\end{align*}

Suppose that the inequality holds for up to $m-1$ indices.
We next prove it for $m$ indices.
Suppose without loss of generality that $i_1 < i_2 < \dots < i_m$.
By the inductive hypothesis:
\[
    \E \big(1-\beta\big)^{\sum_{i \in \indexset, i \ne i_m} X_i} \le \big(1-\alpha\beta\big)^{m-1}
    ~.
\]

It suffices to prove that \emph{for any realized $X_{i_1}, X_{i_2}, \dots, X_{i_{m-1}}$}:
\[
    \E \big[ X_{i_m} \mid X_{i_1}, \dots, X_{i_{m-1}} \big] \ge \alpha
    ~,
\]
as it would imply that:
\[
    \E \Big[ \big(1-\beta\big)^{X_{i_m}} \mid X_{i_1}, \dots, X_{i_{m-1}} = 1 \Big]
    \le \big(1-\beta\big) \alpha + \big(1-\alpha\big)
    = 1-\alpha\beta
    ~.
\]

For any realization such that $X_{i_{m-1}} = 1$: 
\begin{align*}
    \E \big[X_{i_m} \mid X_{i_1}, \dots, X_{i_{m-1}} = 1 \big]
    &
    = f_{i_m-i_{m-1}}
    \tag{Lemmas~\ref{lem:forest-constructor-automata-matched-state}, \ref{lem:forest-constructor-automata-selection-prob}}
    \\[.5ex]
    &
    \ge p^3+(1-p)^2
    \tag{$i_m-i_{m-1} \ge 4$ and Lemma~\ref{lem:forest-constructor-automata-selection-prob-bound}}
    \\[1ex]
    &
    \ge \alpha
    ~.
    \tag{$\alpha = 0.404$, $p = 0.6616$}
\end{align*}

For a realization such that $X_{i_{m-1}} = 0$, the state after processing arc $i_{m-1}$ could be $\unmatched$ or $\ready$.
We first argue that it is the former at least a $\frac{1}{1+p}$ fraction of the time.

\begin{lemma}
    \label{lem:ocs-unmatched-ready-ratio}
    For any realization such that $X_{i_{m-1}} = 0$:
    \[
        \Pr \Big[ \state^{i_{m-1}} = \unmatched \mid X_{i_1}, \dots, X_{i_{m-1}} = 0 \Big] \ge \frac{1}{1+p}
        ~.
    \]
\end{lemma}

Given Lemma~\ref{lem:ocs-unmatched-ready-ratio}, we can lower bound $\E \big[X_{i_m} \mid X_{i_1}, \dots, X_{i_{m-1}} = 0 \big]$ by:
\[
    \frac{1}{1+p} \E \big[X_{i_m} \mid \state^{i_{m-1}} = \unmatched \big] + \frac{p}{1+p} \min \Big\{ \E \big[X_{i_m} \mid \state^{i_{m-1}} = \unmatched \big], \E \big[X_{i_m} \mid \state^{i_{m-1}} = \ready \big] \Big\}
    ~.
\]

Suppose that $i_m - i_{m-1} = 4$.
By Lemma~\ref{lem:forest-constructor-automata-selection-prob-bound}:
\begin{align*}
    \E \big[X_{i_m} \mid \state^{i_{m-1}} = \unmatched \big]
    &
    = f_5
    = 2p(1-p)
    ~;
    \\[1ex]
    \E \big[X_{i_m} \mid \state^{i_{m-1}} = \ready \big]
    &
    = f_3
    = 1-p
    ~.
\end{align*}

Observe that $2p(1-p) > 1-p$ for $p = 0.6616$.
We get that:
\begin{align*}
    \E \big[X_{i_m} \mid X_{i_1}, \dots, X_{i_{m-1}} = 0 \big]
    &
    \ge \frac{1}{1+p} \cdot 2p(1-p) + \frac{p}{1+p} \cdot (1-p)
    \\[.5ex]
    &
    = \frac{3p(1-p)}{1+p}
    \ge \alpha
    ~.
    \tag{$\alpha = 0.404$, $p = 0.6616$}
\end{align*}

Otherwise, we have that $i_m - i_{m-1} \ge 5$.
By Lemma~\ref{lem:forest-constructor-automata-selection-prob-bound}:
\begin{align*}
    \E \big[X_{i_m} \mid \state^{i_{m-1}} = \unmatched \big]
    &
    = f_{i_m - i_{m-1} + 1}
    \ge p^3+\big(1-p\big)^2
    ~;
    \\[1ex]
    \E \big[X_{i_m} \mid \state^{i_{m-1}} = \ready \big]
    &
    = f_{i_m - i_{m-1} - 1}
    \ge p^3+\big(1-p\big)^2
    ~.
\end{align*}

Hence:
%
\[
    \E \big[X_{i_m} \mid X_{i_1}, \dots, X_{i_{m-1}} = 0 \big]
    \ge p^3+\big(1-p\big)^2
    \ge \alpha
    ~.
    \tag{$\alpha = 0.404$, $p = 0.6616$}
\]

\subsection{Forest OCS}
\label{sec:forest-ocs}

\subsubsection{Algorithm}

\begin{algorithm}[t]
    \caption{$(\sqrt{2}-1)$-Forest OCS}
    \label{alg:forest-ocs}
    \begin{algorithmic}
        \smallskip
        \State \textbf{Input:}
        \begin{itemize}
            \item An online selection instance represented by its ex-ante graph $G^\exante = (V, E^\exante)$.
            \item A good forest $G^\forest = (V, E^\forest)$.
        \end{itemize}
        \State \textbf{State variables:}
        \begin{itemize}
            %
            \item Label $\elabel_t(e) \in \big\{ \head, \tail \big\}$ for every round $1 \le t \le T$ and every element $e \in \eset^t$.
            \item State $q^t$ of automaton $\sigma^*$ for every node $1 \le t \le T$.
            %
        \end{itemize}
        \smallskip
        \State \textbf{For each round $t$:} (suppose that $\eset^t=\{e_1, e_2\}$)
        \begin{enumerate}
            \item If $t$ is the root of a directed binary tree in $G^\forest$, let the labels be $\ell_t(e_1) = \head$, $\ell_t(e_2) = \tail$.
            \label{step:forest-ocs-label-root}
            \item Otherwise, suppose without loss of generality that $(t', t)_{e_1}$ is the in-arc of node $t$ in $G^\forest$, let $\elabel_t(e_1) = \elabel_{t'}(e_1)$ and let $\elabel_t(e_2)$ be the other label from $\big\{ \head, \tail \big\}$.
            \label{step:forest-ocs-label-other}
        \item Let $(q^t, \ell) = \sigma^*(q^{t'})$ (artificially let $q^{t'} = q_\origin$ if $t$ is the root of a directed binary tree).
            \item Select the element with label $\ell$.
        \end{enumerate}
    \end{algorithmic}
\end{algorithm}

On observing the elements $\eset^t$ of each round $t$, the forest OCS labels the elements by head ($\head$) and tail ($\tail$).
Then, it calls automaton $\sigma^*$ from Subsection~\ref{sec:ocs-path} with the state of $t$'s parent in the good forest as input to select a label and to get the state of $t$.
Finally, it selects the element whose label is selected by automaton $\sigma^*$.
See Algorithm~\ref{alg:forest-ocs}.


\subsubsection{Analysis: Proof of Lemma~\ref{lem:forest-ocs}}

The lemmas in this subsection assume $\beta = \sqrt{2}-1$, which we shall not restate repeatedly.
We first establish a structural lemma about the subset of arcs with both nodes involving an element $e$ in any good forest, and about the labels in Algorithm~\ref{alg:forest-ocs}.
Recall that $E^\forest_U$ denotes the subset of arcs in forest $G^\forest$ with both nodes in $U$.


\begin{lemma}
    \label{lem:forest-ocs-structural}
    For any good forest $G^\forest = (V, E^\forest)$, any element $e$, and any subset of nodes $U \subseteq V$ involving $e$, $E^\forest_U$ consists of a collection of tree-paths such that:
    \begin{enumerate}
        \item There is no arc between any two nodes in distinct tree-paths; and
        \item Element $e$ has the same label in each tree-path.
    \end{enumerate}
\end{lemma}

We next prove Lemma~\ref{lem:forest-ocs} by an induction on the number of tree-paths in $E^\forest_U$.
The base case with zero tree-path holds vacuously.

Next for some $m \ge 1$ suppose that the lemma holds with at most $m-1$ tree-paths.
Consider an arbitrary instance for which $E^\forest_U$ consists of $m$ tree-paths satisfying the properties of Lemma~\ref{lem:forest-ocs-structural}.
Let $t_1, t_2, \dots, t_m$ denote the first node of these tree-paths.
Let $k_1, k_2, \dots, k_m$ be the lengths.
We assume without loss of generality that $t_m$'s height in its directed binary tree is greater than or equal to the height of any other $t_i$ from the same tree.

\paragraph{Case 1:}
Suppose that $t_m$ is the root of a directed binary tree in $G^\forest$.
First by the inductive hypothesis on the sub-instance that removes the tree rooted at $t_m$, the probability of not selecting element $e$ in these tree-paths is at most $2^{-\sum_{i=1}^{m-1} k_i} (1-\beta)^{\sum_{i=1}^{m-1} (k_i-1)}$.
Then, \emph{conditioned on any realized randomness on the other tree-paths},
the probability of never selecting $e$ on the tree-path starting with $t_m$ is at most $2^{-k_m} (1-\beta)^{k_m-1}$, because of the second part of Lemma~\ref{lem:forest-ocs-structural} and a special case of the next lemma when $i = 1$.

\begin{lemma}
    \label{lem:forest-ocs-from-origin}
    For any $k$ consecutive positive integers $i, i+1, \dots, i+k-1$, and any label $\ell \in \{ \head,\tail \}$, the probability that automaton $\sigma^*$, starting from the original state $q_\origin$, does not select label $\ell$ in its $i$-th to $(i+k-1)$-th selections is at most $2^{-k}(1-\beta)^{k-1}$.
\end{lemma}

\paragraph{Case 2:}
Suppose that $t_m$ is not the root of any directed binary tree in $G^\forest$, but its sibling in $G^\forest$ (if any) is not one of $t_1, t_2, \dots, t_{m-1}$.
By the latter assumption, the fact that $t_m$'s parent is not on the other tree-paths (Lemma~\ref{lem:forest-ocs-structural}, first part), and the assumption that $t_m$ has the largest height compared to other $t_i$ from the same directed binary tree, we get that the nodes rooted from $t_m$'s parent are not on the other tree-paths.
Then, by the inductive hypothesis on the sub-instance that removes nodes rooted from \emph{$t_m$'s parent}, the probability of not selecting element $e$ in these tree-paths is at most $2^{-\sum_{i=1}^{m-1} k_i} (1-\beta)^{\sum_{i=1}^{m-1} (k_i-1)}$.
Then, \emph{conditioned on any realized randomness on the other tree-paths},
the probability of never selecting $e$ on the tree-path starting with $t_m$ is at most $2^{-k_m} (1-\beta)^{k_m-1}$, because of the second part of Lemma~\ref{lem:forest-ocs-structural} and the next lemma.

\begin{lemma}
    \label{lem:forest-ocs-from-other}
    For any $k$ consecutive positive integers $i, i+1, \dots, i+k-1$ starting from $i \ge 2$, and any label $\ell \in \{ \head,\tail \}$, the probability that automaton $\sigma^*$, starting from an arbitrary state, does not select label $\ell$ in its $i$-th to $(i+k-1)$-th selections is at most $2^{-k}(1-\beta)^{k-1}$.
\end{lemma}

\paragraph{Case 3:}
Suppose that $t_m$ is not the root of any directed binary tree in $G^\forest$, and further its sibling in $G^\forest$ is $t_j$, the starting node of another tree-path in Lemma~\ref{lem:forest-ocs-structural}.
Since $t_m$'s parent is not on the other tree-paths (Lemma~\ref{lem:forest-ocs-structural}, first part), and further by the assumption that $t_m$ has the largest height compared to other $t_i$ from the same directed binary tree, the nodes rooted from $t_m$ and $t_j$'s parent are not on the other tree-paths.
By the inductive hypothesis on the sub-instance that removes the nodes rooted from \emph{$t_m$'s parent}, the probability of not selecting element $e$ in these tree-paths is at most $2^{-\sum_{i \ne j, m} k_i} (1-\beta)^{\sum_{i \ne j, m} (k_i-1)}$.
Then, \emph{conditioned on any realized randomness on the other tree-paths},
the probability of never selecting $e$ on the tree-paths starting with $t_j$ and $t_m$ is at most $2^{-k_j-k_m} (1-\beta)^{(k_j-1)+(k_m-1)}$, because of the second part of Lemma~\ref{lem:forest-ocs-structural}, the observation that $e$ must get different labels on these two tree-paths,%
\footnote{
    The parent $p$ of $t_j, t_m$ must not contain $e$ by the second requirement of good forest.
    Suppose that $p$'s common elements with $t_j, t_m$ are $e_j, e_m \ne e$ respectively.
    Then, $e$'s labels in $t_j, t_m$ are the labels of $e_m, e_j$ in $p$ respectively.
}
and the next lemma.

\begin{lemma}
    \label{lem:forest-ocs-fork}
    Consider two independent copies of automaton $\sigma^*$ from an arbitrary but identical initial state.
    Then, for any $k, \hat{k} \ge 1$, the probability that the first copy never selects $\tail$ in the first $k$ rounds, and the second copy never selects $\head$ in the first $\hat{k}$ rounds is at most $2^{-k-\hat{k}} (1-\beta)^{(k-1)+(\hat{k}-1)}$.
\end{lemma}

\subsection{Hardness}

\citet{FahrbachHTZ:FOCS:2020} rule out the possibility of $1$-OCS.
This subsection improves the upper bound to $\frac{1}{4}$.
This holds even for algorithms with unlimited computational power, and even if the algorithms know the instance beforehand.

\begin{theorem}
    \label{thm:ocs-hardness}
    There is no $(\frac{1}{4}+\epsilon)$-OCS for any constant $\epsilon > 0$.
\end{theorem}

\begin{proof}
    Consider any $\gamma$-OCS.
    Consider an online selection instance with three elements $\{0, 1, 2\}$ and $T = 2i+1$ rounds.
    The elements in odd rounds are $\{0, 1\}$;
    the elements in even rounds are $\{0, 2\}$.
    We shall prove that $\gamma \le \frac{i}{4i-1}$ even if we omit the properties of $\gamma$-OCS concerning three or more rounds.
    Theorem~\ref{thm:ocs-hardness} then follows by choosing a sufficiently large $i$.

    For any even $1 \le t \le T$ (so that the elements are $\{0, 2\}$), let $A^t$ be the event that the algorithm selects element $0$ in rounds $t$.
    We have:
    \[
        \forall 1 \le t \le T,~t \equiv 0 \bmod{2} ~: \qquad \Pr \big[ A^t \big] = \frac{1}{2}
        ~.
    \]

    For any $1 \le t < t' \le T$ with distinct parities (so that $0$ is the only common element), let $B^{t,t'}$ be the event that the algorithms does not select element $0$ in both rounds $t$ and $t'$.
    We have:
    \[
        \forall 1 \le t < t' \le T,~t + t' \equiv 1 \bmod{2} ~: \qquad
        \Pr \big[ B^{t, t'} \big] \le
        \begin{cases}
            \frac{1}{4} & t' \ge t+3 ~;\\[1ex]
            \frac{1-\gamma}{4} & t' = t+1 ~.
        \end{cases}
    \]

    For any $1 \le t < t' \le T$ with the same parity (so that they have the same two elements), let $C^{t,t'}$ be the event that the algorithms selects element $0$ in both rounds $t$ and $t'$ (i.e., it does not select the other element in both rounds).
    We have:
    \[
        \forall 1 \le t < t' \le T,~t + t' \equiv 0 \bmod{2} ~: \qquad
        \Pr \big[ C^{t, t'} \big] \le
        \begin{cases}
            \frac{1}{4} & t' \ge t+4 ~;\\[1ex]
            \frac{1-\gamma}{4} & t' = t+2 ~.
        \end{cases}
    \]

    Summing together:
    \begin{align}
        \sum_{t \equiv 0 \bmod{2}} \Pr \big[ A^t \big] & + \sum_{t < t' : t+t' \equiv 1 \bmod{2}} \Pr \big[ B^{t,t'} \big] + \sum_{t < t' : t+t' \equiv 1 \bmod{2}} \Pr \big[ C^{t,t'} \big]
        \notag
        \\[1ex]
        &
        \le \underbrace{\frac{1}{2} \cdot i}_\text{events $A^t$} + \underbrace{\frac{1}{4} \cdot i(i-1) + \frac{1-\gamma}{4} \cdot 2i}_\text{events $B^{t, t'}$} + \underbrace{\frac{1}{4} \cdot (i-1)^2 + \frac{1-\gamma}{4} \cdot (2i-1)}_\text{events $C^{t, t'}$}
        \notag
        \\[1ex]
        &
        = \frac{2i^2+3i}{4} - \frac{4i-1}{4} \gamma
        ~.
        \label{eqn:ocs-hardness-1}
    \end{align}

    Next consider any selections $(s^1, s^2, \dots, s^T)$ by the algorithm. 
    Let $j_o$ and $j_e$ be the numbers of odd and even rounds that select $s^t = 0$ respectively.
    The number of events that it satisfies equals:
    \begin{align*}
        \underbrace{\vphantom{\bigg|}j_e}_\text{events $A^t$}
        +
        \underbrace{\vphantom{\bigg|}(i+1-j_o)(i-j_e)}_\text{events $B^{t,t'}$}
        +
        \underbrace{\vphantom{\bigg|}\binom{j_o}{2} + \binom{j_e}{2}}_\text{events $C^{t,t'}$}
        &
        =
        \frac{1}{2} \big(j_o+j_e-i-\frac{1}{2}\big)^2+\frac{i(i+1)}{2}-\frac{1}{8} \\
        &
        \ge \frac{i(i+1)}{2}
        ~.
        \tag{$i, j_o, j_e$ are integers}
    \end{align*}

    As a result we get that:
    \begin{equation}
        \label{eqn:ocs-hardness-2}
        \sum_{t \equiv 0 \bmod{2}} \Pr \big[ A^t \big] + \sum_{t < t' : t+t' \equiv 1 \bmod{2}} \Pr \big[ B^{t,t'} \big] + \sum_{t < t' : t+t' \equiv 1 \bmod{2}} \Pr \big[ C^{t,t'} \big] \ge \frac{i(i+1)}{2}
        ~.
    \end{equation}

    Combining Equations~\eqref{eqn:ocs-hardness-1} and \eqref{eqn:ocs-hardness-2} gives $\gamma \le \frac{i}{4i-1}$ as desired.
\end{proof}

\section{Applications in Online Bipartite Matching}
\label{sec:matching}

\subsection{Online Bipartite Matching Preliminaries}

Consider an undirected bipartite graph $G = (L, R, E)$, where $L$ and $R$ are the sets of left-hand-side and right-hand-side vertices respectively, and $E$ is the set of edges.
Each edge $(u, v) \in E$ has a positive edge-weight $w_{uv} > 0$.
The problem is \emph{unweighted} if $w_{uv} = 1$ for all $(u, v) \in E$, is \emph{vertex-weighted} if $w_{uv} = w_u$ for some positive vertex-weights $(w_u)_{u \in L}$ of the left-hand-side vertices, and is \emph{edge-weighted} if the edge-weights could be arbitrary.

In online bipartite matching problems, we refer to the left-hand-side and right-hand-side vertices as \emph{offline} and \emph{online} vertices respectively.
Initially, the algorithm only knows the offline vertices, and the vertex-weights in the vertex-weighted case.
Then, the online vertices arrive one at a time.
When an online vertex $v \in R$ arrives, the algorithm sees its incident edges, and the edge-weights in the edge-weighted case.
The algorithm then immediately and irrevocably matches $v$ to an offline neighbor $u \in L$.

The objective is to maximize the sum of the maximal edge-weight matched to each offline vertex.
In the unweighted and vertex-weighted problems, matching an offline vertex more than once does not further increase the objective.
Therefore, we may assume without loss of generality that the algorithm matches each offline vertex at most once and the matched edges indeed form matching.
The objectives in these two cases are equivalent to maximizing the cardinality of the matching, and maximizing the sum of the vertex-weights of matched offline vertices, respectively.

In edge-weighted online bipartite matching, we may alternatively view the above objective as allowing disposals of previously matched edges so that a matched offline vertex could be rematched to a new edge with a larger edge-weight.
In other words, we may think of the matching as being comprised of the heaviest edge matched to each offline vertex, and seek to maximize the total edge-weight of the matching.
Further in online advertising, it corresponds to displaying an advertiser's ad multiple times but only charges for the most valuable one.
\citet{FeldmanKMMP:WINE:2009} introduce this \emph{free disposal} model which has then become the standard model of edge-weighted online bipartite matching under worst-case competitive analysis.

We compare the expected objective of the matching by the algorithm, and the optimal matching that maximizes the objective in hindsight given full information of the bipartite graph $G = (L, R, E)$ and the edge-weights $(w_{uv})_{(u, v) \in E}$.
The competitive ratio of an online algorithm is the infimum of this ratio over all possible instances.

\subsection{Semi-OCS and Unweighted and Vertex-weighted Online Bipartite Matching}
\label{sec:semi-ocs-unweighted-vertex-weighted-matching}

\citet{FahrbachHTZ:FOCS:2020} give a two-choice greedy algorithm for unweighted online bipartite matching, using a semi-OCS as a sub-routine.
Their original theorem is only for the unweighted problem and only for the guarantee of $\gamma$-semi OCS, i.e., $p(k) = 2^{-k}(1-\gamma)^{k-1}$.
Nonetheless, the algorithm and analysis generalize to the vertex-weighted case and for general $p(k)$ by standard techniques in the online matching literature.
We state the more general theorem below.

\begin{theorem}[c.f., \citet{FahrbachHTZ:FOCS:2020}]
    \label{thm:semi-ocs-unweighted-vertex-weighted-matching}
    Given a semi-OCS such that the probability of never selecting an element $e$ that appears $k$ times is at most $p(k)$, there is a\, $\Gamma$-competitive two-choice greedy algorithm for unweighted and vertex-weighted online bipartite matching, where the competitive ratio $\Gamma$ is the optimal value of the following linear program (LP):
    \begin{align}
        \text{\rm maximize} \quad
        &
        \Gamma
        \tag{\matchinglp}
        \\[2ex]
        \text{\rm subject to} \quad
        &
        a(k)+b(k)\leq p(k)-p(k+1) & \forall k\geq 0
        \label{eqn:two-way-lp-gain-split}
        \\
        &
        \sum_{i=0}^{k-1} a(i)+2 b(k)\geq \Gamma & \forall k\geq 0
        \label{eqn:approximate-dual-feasible}
        \\
        & b(k+1)\leq b(k) & \forall k\geq 0\label{eqn:two-way-lp-monotone}
        \\[2ex]
        & a(k),b(k)\geq 0 & \forall k \ge 0
        \notag
    \end{align}
\end{theorem}


We will not present the generalized algorithm and the proof of Theorem~\ref{thm:semi-ocs-unweighted-vertex-weighted-matching} because they will be subsumed by the algorithm and theorem in the next subsection.
Instead, the main result of this subsection is an explicit optimal solution to the LP.
By contrast, \citet{FahrbachHTZ:FOCS:2020} rely on solving a finite approximation of the LP numerically using LP solvers.


\begin{theorem}
    \label{thm:lp-solution}
    Suppose that $p(0) = 1$ and $p(k+1)\leq \frac{2}{3}p(k)$ for any $k \ge 0$. 
    Then, the \matchinglp admits an optimal solution as follows: 
    \begin{align*}
        \Gamma
        &
        = 1-\frac{1}{3} \sum_{i=0}^{\infty} \Big(\frac{2}{3}\Big)^i p(i)
        ~;
        \\
        b(k)
        &
        =
        \frac{1}{3} \sum_{i=k}^{\infty} \Big(\frac{2}{3}\Big)^{i-k} \big( p(i)-p(i+1) \big)
        &
        \forall k\ge 0
        ~;
        \\[1ex]
        a(k)
        &
        = p(k)-p(k+1)-b(k)
        &
        \forall k \ge 0
        ~.
    \end{align*}
\end{theorem}

The assumption of $p(k+1)\leq \frac{2}{3}p(k)$ is essentially without loss of generality since any natural online selection algorithm shall at least halve the unselected probability after each round involving the element.
Indeed, even the trivial independent sampling satisfies the stronger $p(k+1)\leq \frac{1}{2}p(k)$. The proof of this theorem is deferred to Appendix~\ref{app:lp-solution}.

For $\gamma$-semi-OCS, it recovers a result by \citet{HuangZZ:FOCS:2020} as a corollary.

\begin{corollary}[c.f., \citet{HuangZZ:FOCS:2020}]
    Suppose that $p(k) = 2^{-k} (1-\gamma)^{k-1}$.
    Then the optimal value of the \matchinglp is:
    \[
        \frac{3+2\gamma}{6+3\gamma}
        ~.
    \]
\end{corollary}

Since the optimal semi-OCS in Section~\ref{sec:semi-ocs} gives $p(k) = 2^{-2^k+1}$, we have the next corollary through a numerical calculation.

\begin{corollary}
    \label{cor:semi-ocs-unweighted-vertex-weighted-matching}
    The two-choice greedy algorithm using the optimal semi-OCS as a sub-routine is at least $0.536$-competitive for unweighted and vertex-weighted online bipartite matching.
\end{corollary}


\subsection{OCS and Edge-weighted Online Bipartite Matching}
\label{sec:ocs-edge-weighted-matching}

\subsubsection{Online Primal-Dual Algorithm}

This subsection gives a variant of the online primal-dual algorithm of \citet{FahrbachHTZ:FOCS:2020} for edge-weighted online bipartite matching, using an OCS as a sub-routine.
This variant simplifies the analysis in the next subsection.
To simplify exposition, we assume that for every online vertex $v$ there is a unique offline dummy vertex such that the edge between them has weight $0$.
Then, every online vertex will be matched, although being matched to the dummy vertex is the same as being left unmatched.

For each online vertex $v$, the algorithm shortlists two candidates $u_1, u_2$ from $v$'s neighbors.
If the shortlisted candidates are the same, the algorithm matches $v$ to it.
Otherwise, the algorithm lets the OCS selects one of them and matches $v$ to the selected one.
To explain how the algorithm makes the shortlists, let $k_u(w)$ be the number of times that $u$ is shortlisted thus far due to online vertices with edge-weight $w_{uv} \ge w$.
In a round in which $u_1 = u_2 = u$, the corresponding $k_u(w)$'s increase by $2$.
We remark that $k_u(w) = 0$ for any dummy offline vertex $u$ and for any $w > 0$.
The algorithm is parameterized by the optimal solution to the $\matchinglp$ in Theorem~\ref{thm:lp-solution}.
Given the optimal solution, define the ``value'' of matching an online vertex $v$ to an offline vertex $u$ as:
\begin{align}
    \label{eqn:matching-dual-update-beta}
    \Delta_u \beta_v
    \defeq
    \int_0^{w_{uv}} b\big(k_u(w)\big)dw -\frac{1}{2} \int_{w_{uv}}^{\infty} \sum_{i=0}^{k_u(w)-1} a(i) dw
    ~.
\end{align}

For each online vertex $v$, the algorithm first finds $u_1$ with the maximum $\Delta_u \beta_v$, and then finds $u_2$ with the maximum \emph{updated} $\Delta_u \beta_v$.
If $u_1 = u_2$, match $v$ to it.
Otherwise, match $v$ to the one that the OCS selects.
Following the terminology of \citet{FahrbachHTZ:FOCS:2020}, we call the former a deterministic round, and the latter a randomized round.
Their algorithm computes the ``values'' of deterministic and randomized rounds using different equations.
By contrast, our variant computes the ``values'' using the same Eqn.~\eqref{eqn:matching-dual-update-beta}.
See Algorithm~\ref{alg:ocs-matching}.

\begin{algorithm}[t]
    \caption{Online primal-dual edge-weighted bipartite matching algorithm}
    \label{alg:ocs-matching}
    \begin{algorithmic}
        \smallskip
        \State \textbf{State variables:} (for each offline vertex $u$)
        \begin{itemize}
            \item $k_u(w)$: the number of times $u_1 = u$ or $u_2 = u$ and further its edge weight is at least $w$.
        \end{itemize}
        \State \textbf{On the arrival of an online vertex $v \in R$:}
        \begin{enumerate}
            \item For $\ell \in \{1,2\}$:
            \begin{enumerate}
                \item Find $u_\ell$ with maximum $\Delta_u \beta_v$ given by Eqn.~\eqref{eqn:matching-dual-update-beta}.
                \item Increase $k_{u_\ell}(w)$ by 1 for $0 \le w \leq w_{u_\ell v}$.
            \end{enumerate}
            \item If $u_1 \neq u_2$, let the OCS select one of them, and match $v$ to it. \hfill \textbf{(Randomized round)}
            \item Otherwise, match $v$ to $u_1 = u_2$. \hfill \textbf{(Deterministic round)}
        \end{enumerate}
    \end{algorithmic}
\end{algorithm}

\subsubsection{Improved Online Primal-Dual Analysis}

This subsection improves the analysis of \citet{FahrbachHTZ:FOCS:2020} in twofold.
First, our edge-weighted result uses the LP in Theorem~\ref{thm:semi-ocs-unweighted-vertex-weighted-matching} and its optimal solution in Theorem~\ref{thm:lp-solution}, same as the unweighted and vertex-weighted cases.
By contrast, the analysis of \citet{FahrbachHTZ:FOCS:2020} for edge-weighted online bipartite matching needs to consider a LP with additional constraints.
Second, our analysis indicates that the online selection algorithm only needs to guarantee a condition strictly weaker than the property of $\gamma$-OCS.
It enables us to further explore a variant of OCS in the next subsection to further improve the competitive ratio in edge-weighted online bipartite matching.

\begin{theorem}
    \label{thm:two-way-online-primal-dual-analysis}
    Suppose that $( p(k) )_{k \ge 0}$ is non-increasing and satisfies $p(0) = 1$, and $\Gamma$, $( a(k) )_{k \ge 0}$, and $( b(k) )_{k \ge 0}$ form a solution to the \matchinglp.
    Algorithm~\ref{alg:ocs-matching} is $\Gamma$-competitive for edge-weighted online bipartite matching if the OCS ensures that for any online selection instance, any element $e$, and any consecutive subsequences of the rounds involving the element with lengths $k_1, k_2, \dots, k_m$, element $e$ is unselected in these rounds with probability at most:
    %
    \begin{equation}
        \label{eqn:relaxed-ocs}
        p \Big( \sum_{i=1}^m k_i \Big) +\frac{1}{2}\sum_{i=2}^m \sum_{j=0}^{k_1+\dots+k_{i-1}-1} a(j)
        ~.
    \end{equation}
\end{theorem}

We make three remarks before presenting the proof of the theorem.
First, for unweighted and vertex-weighted online bipartite matching, the online selection algorithm only needs to ensure the above property for the subset of all $k$ rounds involving an element.
Then, it degenerates to the guarantee of semi-OCS because $m = 1$ and thus the second term involving the $a(j)$'s disappears.
The proof below shall make this explicit.

Further, the guarantee in Eqn.~\eqref{eqn:relaxed-ocs} holds almost trivially for natural online selection algorithms when $m \ge 3$.
On the one hand, any natural algorithm would at least halve the unselected probability for every round involving the element.
Hence, after $\sum_{i=1}^m k_i \ge 3$ rounds, the unselected probability is at most $\frac{1}{8}$.
On the other hand, the optimal LP solution from Theorem~\ref{thm:lp-solution} satisfies that $a(0) \ge \frac{2}{9}$ for all online selection algorithms in the literature and in this paper, and even for the overly idealized algorithm that ensure selecting an element when it appears more than once.
Hence, the $a(0)$'s in the second term of Eqn.~\eqref{eqn:relaxed-ocs} sum to at least $\frac{2}{9} > \frac{1}{8}$.
A similar argument shows that the guarantee holds almost trivially for $m = 2$ if $k_1 + k_2 \ge 3$.
Hence, it suffices to slightly enhance the semi-OCS guarantee to further handle either a single consecutive subsequence (but not necessarily starting from the earliest round involving the element as in semi-OCS), or two very short consucutive subsequences.
This motivates the variant of OCS in the next subsection.

Finally, Theorem~\ref{thm:two-way-online-primal-dual-analysis} subsumes the analysis of \citet{FahrbachHTZ:FOCS:2020} because the original guarantee of $\gamma$-OCS satisfies Eqn.~\eqref{eqn:relaxed-ocs}, as we will prove in the next lemma.

\begin{lemma}
    Suppose that $\gamma \in [0, \frac{1}{4}]$,%
    \footnote{Theorem~\ref{thm:ocs-hardness} shows that there is no $\gamma$-OCS for $\gamma > \frac{1}{4}$.}
    and $p(k) = 2^{-k}(1-\gamma)^{k-1}$ for $k \ge 0$.
    Let $(a(k))_{k \ge 0}$ take values as in the optimal LP solution in Theorem~\ref{thm:lp-solution}.
    Then, for any positive integers $k_1, k_2, \dots, k_m$:
    \[
        \prod_{i=1}^m 2^{-k_i}(1-\gamma)^{k_i-1} \le 2^{-\sum_{i=1}^m k_i}(1-\gamma)^{\sum_{i=1}^m k_i-1} +\frac{1}{2}\sum_{i=2}^m \sum_{j=0}^{k_1+\dots+k_{i-1}-1} a(j)
        ~.
    \]
\end{lemma}

\begin{proof}
    In fact we will prove it even dropping all $a(j)$'s for $j \ge 1$.
    If $m = 1$ the left-hand-side equals the first term on the right-hand-side.
    If $m \ge 2$, the difference between the left-hand-side and the first term on the right-hand-side is:
    \begin{align*}
        \big( 1 - (1-\gamma)^{m-1} \big) \prod_{i=1}^m 2^{-k_i}(1-\gamma)^{k_i-1}
        &
        \le (m-1) \gamma \prod_{i=1}^m 2^{-k_i}(1-\gamma)^{k_i-1}
        \\
        &
        \le \frac{(m-1) \gamma}{4}
        ~.
        \tag{$\sum_{i=1}^k k_i \ge 2$}
    \end{align*}

    On the other hand, Theorem~\ref{thm:lp-solution} indicates that for $p(k) = 2^{-k} (1-\gamma)^{k-1}$:
    \[
        a(0) = \frac{3+\gamma}{12+6\gamma} \ge \frac{\gamma}{2}
        ~,
    \]
    for any $\gamma \le \frac{1}{4}$.%
    \footnote{In fact, this holds for any $0 \le \gamma \le \frac{\sqrt{61}-5}{6} \approx 0.468$}
    Hence, the $a(0)$'s in the second term on the right sum to at least $\frac{(m-1)\gamma}{4}$.
\end{proof}

\begin{proof}[Proof of Theorem~\ref{thm:two-way-online-primal-dual-analysis}]
    For any offline vertex $u$, consider the subset of rounds in which $u$ is shortlisted as $u_1$ or $u_2$ by Algorithm~\ref{alg:ocs-matching}.
    Further for any weight level $w > 0$, suppose that the subset of rounds in which $u$ is shortlisted by an online vertex $v$ with edge-weight $w_{uv} \ge w$ form consecutive subsequences of lengths $k_1, k_2, \dots, k_m$.
    We remark that if $u_1 = u_2 = u$ in the round of some online vertex $v$ with $w_{uv} \ge w$, this deterministic round contributes $2$ to the corresponding $k_i$.
    In such cases the sequel probability bounds hold trivially because $u$ is matched to an edge weight weight at least $w$ with certainty.
    The binding case of our analysis is when there are only randomized rounds.
    The OCS guarantee in the theorem statement ensures that the probability of matching $u$ to one of them is at least:
    \begin{equation}
        \label{eqn:matching-objective-by-vertex}
        y_u(w) \defeq 1 - p \Big( \sum_{i=1}^m k_i \Big) - \frac{1}{2}\sum_{i=2}^m \sum_{j=0}^{k_1+\dots+k_{i-1}-1} a(j)
        ~.
    \end{equation}
    

    Therefore, the expected maximum edge-weight matched to vertex $u$ is at least $\int_0^\infty y_u(w) dw$.
    The expected total weight of the matching by Algorithm~\ref{alg:ocs-matching} is at least:
    %
    \[
        \alg \defeq \sum_{u \in L} \int_0^\infty y_u(w) dw
        ~.
    \]

    The competitive analysis is a charging argument.
    For every online vertex $v \in R$, we split the changes of $\alg$ among the shortlisted offline verties $u_1, u_2$ and the online vertex $v$.
    Formally, let $\alpha_u = \int_0^\infty \alpha_u(w) dw$ be the gain of each offline vertex $u \in L$, where $\alpha_u(w)$ is the contribution from weight-level $w$.
    Let $\beta_v$ denote the gain of each online vertex $v$.
    Both are initially zero.
    Then, as an online vertex $v$ arrives and when $u \in \{ u_1, u_2 \}$ is shortlisted, suppose that $y_u(w)$ changes by $\Delta y_u(w)$ for any $0 \le w \le w_{uv}$:
    \begin{itemize}
        \item Increase $\beta_v$ by $\Delta_u \beta_v$ according to Eqn.~\eqref{eqn:matching-dual-update-beta}, which we restate below:
        \[
            \Delta_u \beta_v
            \defeq
            \int_0^{w_{uv}} b\big(k_u(w)\big)dw -\frac{1}{2} \int_{w_{uv}}^\infty \sum_{i=0}^{k_u(w)-1} a(i) dw
            ~.
        \]
        \item Increase $\alpha_u(w)$ by:
    \end{itemize}
    \begin{align}
        \label{eqn:matching-dual-update-alpha}
        \Delta \alpha_u(w)
        \defeq
        \begin{cases}
            \Delta y_u(w) - b\big(k_u(w)\big) & \text{if $w_{uv} \ge w$}\\[1ex]
            \frac{1}{2}\sum_{i=0}^{k_u(w)-1} a(i) & \text{if $0 \le w_{uv} < w$}
        \end{cases}
        ~.
    \end{align}
    
    We remark that the values of $k_u(w)$'s in the above charging rules are at the moment when $u$ is shortlisted by the algorithm for online vertex $v$.

%
%
    \paragraph{Feasibility of the Charging Rule.}
    We first verify that the total change in $\alpha_u$ and $\beta_v$ equals the change of $\alg$ due to online vertex $v$.
    By Equations~\eqref{eqn:matching-dual-update-beta} and \eqref{eqn:matching-dual-update-alpha}, the total change in the vertices' gains equals:
    \begin{align*}
        \underbrace{\int_0^{w_{uv}} \big(\Delta y_u(w) - b(k_u(w)) \big) dw + \frac{1}{2} \int_{w_{uv}}^\infty \sum_{i=0}^{k_u(w)-1} a(i) dw}_\text{change of $\alpha_u$}
        &
        \\
        +
        \underbrace{\int_0^{w_{uv}} b\big(k_u(w)\big)dw - \frac{1}{2} \int_{w_{uv}}^\infty \sum_{i=0}^{k_u(w)-1} a(i)dw}_\text{change of $\beta_v$}
        &
        = \int_0^{w_{uv}} \Delta y_u(w) dw
        ~.
    \end{align*}

    \paragraph{Invariant of Offline Gain.}
    Next we show that for any offline vertex $u$, and any positive weight-level $w > 0$:
    \begin{equation}
        \label{eqn:two-way-matching-offline-invariant}
        \alpha_u(w) \geq \sum_{i = 0}^{k_u(w)-1} a(i)
        ~.
    \end{equation}

    Consider the rounds in which $u$ is shortlisted and $u_1$ or $u_2$ (or both) and the edge-weight is at least $w$.
    Partition them into consecutive subsequences of the rounds that shortlist $u$, regardless the edge-weights.
    Let $k_1, k_2, \dots, k_m$ be the lengths of the consecutive subsequences.
    By considering the changes to $\alpha_i(w)$ due to the rounds in the subsuequences, and any $m-1$ rounds involving $u$ between the subsequences, one for each pair of neighboring subsequences, we get that:
    \begin{align*}
        \alpha_u(w)
        &
        \ge
        \underbrace{y_u(w) - \sum_{i=0}^{k_u(w)-1} b(i)}_\text{rounds in subsequences, 1st case of Eqn.~\eqref{eqn:matching-dual-update-alpha}}
        +
        \underbrace{\frac{1}{2} \sum_{i=2}^m \sum_{j = 0}^{k_1+\dots+k_{i-1}-1} a(i)}_\text{rounds in between, 2nd case of Eqn.~\eqref{eqn:matching-dual-update-alpha}}
        \\
        &
        = 1 - p\big(k_u(w)\big) - \sum_{i=0}^{k_u(w)-1} b(\ell)
        \tag{Eqn.~\eqref{eqn:matching-objective-by-vertex}, and $k_u(w) = k_1 + \dots + k_m$}
        \\
        &
        = \sum_{i=0}^{k_u(w)-1} \big( p(i) - p(i+1) - b(i) \big)
        \tag{$p(0) = 1$}
        \\
        &
        \ge \sum_{i=0}^{k_u(w)-1} a(i)
        ~.
        \tag{Eqn.~\eqref{eqn:two-way-lp-gain-split}}
    \end{align*}

    This is the only place in our argument that uses Eqn.~\eqref{eqn:relaxed-ocs} about the online selection algorithm, indirectly through Eqn.~\eqref{eqn:matching-objective-by-vertex}.
    We remark that in unweighted and vertex-weighted online bipartite matching, there is only one weight level $w = 1$ or $w = w_u$ of concern for any offline vertex $u$.
    Hence, there is only a single subsequence with all rounds that shortlist $u$ in the above argument.
    It suffices to replace Eqn.~\eqref{eqn:relaxed-ocs} by the weaker property of semi-OCS.

    \paragraph{Non-negativity of Gains.}
    The non-negativity of offline gains follows from the above invariant.
    The non-negativity of online gains follows by that $\Delta_u \beta_v = 0$ for the dummy vertex $u$.
    Hence, the offline neighbors $u_1, u_2$ shortlisted by Algorithm~\ref{alg:ocs-matching} have non-negative $\Delta_{u_1} \beta_v, \Delta_{u_2} \beta_v$.

    \paragraph{$\Gamma$-Approximate Equilibrium.}
    The gains cumulated by the online and offline vertices satisfy an approximate equilibrium condition in the sense that for any edge $(u, v) \in E$ the total gain of $u$ and $v$ is at least $\Gamma$ times the edge weight $w_{uv}$.
    By the definition of Algorithm~\ref{alg:ocs-matching}, and by that $\Delta_u \beta_v$ in Eqn.~\eqref{eqn:matching-dual-update-beta} is non-increasing in $k_u(w)$'s, we have $\beta_v \ge 2 \Delta_u \beta_v$ even when we compute $\Delta_u \beta_v$ \emph{using the final values of $k_u(w)$'s}.
    Hence:
    \begin{align*}
        \alpha_u + \beta_v
        &
        \ge \int_0^\infty \sum_{i=0}^{k_u(w)-1} a(i) dw + 2 \Delta_u \beta_v
        \tag{Eqn.~\eqref{eqn:two-way-matching-offline-invariant}}
        \\
        &
        = \int_0^{w_{uv}} \sum_{i=0}^{k_u(w)-1} a(i) dw + 2 \int_0^{w_{uv}} b\big(k_u(w)\big) dw 
        \tag{Eqn.~\eqref{eqn:matching-dual-update-beta}}
        \\[2ex]
        &
        \geq \Gamma w_{uv}
        ~.
        \tag{Eqn.~\eqref{eqn:approximate-dual-feasible}}
    \end{align*}

    Then, consider an optimal matching $M \subseteq E$.
    Algorithm~\ref{alg:ocs-matching} is $\Gamma$-competitive because:
    \begin{align*}
        \alg
        &
        = \sum_{u \in L} \alpha_u + \sum_{v \in R} \beta_v
        \\
        &
        \ge \sum_{(u, v) \in M} \big( \alpha_u + \beta_v \big)
        &
        \\
        & \ge \Gamma \sum_{(u,v) \in M} w_{uv}
        ~.
    \end{align*}

    Finally, we remark that the above analysis is mathematically equivalent an online primal dual analysis under the framework of \citet{DevanurJK:SODA:2013} for online bipartite matching, and also \citet{DevanurHKMY:TEAC:2016} and \citet{FahrbachHTZ:FOCS:2020} for the edge-weighted case.
    We choose the above exposition to avoid having to introduce the more general framework.
\end{proof}

Combining Algorithm~\ref{alg:ocs-matching} with the improved $0.167$-OCS from Theorem~\ref{thm:ocs} in Section~\ref{sec:ocs} surpasses the state-of-the-art $0.508$-competitive algorithm for edge-weighted online bipartite matching by \citet{FahrbachHTZ:FOCS:2020}.

\begin{corollary}
    There is an two-choice greedy algorithm for edge-weighted online bipartite matching that is at least $0.512$-competitive.
\end{corollary}

\subsection{A Variant of OCS and Edge-weighted Online Bipartite Matching}

This subsection considers another online selection algorithm tailored for the relaxed condition in Eqn.~\eqref{eqn:relaxed-ocs}.
Each element is associated with a flag $1$ or $0$, initialized uniformly at random.
In each round $t$, the algorithm samples an element $e$ from $\eset^t$ uniformly at random to probe its flag.
If its flag is $1$, the algorithm selects $e$ and sets its flag to $0$.
Otherwise, the algorithm selects the other element and sets $e$'s flag to $1$.
In other words, the algorithm randomly samples an element, lets its flag decides the selection, and flips the flag.
See Algorithm~\ref{alg:relaxed-ocs}.



\begin{algorithm}[t]
    \caption{A variant of OCS designed for edge-weighted online bipartite matching}
    \label{alg:relaxed-ocs}
    \begin{algorithmic}
        \smallskip
        \State \textbf{State variables:} (for each element $e$)
        \begin{itemize}
            \item $\tau_{e} \in \big\{0, 1\big\}$; its initial value $\tau_e^0$ is independently and uniformly at random.
        \end{itemize}
        \State \textbf{For each round $t$: }
        \begin{enumerate}
            \item Draw $e^t \in \eset^t$ uniformly at random.
            \item If $\tau_{e^t} = 1$, select $e^t$ and let $\tau_{e^t} = 0$.
            \item Otherwise, select the other element in $\eset^t$ and set $\tau_{e^t} = 1$.
        \end{enumerate}
    \end{algorithmic}
\end{algorithm}

\begin{theorem}
    \label{thm:relaxed-ocs}
    Algorithm~\ref{alg:relaxed-ocs} ensures the selection probability in Eqn.~\eqref{eqn:relaxed-ocs} for:
    \[
        p(k) = 2^{-k-\min\{k,\lceil\frac{k+2}{2}\rceil\}}+ k\cdot 2^{-k-\min\{k,\lceil\frac{k+3}{2}\rceil\}}
        ~.
    \]
\end{theorem}

Combining with Theorem~\ref{thm:two-way-online-primal-dual-analysis} further improves the competitive ratio of edge-weighted online bipartite matching.

\begin{corollary}
    \label{cor:edge-weighted-519}
    Algorithm~\ref{alg:ocs-matching}, using Algorithm~\ref{alg:relaxed-ocs} for online selections, is at least $0.519$-competitive for edge-weighted online bipartite matching.
\end{corollary}

\paragraph{Preliminaries on Boolean Formula with Uniform Input.}
Algorithm~\ref{alg:relaxed-ocs} uses two kinds of random bits that are sampled independently and uniformly:
the initial flags $(\tau_e^0)_{e \in \eset}$, and the sampled elements $(e^t)_{1 \le t \le T}$.
Viewing these random bits as boolean variables, we will represent each selection event by an XOR clause, i.e., an XOR of a subset these boolean variables, their negates, and the constant $1$, such as $X_1 \oplus \neg X_2 \oplus X_3 \oplus 1$.
Next we introduce two properties related to XOR clauses with uniform input.



\begin{lemma}
    \label{lem:xor-1}
    For any uniform and independent boolean variables and $m$ XOR clauses such that each variable is in at most one clause, the probability of satisfying all clauses equals $2^{-m}$.
\end{lemma}

\begin{proof}
    This is because each clause independently holds with probability half.
\end{proof}

\begin{lemma}
    \label{lem:xor-2}
    For any uniform and independent boolean variables and $m$ XOR clauses such that each variable is in at most \emph{two} clauses, the probability of satisfying all clauses is at most $2^{-\lceil\frac{m}{2}\rceil}$.
\end{lemma}

\begin{proof}
    Consider an undirected graph $G=(V,E)$ in which the vertices $V = \{1, 2, \dots m\}$ correspond to clauses, and the edges correspond the boolean variables that appears in two clauses:
    \[
        E = \big\{ (i,j)_k: \text{clauses $i$ and $j$ both involve variable $k$} \big\}
        ~.
    \]

    Next, consider any maximal matching $M$ of $G$.
    Let $\ell = |M|$ be the size of the matching.
    Let $U \subseteq V$ denote the set of $m-2\ell$ unmatched vertices (i.e., clauses).
    Since the matching is maximal, the unmatched clauses do not share any variables.
    Hence, over the randomness of the variables not in the matching $M$, the probability of satisfying all unmatched clauses equals $2^{-m+2\ell}$ by Lemma~\ref{lem:xor-1}.

    Further, \emph{conditioned on any realization of the variables not in the matching $M$} and over the randomness of the variables in $M$, each pair of matched clauses hold with probability at most $\frac{1}{2}$.
    
    Therefore, the probability of satisfying all clauses is at most $2^{-m+2\ell} \cdot 2^{-\ell} = 2^{-m+\ell}$.
    The lemma then follows by $\ell \le \lfloor \frac{m}{2} \rfloor$.
\end{proof}


\paragraph{Selection Probabilities.}
We next develop a lemma about probability of not selecting an element $e$ conditioned on the sampled elements $e^t$'s.
The proof of Theorem~\ref{thm:relaxed-ocs} will repeatedly use the lemma.

\begin{lemma}
    \label{lem:relaxed-ocs-conditional-probability-bound}
    For any element $e$ and any $k$ rounds $t_1 < t_2 < \cdots < t_k$ involving $e$, conditioned on any realization of $e^{t_1},\cdots, e^{t_k}$, the probability that $e$ is never selected in these $k$ rounds is at most:
    \[
        2^{-\min \left\{k,\lceil\frac{k+2}{2}\rceil \right\}}
        ~.
    \]
\end{lemma}

\begin{proof}
    We will prove a stronger result.
    If there are $d$ distinct elements in the realized $e^{t_1}, \dots, e^{t_k}$, then $e$ is unselected in these rounds with probability at most:
    \[
        \begin{cases}
        2^{-k} & \text{if $d=1$;}\\
        2^{-\lceil\frac{k+d}{2} \rceil} & \text{if $d\geq 2$.}
        \end{cases}
    \]

    We next introduce an XOR clause for each $t_i$, $1 \le i \le k$, so that not selecting $e$ in these rounds is equivalent to satisfying all $k$ clauses.
    If round $t_i$ is the earliest among these $k$ rounds that samples element $e^{t_i}$, i.e., $e^{t_i} \ne e^{t_j}$ for any $j < i$, consider a clause that represents the value of flag $\tau_{e^{t_i}}$ at the beginning of round $t_i$:
    \[
        \begin{cases}
            \tau_{e^{t_i}}^0 \oplus \big( \bigoplus_{t < t_i : e^{t_i} \in \eset^t} \mathbf{1} (e^t = e^{t_i} ) \big)
            &
            e^{t_i} = e ~;
            \\[1ex]
            1 \oplus \tau_{e^{t_i}}^0 \oplus \big( \bigoplus_{t < t_i : e^{t_i} \in \eset^t} \mathbf{1} (e^t = e^{t_i} ) \big)
            &
            e^{t_i} \ne e ~.
        \end{cases}
    \]
    We shall refer to such clauses as type-A clauses.

    Otherwise, suppose that element $e^{t_i}$ was most recently sampled in $t_j$, i.e., $e^{t_i} = e^{t_j}$ and $e^{t_i} \ne e^{t_\ell}$ for $j < \ell < i$.
    Consider a clause that represents the parity of the number of times that flag $\tau_{e^{t_i}}$ flips between the two rounds, \emph{including the flip due to round $t_j$}, i.e.: 
    \[
        \textstyle
        1 \oplus \big( \bigoplus_{t_j < t < t_i : e^{t_i} \in \eset^t} \mathbf{1} (e^t = e^{t_i} ) \big)
        ~.
    \]
    We shall refer to such clauses as type-B clauses.
    It captures if the value of $\tau_e$ at the begining of $t_i$ is the same as that at the beginning of round $t_j$, and thus still leads to not selecting element $e$.
    


    If $d = 1$, there are $1$ type-A clause and $k-1$ type-B clauses.
    Further, each variable appears in at most one clause.
    It then follows by Lemma~\ref{lem:xor-1}.

    If $d \ge 2$, there are $d$ type-A clauses and $k-d$ type-B clauses.
    First consider the type-B clauses and the random variables corresponding to the sampled elements $(e^t)_{1 \le t \le T}$.
    Each of these variables appears in at most two clauses.
    By Lemma~\ref{lem:xor-2} the probability of satisfying all these clauses is at most $2^{-\lceil \frac{k-d}{2} \rceil}$.
    Further, each type-A clause has a unique variable $\tau_{e^{t_i}}^0$.
    Hence, over any realization of the sampled elements, the probability of satisfying all $d$ type-A clauses is $2^{-d}$.
    Combining the two bounds proves the lemma.
\end{proof}

\begin{proof}[Proof of Theorem~\ref{thm:relaxed-ocs}]
    First recall the requirement of Eqn.~\eqref{eqn:relaxed-ocs}.
    For any element and any consecutive subsequences of the rounds involving the element with lengths $k_1, k_2, \dots, k_m$, we will upper bound the probability that Algorithm~\ref{alg:relaxed-ocs} never selects $e$ in these rounds by:
    \[
        p \Big( \sum_{i=1}^m k_i \Big) +\frac{1}{2}\sum_{i=2}^m \sum_{j=0}^{k_1+\dots+k_{i-1}-1} a(j)
        ~,
    \]
    where:
    \[
        p(k) = 2^{-k - \min\{k, \lceil \frac{k+2}{2} \rceil \}} + k 2^{-k-\min\{k, \lceil \frac{k+3}{2} \rceil \}}
        ~,
    \]
    and $(a(j))_{j \ge 0}$ take values as in the optimal solution given by Theorem~\ref{thm:lp-solution}.
    Importantly:
    \[
        a(0) \approx 0.2403 > 0.24
        ~.
    \]

    Further recall a remark after Theorem~\ref{thm:two-way-online-primal-dual-analysis} that the guarantee in Eqn.~\eqref{eqn:relaxed-ocs} holds almost trivially for three or more consecutive subsequences and for two subsequences whose total lengths are more than three.
    Hence, the proof will first handle the remaining cases before substantiating the remark.

    \paragraph{One Subsequence.}
    In this case the second term in Eqn.~\eqref{eqn:relaxed-ocs} disappears, so we will upper bound the probability by $p(k)$ alone.
    Suppose that $t_1, t_2, \dots, t_k$ are the rounds in the consecutive subsequence.
    Consider the number of them that sample $e^{t_i}=e$.
    If there are at least two, the flag $\tau_e$ must be $1$ in at least one of them, and by definition Algorithm~\ref{alg:relaxed-ocs} selects $e$ there. 

    If exactly one out of the $k$ rounds samples $e^{t_i} = e$, which happens with probability $k 2^{-k}$ over the randomness of $e^{t_i}$'s, Lemma~\ref{lem:relaxed-ocs-conditional-probability-bound} indicates that the probability of never selecting $e$ in the $k-1$ rounds with $e^{t_i} \ne e$, \emph{conditioned on the realized $e^{t_i}$'s}, is at most $2^{-\min \{k-1, \lceil \frac{k+1}{2} \rceil\}}$.
    Further, over the randomness of the initial flag $\tau_e$, the round with $e^{t_i} = e$ selects $e$ with probability half, independent to the realization of the other $k-1$ rounds.
    In sum, this case happens with probability at most:
    \[
        k 2^{-k - \min \{ k, \lceil \frac{k+3}{2} \rceil \}}
        ~.
    \]

    Finally, if none of the $k$ rounds samples $e^{t_i} = e$, which happens with probability $2^{-k}$ over the randomness of $e^{t_i}$'s, Lemma \ref{lem:relaxed-ocs-conditional-probability-bound} indicates that the probability of never selecting $e$ in these round is at most $2^{-\min\{k,\lceil\frac{k+2}{2}\rceil\}}$.
    In sum, this case happens with probability at most:
    \[
        2^{-k - \min \{ k, \lceil \frac{k+2}{2} \rceil \}}
        ~.
    \]

    Summing the probability bounds in the last two cases gives exactly $p(k)$.

    \paragraph{Two Subsequences, Two Rounds.}
    By Lemma~\ref{lem:relaxed-ocs-conditional-probability-bound}, the probability of not selecting element $e$ in the two rounds is at most $\frac{1}{4}$
    It then follows by $p(2) = \frac{3}{16}$ and by $a(0) > 0.24$.

    \paragraph{Two Subsequences, Three or More Rounds.}
    At least one subsequence must have at least two rounds.
    If two neighboring rounds in the same subsequence both sample $e^t = e$, which happens with probability $\frac{1}{4}$, flag $\tau_e$ must be $1$ in one of them and Algorithm~\ref{alg:relaxed-ocs} selects $e$ in that round.
    Otherwise, Lemma~\ref{lem:relaxed-ocs-conditional-probability-bound} indicates that the probability of never selecting $e$ in these three or more rounds is at most $\frac{1}{8}$.
    In total, the probability is at most $(1-\frac{1}{4}) \frac{1}{8}=\frac{3}{32}$.
    It then follows by $a(0) > 0.24$.

    \paragraph{Three or More Subsequences.}
    Lemma~\ref{lem:relaxed-ocs-conditional-probability-bound} indicates that the probability of never selecting $e$ in these three or more rounds is at most $\frac{1}{8}$.
    It then follows by $a(0) > 0.24$, since the second term in Eqn.~\eqref{eqn:relaxed-ocs} sum to at least $a(0)$.
\end{proof}

\subsection{Multi-way Semi-OCS and Unweighted and Vertex-weighted Online Bipartite Matching}
\label{sec:multi-way-semi-ocs-matching}

This subsection introduces an algorithm \balanceocs that combines an unbounded variant of the \balance algorithm~\cite{KalyanasundaramP:TCS:2000, MehtaSVV:JACM:2007} and a multi-way semi-OCS.
The former assigns one unit of masses to the offline neighbors of each online vertex.
The latter then selects one of them to which the online vertex will match.
We shall analyze its competitive ratio in the unweighted and vertex-weighted online bipartite matching problems.

\balance is parameterized by a non-increasing discounting function $b : [0, +\infty) \to [0, 1]$.
For each online vertex $v$, it continuously assigns one unit of masses to $v$'s neighbors,%
\footnote{In the context of fractional online matching, \balance fractionally matches $v$ to its neighbors.}
prioritizing the ones with the largest discounted weight $w_u b(y_u)$ where $y_u$ denotes the total mass assigned to $u$ so far.
To describe it as an algorithm instead of a continuous process, for any offline vertex $u$ and any threshold marginal utility $\theta \ge 0$, define:
\begin{equation}
    \label{eqn:multi-way-matching-mass-update}
    y_u(\theta) = b^{-1} \Big( \frac{\theta}{w_u} \Big)
    ~.
\end{equation}

We will explain shortly how to interpret the algorithm if $b$ is not continuous or is not strictly increasing, i.e., when the inverse function is not well-defined.
In any case, the discount function $b$ used by our algorithm is continuous and strictly monotone.
%

Define $y_u(\theta) = 0$ if $w_u b(0) < \theta$, e.g., when $\theta > w_u$.
Let $z^+$ denote $\max \{ z , 0 \}$.
Then, we may equivalently interpret the \balance algorithm as choosing a threshold $\theta$ such that:
\begin{equation}
    \label{eqn:multi-way-matching-mass-total}
    \sum_{u : (u, v) \in E} \big( y_u(\theta) - y_u \big)^+ = 1
    ~,
\end{equation}
and then assigning mass $\big( y_u(\theta) - y_u \big)^+$ to each vertex $u$.

For a discount function $b$ whose inverse is not well-defined, let $y_u^-(\theta) = \sup \{ y \ge 0 : w_u b(y) < \theta \}$ and $y_u^+(\theta) = \inf \{ y \ge 0 : w_u b(y) > \theta \}$.
The \balance algorithm chooses an appropriate threshold $\theta$ and choose $y_u(\theta) \in [y_u^-(\theta), y_u^+(\theta)]$ to satisfy Eqn.~\eqref{eqn:multi-way-matching-mass-total}.

The original \balance algorithm cannot assign more than one unit of total mass to any offline vertex, which introduces boundary considerations that complicate the above description.
In our setting, however, the masses are merely input of the multi-way semi-OCS, and therefore the total mass of an offline vertex could be arbitrarily large.

\begin{algorithm}[t]
    \caption{\balanceocs}
    \label{alg:unbounded-balanced}
    \begin{algorithmic}
        \smallskip
        \State \textbf{State variables:} (for each offline vertex $u$)
        \begin{itemize}
            \item Total mass $y_u$ allocated to offline vertex $u$ far;
                initially, $y_u=0$.
        \end{itemize}
        \State \textbf{On the arrival of an online vertex $v\in R$:}
        \begin{enumerate}
            \item Find threshold $\theta \in [0, \infty)$ that satisfies Eqn.~\eqref{eqn:multi-way-matching-mass-total}.
            \item For each neighbor $u$, let $x_u^v = \big( y_u(\theta) - y_u \big)^+$ be its mass in this round.
            \item Match $v$ to the neighbor that the multi-way semi-OCS selects, with $\vec{x}^v$ as the mass vector in this round.
        \end{enumerate}
    \end{algorithmic}
\end{algorithm}

\begin{theorem}
    \label{thm:multi-way-semi-ocs-matching}
    Suppose that $p : [0,\infty) \to [0,1]$ is decreasing and differentiable, and $p(0) = 1$.
    Then, unbounded \balance with a $p$-multi-way semi-OCS (Algorithm~\ref{alg:unbounded-balanced}) is $\Gamma$-competitive for unweighted and vertex-weighted online bipartite matching, where the competitive ratio $\Gamma$ and the corresponding discount function $b$ are from an optimal solution of the following continuous LP
        :
    \begin{align}
        \text{\rm maximize} \quad
        &
        \Gamma
        \tag{\balancelp}
        \\[1ex]
        \text{\rm subject to} \quad
        &
        a(y) + b(y) \leq -p'(y) 
        && \forall y \geq 0
        \label{eqn:multi-way-lp-gain-split}
        \\
        &
        \int_0^y a(z) dz + b(y) \ge \Gamma
        && \forall y \geq 0
        \label{eqn:multi-way-approximate-dual-feasible}
        \\
        & b(y') \le b(y)
        && \forall y' \ge y \label{eqn:multi-way-lp-monotone}
        \\[1.5ex]
        & a(y), b(y) \ge 0
        && \forall y \ge 0
        \notag
    \end{align}
\end{theorem}

\begin{proof}
    By the guarantee of $p$-multi-way semi-OCS, each offline vertex $u$ is matched by unbounded \balance with probability at least:
    \[
        1 - p\big( y_u \big)
        ~.
    \]

    Therefore, the expected total weight of the matched vertices is at least:
    \[
        \alg \defeq \sum_{u \in L} w_u \big( 1 - p (y_u) \big)
        ~.
    \]

    Similar to the competitive analysis of the two-choice algorithm (Algorithm~\ref{alg:ocs-matching}), for every online vertex $v$ we will distribute the increase of $\alg$ between vertex $v$ and its offline neighbors.
    Let $\alpha_u$ and $\beta_v$ be the distributed gain of any offline vertex $u$ and any online vertex $v$.
    They are initially zero.
    In the round of online vertex $v$, for each offline neighbor $u$, increase $\alpha_u$ by:
    \[
        w_u \int_{y_u}^{y_u + x_u^v} a(z) dz
        ~,
    \]
    where $y_u$ is the value right before $v$ arirves.

    Further, let $\beta_v$ be:
    \[
        \beta_v \defeq \sum_{u : (u, v) \in E} w_u \int_{y_u}^{y_u + x_u^v} b(z) dz
        ~.
    \]

    \paragraph{Feasibility of the Charging Rule.}
    The total gain distributed above is upper bounded by the increase in $\alg$ because:
    \begin{align*}
        \sum_{u : (u, v) \in E} w_u \int_{y_u}^{y_u + x_u^v} \big( a(z) + b(z) \big) dz
        &
        \le \sum_{u : (u, v) \in E} w_u \int_{y_u}^{y_u + x_u^v} -p'(z) dz
        \tag{Eqn.~\eqref{eqn:multi-way-lp-gain-split}}
        \\
        &
        = \sum_{u : (u, v) \in E} w_u \Big( \big( 1 - p(y_u + x_u^v) \big) - \big( 1 - p(y_u) \big) \Big)
        ~.
    \end{align*}

    \paragraph{Invariant of Offline Gain.}
    By definition, for any offline vertex $u$: 
    \[
        \alpha_u = w_u \int_0^{y_u} a(z) dz
        ~.
    \]

    \paragraph{Invariant of Online Gain.}
    For any online vertex $v$, since the algorithm prefers neighbors with larger $w_u b(y_u)$ and assigns one unit of mass, we get that for any $v$'s neighbor $u$:
    \[
        \beta_v \ge w_u b(y_u)
        ~.
    \]
    This holds for the final value of $y_u$ because $y_u$ only increase over time and the discount function $b$ is non-increasing.

    \paragraph{$\Gamma$-Approximate Equilibrium.}
    The gains satisfy an approximate equilibrium condition in the sense that for any edge $(u, v)$, the total gain of $u$ and $v$ is at least $\Gamma$ times the vertex-weight $w_u$:
    \begin{align*}
        \alpha_u + \beta_v
        &
        \ge w_u \int_0^{y_u} a(z) dz +  w_u b(y_u)
        \tag{Invariants}
        \\
        &
        \ge \Gamma w_u
        ~.
        \tag{Eqn.~\eqref{eqn:multi-way-approximate-dual-feasible}}
    \end{align*}

    Then, for any optimal matching $M \subseteq E$, unbounded \balance is $\Gamma$-competitive because:
    \begin{align*}
        \alg
        &
        \ge \sum_{u \in L} \alpha_u + \sum_{v \in R} \beta_v
        \\
        &
        \ge \sum_{(u, v) \in M} \big( \alpha_u + \beta_v \big)
        \\
        &
        \ge \Gamma \sum_{(u, v) \in M} w_u
        ~.
    \end{align*}

\end{proof}

The LP in Theorem~\ref{thm:multi-way-semi-ocs-matching} has continuously many variables and constraints.
Fortunately, we have an explicit optimal solution for most natural $p$ functions.

\begin{theorem}
    \label{thm:multi-way-lp-solution}
    Suppose that function $p : [0,\infty) \to [0,1]$ is decreasing, convex, and differentiable, and $p(0) = 1$.
    Then, an optimal solution to the \balancelp is:
    \begin{align*}
        \Gamma
        &
        = \int_{0}^\infty e^{-z} \big( 1 - p(z) \big) dz
        ~;
        \\
        b(y)
        &
        = -e^y\int_{y}^\infty p'(z)e^{-z}dz & \forall y\geq 0
        ~;
        \\[1ex]
        a(y)
        &
        = -p'(y)-b(y)&\forall y\geq 0
        ~.
    \end{align*}
\end{theorem}

The proof of this theorem is deferred to Appendix~\ref{app:multi-way-lp-solution}.

\begin{corollary}
    \label{cor:multi-way-matching}
    Unbounded \balance with the multi-way semi-OCS from Theorem~\ref{thm:multiway-ocs} in Section~\ref{sec:multi-way} is at least $0.593$-competitive for unweighted and vertex-weighted online bipartite matching.
\end{corollary}

\section*{Acknowledgment}

We thank Zhihao Gavin Tang and Hu Fu for helpful discussions on online contention resolution schemes.

\bibliographystyle{plainnat}
\bibliography{ocs}

\appendix

\section{Missing Proofs in Section~\ref{sec:semi-ocs}}
\label{app:semi-ocs}

\subsection{Positive Correlation in $3$-Way Sampling without Replacement}
\label{app:positive-correlation}

Consider the following counter-example which shows that there could be positive correlation in $3$-way (unweighted) sampling without replacement.
The elements are integers from $1$ to $9$.
It has $7$ rounds:
\[
    \eset^1 = \eset^2 = \big\{ 1, 4, 5 \big\}, \eset^3 = \eset^4 = \big\{ 2, 6, 7 \big\}, \eset^5 = \eset^6 = \big\{ 3, 8, 9 \big\}, \eset^7 = \big\{ 1, 2, 3 \big\}
    ~.
\]

Recall that $\uset^t$ denotes the subset of unselected elements after round $t$, and thus $e \in \uset^t$ denotes the event that element $e$ remains unselected after round $t$.
Further $\uset = \uset^7$ denotes the subset of unselected elements at the end.
On the one hand:
\[
    \Pr \big[ 1, 2 \in \uset \big] = \Pr \big[ s^1, s^2 \ne 1 \big] \Pr \big[ s^3, s^4 \ne 2 \big] \Pr \big[ s^5, s^6 \ne 3 \big] \Pr \big[ s^7 = 3 \mid 1, 2, 3 \in \uset^6 \big] 
    = \Big( \frac{1}{3} \Big)^4 = \frac{1}{81}
    ~.
\]

On the other hand:
\begin{align*}
    \Pr \big[ 1 \in \uset \big] &
    = \Pr \big[ s^1, s^2 \ne 1 \big] \Big(
    \Pr \big[ s^3, s^4 \ne 2 \big] \Pr \big[ s^5, s^6 \ne 3 \big] \Pr \big[ s^7 \ne 1 \mid 1, 2, 3 \in \uset^6 \big]  \\
    & 
    \hspace{90pt}
    + 
    \Pr \big[ s^3, s^4 \ne 2 \big] \Pr \big[ 3 \in \big\{ s^5, s^6 \big\} \big] \Pr \big[ s^7 \ne 1 \mid 1, 2 \in \uset^6, 3 \notin \uset^6 \big]
    \\
    & 
    \hspace{90pt}
    + 
    \Pr \big[ 2 \in \big\{ s^3, s^4 \big\} \big] \Pr \big[ s^5, s^6 \ne 3 \big] \Pr \big[ s^7 \ne 1 \mid 1, 3 \in \uset^6, 2 \notin \uset^6 \big] \Big)
    \\
    &
    = \frac{1}{3} \left( \frac{1}{3} \cdot \frac{1}{3} \cdot \frac{2}{3} + \frac{1}{3} \cdot \frac{2}{3} \cdot \frac{1}{2} + \frac{2}{3} \cdot \frac{1}{3} \cdot \frac{1}{2} \right) = \frac{8}{81} ~.
\end{align*}

Further by symmetry:
\[
    \Pr \big[ 2 \in \uset \big] = \frac{8}{81}
    ~.
\]

Therefore:
\[
    \frac{\Pr[1, 2 \in \uset]}{\Pr[1 \in \uset] \Pr[2 \in \uset]} = \frac{81}{64} > 1
    ~.
\]

\subsection{Proof of Theorem~\ref{thm:semi-ocs-hardness}}
\label{app:semi-ocs-hardness}

We shall construct a distribution of instances and prove the desired probability bound holds on average.
By considering an appropriate distribution of instance, we ensure that the randomness of the instance dictates the selection result.

Consider elements $\eset = \{ e^0_1 = 1, e^0_2 = 2, \dots, e^0_{2^k} = 2^k \}$.
Further for $i$ from $1$ to $k$, recursively define $e^i_j$ to be either $e^{i-1}_{2j-1}$ or $e^{i-1}_{2j}$ uniformly at random.
The instance has pairs $\{ e^i_{2j-1}, e^i_{2j} \}$ for all $0 \le i \le k-1$ and all $1 \le j \le 2^{k-i}$, in ascending order of $i$;
the order with respect to $j$ for any fixed $i$ is unimportant yet for concreteness we define it to be in ascending order as well.

In other words, this is a knockout-tournament-like instance.
First partition the $2^k$ elements into $2^{k-1}$ pairs in lexicographical order.
We shall refer to these pairs as the first stage of the instance.
Then, a randomly chosen ``winner'' from each pair advances to the next stage.
Repeat this process until we have the final ``winner'', denoted as $e^k_1$ by the construction above.
We shall refer to the pairs defined with respect to elements $e^{i-1}_j$, $1 \le j \le 2^{k-i+1}$, as the $i$-th stage of the instance.
For example, consider $2^3$ elements $\{1, 2, \dots, 8\}$.
A possible realization of the random instance proceeds as $\{ 1, 2 \}, \{3, 4\}, \{5, 6\}, \{7, 8\}, \{1, 3\}, \{5, 8\}, \{3, 5\}$.
See Figure 1 for an illustration.
\begin{figure}
    \begin{center}
        \begin{tikzpicture}[> = latex]
        \draw (-1,5.5) node{1};
        \draw (-1,5) node{2};
        \draw (-1.5,5.3) -- (-0.5,5.3);
        \draw (-0.5,5.3) -- (-0.5,4.8);
        \draw (-1.5,4.8) -- (-0.5,4.8);
        
        \draw (-1,4) node{3};
        \draw (-1,3.5) node{4};
        \draw (-1.5,3.8) -- (-0.5,3.8);
        \draw (-0.5,3.8) -- (-0.5,3.3);
        \draw (-1.5,3.3) -- (-0.5,3.3);
        
        \draw (-1,2.5) node{5};
        \draw (-1,2) node{6};
        \draw (-1.5,2.3) -- (-0.5,2.3);
        \draw (-0.5,2.3) -- (-0.5,1.8);
        \draw (-1.5,1.8) -- (-0.5,1.8);
        
        \draw (-1,1) node{7};
        \draw (-1,0.5) node{8};
        \draw (-1.5,0.8) -- (-0.5,0.8);
        \draw (-0.5,0.8) -- (-0.5,0.3);
        \draw (-1.5,0.3) -- (-0.5,0.3);
        
        \draw (0,5.25) node{1};
        \draw (0,3.75) node{3};
        \draw (-0.5,5.05) -- (0.5,5.05);
        \draw (0.5,5.05) -- (0.5,3.55);
        \draw (-0.5,3.55) -- (0.5,3.55);
        
        \draw (0,2.25) node{5};
        \draw (0,0.75) node{8};
        \draw (-0.5,2.05) -- (0.5,2.05);
        \draw (0.5,2.05) -- (0.5,0.55);
        \draw (-0.5,0.55) -- (0.5,0.55);
        
        \draw (1,4.5) node{3};
        \draw (1,1.5) node{5};
        \draw (0.5,4.3) -- (1.5,4.3);
        \draw (1.5,4.3) -- (1.5,1.3);
        \draw (0.5,1.3) -- (1.5,1.3);
     \end{tikzpicture}
    \end{center}
    \caption{Knockout tournament-like input sequence.}
    \label{tournament}
\end{figure}
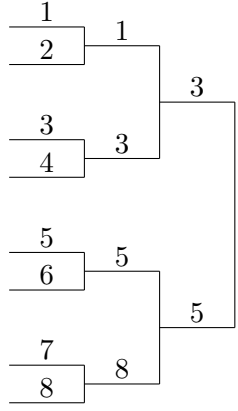

Next we show that with probability at least $2^{-2^k+1}$, the final ``winner'' $e^k_1$ is never selected by the algorithm, despite its $k$ appearances.
Here the probability space is over both the randomness of the algorithm and that of the instance.
In fact, we shall inductively prove the following stronger invariant;
the above claim is the special case when $i = k$. 

\paragraph{Invariant:}
For any $1 \le i \le k$, after processing the pairs involving elements $e^{i-1}_j$'s, i.e., after the first $i$ stages of the tournament, all elements $e^i_j$, $1 \le j \le 2^{k-i}$, in the next stage remain unselected with probability at least:
\[
    2^{-2^k+2^{k-i}}
    ~.
\]

The base case when $i = 0$ is vacuously true.

Next suppose that the invariant holds for $i-1$, and consider the case of $i$.
Below are a set of sufficient conditions under which all elements $e^i_j$ are unselected after the first $i$ stages:
\begin{enumerate}
    \item for any $1 \le j \le 2^{k-i+1}$, $e^{i-1}_j$ is unselected after the first $i-1$ stages;
    \item for any $1 \le j \le 2^{k-i}$, $e^i_j \in \{ e^{i-1}_{2j-1}, e^{i-1}_{2j} \}$ is \emph{not} the one selected by the algorithm for the pair.
\end{enumerate}

By the construction of the random instance, the two events are independent.
The first holds with probability at least $2^{-2^k+2^{k-i+1}}$ by the inductively hypothesis.
The second holds with probability $2^{-2^{k-i}}$ by the construction of the random instance.
Hence, the probability that all elements $e^i_j$ remain unselected after the first $i$ stages is at least:
\[
    2^{-2^k+2^{k-i+1}} \cdot 2^{-2^{k-i}} = 2^{-2^k+2^{k-i}}
    ~.
\]

\section{Missing Proofs in Section~\ref{sec:multi-way}}
\label{app:multi-way}

\subsection{Proof of Lemma~\ref{lemma:multiway-cubic}}
\label{app:multiway-cubic}

Both sides equal $1$ when $x = 0$.
The rest of the proof considers $0<x<1$.
Let $t=\frac{y}{x}$.
We have:
\begin{align*}
    \frac{x}{1-x} w(y)+1-\frac{w(y+x)}{w(y)}
    =&\frac{x}{1-x} w(tx)+1-\frac{w(tx+x)}{w(tx)}\\
    =& \int_{0}^{x}\left( B(x',t) w(t x')-A(x',t)\frac{w(tx'+x')}{w(tx')}\right)dx'\\
    =&\int_{0}^{x} B(x',t) w(t x')\left(1-\frac{A(x',t) w(tx'+x')}{B(x',t) w(tx')^2}\right) dx'\\
    =&\int_{0}^{x} B(x',t) w(t x')\int_{0}^{x'}\frac{C(x'',t)w(tx''+x'')}{B(x'',t)^2 w(tx'')^2}dx'' dx',
\end{align*}
where:
\begin{align*}
    Q(x)
    &
    =\frac{1}{w(x)}\frac{dw(x)}{dx}=1+x+3c x^2
    ~,
    \\[1ex]
    A(x,t)
    &
    =(t+1)Q(t x+x)-t Q(t x) 
    ~,
    \\[1ex]
    B(x,t)
    &
    = \frac{1}{(1-x)^2}+\frac{tx}{1-x} Q(tx)
    ~,
    \\
    C(x,t)
    &
    =A(x,t)\frac{dB}{dx}(x,t)-B(x,t)\frac{dA}{dx}(x,t)-A(x,t)B(x,t)\left((t+1)Q(t x+x)-2t Q(t x)\right)
    ~.
\end{align*}

To prove the lemma, i.e.:
\[
    \frac{w(y+x)}{w(y)}\le \frac{x}{1-x} w(y)+1
    ~,
\]
for any $0 < x < 1$ and $y \ge 0$, we only need to prove $C(x,t) \ge 0$ for any $0 < x < 1$ and any $t \ge 0$.
It is equivalent to show that for any $x, t\ge 0$:
\[
    \frac{(x+1)^6}{x} C\left(\frac{x}{1+x},t\right)\ge 0
    ~.
\]

By computation:
\[
    \frac{(x+1)^6}{x} C\left(\frac{x}{1+x},t\right)=\sum_{i=0}^8 x^i P_i(t)
    ~,
\]
where 
\begin{align*}
    P_0(t)=& -18 c t^2 - 18 c t - 6 c \\
    &+ 6t^2 + 2,\\
    P_1(t)=& 12 c t^3 - 126 c t^2 - 126 c t - 42 c \\&
    + 8 t^3 + 36 t^2 + 17,\\
    P_2(t)=& 42 c t^4 + 54 c t^3 - 408 c t^2 - 390 c t - 126 c \\
    &+ 5 t^4 + 38 t^3 + 88 t^2 + 5 t + 64,\\
    P_3(t)= &54 c^2 t^5 - 27 c^2 t^4 - 144 c^2 t^3 - 135 c^2 t^2 - 54 c^2 t - 9 c^2 \\
    &+ 36 c t^5 + 156 c t^4 + 36 c t^3 - 780 c t^2 - 672 c t - 204 c \\
    &+ 2 t^5 + 17 t^4 + 72 t^3 + 115 t^2 + 30 t + 139,\\
    P_4(t)= &63 c^2 t^6 + 135 c^2 t^5 - 306 c^2 t^4 - 729 c^2 t^3 - 594 c^2 t^2 - 225 c^2 t - 36 c^2 \\
    &+ 21 c t^6 + 75 c t^5 + 183 c t^4 - 117 c t^3 - 882 c t^2 - 654 c t - 180 c \\
    &+ 6 t^5 + 21 t^4 + 72 t^3 + 91 t^2 + 74 t + 190,\\
    P_5(t)= &72 c^2 t^7 - 144 c^2 t^5 - 783 c^2 t^4 - 1242 c^2 t^3 - 927 c^2 t^2 - 342 c^2 t - 54 c^2 \\
    &+ 42 c t^6 + 42 c t^5 + 96 c t^4 - 174 c t^3 - 510 c t^2 - 300 c t - 66 c \\
    &+ 6 t^5 + 11 t^4 + 44 t^3 + 49 t^2 + 96 t + 167,\\
    P_6(t)= &81 c^3 t^8 - 162 c^3 t^7 - 459 c^3 t^6 - 405 c^3 t^5 - 162 c^3 t^4 - 27 c^3 t^3 \\
    &+ 72 c^2 t^7 - 63 c^2 t^6 - 171 c^2 t^5 - 531 c^2 t^4 - 801 c^2 t^3 - 603 c^2 t^2 - 225 c^2 t - 36 c^2\\
    &+ 21 c t^6 + 3 c t^5 + 57 c t^4 - 51 c t^3 - 66 c t^2 + 12 c t + 18 c \\
    &+ 2 t^5 + 2 t^4 + 18 t^3 + 19 t^2 + 69 t + 92,\\
    P_7(t)=& 54 c^2 t^5 - 27 c^2 t^4 - 144 c^2 t^3 - 135 c^2 t^2 - 54 c^2 t - 9 c^2 \\
    &+ 30 c t^4 + 12 c t^3 + 60 c t^2 + 66 c t + 24 c \\
    &+ 4 t^3 + 4 t^2 + 26 t + 29,\\
    P_8(t)=& 18 c t^2 + 18 c t + 6 c \\
    &+ 4 t + 4.
\end{align*}

When $c=\frac{4-2\sqrt{3}}{3}$, we can verify, using numerical computation software, that $P_i(t)$ has positive leading coefficient and no non-negative real roots for $i=1,2,\ldots, 8$, and:
\[
    P_0(t)=\left( 4 \sqrt{3}-6\right) \left(1 - \sqrt{3} t\right)^2\ge 0
    ~. 
\]

Therefore, for any $x, t \ge 0$:
\[
    \frac{(x+1)^6}{x} C\left(\frac{x}{1+x},t\right)\ge 0
    ~.
\]

So the lemma holds.

\subsection{Proof of Lemma~\ref{lemma:multiway-condition}}
\label{app:multiway-condition}

We will use the following lemma which follows by the definition of the weight function $w$ in Eqn.~\eqref{eqn:multiway-weight}.

\begin{lemma}
    \label{lem:multiway-weight-log-convex}
    Function $\log \circ\,w$ is convex in $[0, \infty)$.
\end{lemma}

If $\sum_{i=1}^k x_i \ge 1$, the left-hand-side is zero so the inequality holds trivially.
If $\sum_{i=1}^k x_i = 0$, both sides are equal to $1$ so the inequality also holds.
The rest of the proof consider $0 < \sum_{i=1}^k x_i < 1$.

Fix any non-negative $y_1, y_2, \dots, y_k \ge 0$.
Define function $f:[0,+\infty)^k\to (-\infty,+\infty)$ over $\vec{x} = (x_1, x_2, \dots, x_k)$ be the difference between the logarithms of the two sides, i.e.:
\[
    f(\vec{x})= \log \frac{1-\sum_{i=1}^k x_i}{\sum_{i=1}^k x_i w(y_i) +1- \sum_{i=1}^k x_i} -\sum_{i=1}^k\log\frac{w(y_i)}{w(y_i+x_i)}
    ~.
\]

We first argue that $f$ is convex over a simplex, using the log-convexity of the weight function $w$ (Lemma~\ref{lem:multiway-weight-log-convex}).
Consider any $\vec{x} = (x_1, \dots, x_k), \vec{x'} = (x_1', \dots, x_k') \in [0, +\infty)^k$ such that:
\[
    \sum_{i=1}^k x_i = \sum_{i=1}^k x'_i
    ~.
\]

Then, for any $t \in [0,1]$ and the linear $\vec{x''} = t \vec{x} + (1-t) \vec{x'}$, we have:
\begin{align*}
    t f(\vec{x})+(1-t) f(\vec{x'})
    &
    =
    \log\left(1-\sum_{i=1}^k x_i\right)-t \log\left(\sum_{i=1}^k x_i w(y_i) +1- \sum_{i=1}^k x_i\right)\\
    &
    \quad
    -(1-t)\log\left(\sum_{i=1}^k x'_i w(y_i) +1- \sum_{i=1}^k x_i\right)\\
    &
    \quad
    -\sum_{i=1}^{k} \log w(y_i) +\sum_{i=1}^k\left( t \log w(y_i+x_i) +(1-t) \log w(y_i+x'_i)\right)
    ~.
\end{align*}

By the log-convexity of $w$, and by Jensen's inequality on the second and third terms and on the last two terms, this is at least:
\[
    \log\left(1-\sum_{i=1}^k x_i\right)-\log\left(\sum_{i=1}^k x''_i w(y_i) +1- \sum_{i=1}^k x_i\right)-\sum_{i=1}^{k} \log w(y_i) +\sum_{i=1}^k  \log w(y_i+x''_i)
    = f(\vec{x''})
    ~.
\]

Therefore, we have:
\begin{align*}
    f(\vec{x})
    &
    \le \frac{1}{\sum_{i=1}^k x_i} \sum_{i=1}^k x_i f \Big(0,\ldots , 0, \underbrace{\sum_{j=1}^k x_j}_\text{$i$-th entry}, 0,\ldots, 0 \Big)
    \tag{convexity of $f$ on simplex}
    \\
    &
    = \frac{1}{\sum_{i=1}^k x_i}\sum_{i=1}^k x_i \bigg(\log\frac{w(y_i+\sum_{j=1}^k x_j)}{w(y_i)}-\log\Big(\frac{\sum_{j=1}^k x_j}{1-\sum_{j=1}^k x_j} w(y_i)+1\Big)\bigg)
    \\[3ex]
    &
    \le 0
    ~.
    \tag{Lemma~\ref{lemma:multiway-cubic}}
\end{align*}
    
\subsection{A Weaker Version of Lemma~\ref{lemma:multiway-cubic}}
\label{app:multiway-cubic-weak}

The following lemma is a weak version of Lemma~\ref{lemma:multiway-cubic}.
We provide it and a proof that does not involve computer-aided numerical verification.

\begin{lemma}
    For any $y \ge 0$ and any $0 < \delta < 1$:
    \[
        \exp \Big( (y+\delta) + \frac{(y+\delta)^2}{2} + \frac{(y+\delta)^3}{6} \Big) \le \exp \Big( y + \frac{y^2}{2} + \frac{y^3}{6} \Big) + \frac{\delta}{1-\delta} \exp \Big( 2y + y^2 + \frac{y^3}{3} \Big)
        ~.
    \]
\end{lemma}

\begin{proof}
    The inequality holds with equality when $\delta = 0$.
    Hence, it suffices to consider the partials of both sides with respect to $\delta$, and to show that the partial of the right-hand-side is larger, i.e.:
    \[
        \Big(1 + (y+\delta) + \frac{(y+\delta)^2}{2}\Big) \exp \Big( (y+\delta)+\frac{(y+\delta)^2}{2}+\frac{(y+\delta)^3}{6} \Big) \le \frac{1}{(1-\delta)^2} \exp \Big( 2y + y^2 + \frac{y^3}{3} \Big)
        ~.
    \]

    Taking logarithm of both sides and rearrange terms, it is equivalent to:
    \[
        \ln \Big( 1 + (y+\delta) + \frac{(y+\delta)^2}{2} \Big) + 2 \ln \big( 1-\delta \big) \le (y-\delta) + \frac{(y-\delta)^2}{2} + \frac{(y-\delta)^3}{6} - \delta^2 - y\delta^2
        ~.
    \]

    Since $\ln(1-\delta) \le - \delta - \frac{\delta^2}{2} - \frac{\delta^3}{3} - \frac{\delta^4}{4}$ for all $0 < \delta < 1$, after an rearrangement of terms it suffices to show that:
    \[
        \ln \Big( 1 + (y+\delta) + \frac{(y+\delta)^2}{2} \Big) \le (y+\delta) + \frac{(y-\delta)^2}{2} + \frac{y^3}{3} + \frac{2\delta^3}{3} - \frac{(y+\delta)^3}{6} + \frac{\delta^4}{2}
        ~.
    \]

    By the inequality of arithmetic and geometric means:
    \[
        \frac{(y-\delta)^2}{2} + \frac{\delta^4}{2} \ge \big| y - \delta \big| \delta^2
        ~.
    \]

    Hence, we arrive at the final inequality that is sufficient for establishing the inequality regarding the partial derivatives with respect to $y$ and thus, the correct of the lemma:
    \begin{equation}
        \label{eqn:multi-way-idealized-transform}
        \ln \Big( 1 + (y+\delta) + \frac{(y+\delta)^2}{2} \Big) \le (y+\delta) + \big|y-\delta\big|\delta^2 + \frac{y^3}{3} + \frac{2\delta^3}{3} - \frac{(y+\delta)^3}{6}
        ~.
    \end{equation}

    \paragraph{Roadmap of Proving Eqn.~\eqref{eqn:multi-way-idealized-transform}.}
    The naural next step is to upper bound $\ln\big(1+x+\frac{x^2}{2}\big)$ by a polynomial of $x$ to tranform the left-hand-side of the inequality a polynomial over $y$ and $\delta$, just like the right-hand-side.
    The Taylor series $\ln\big(1+x+\frac{x^2}{2}\big) = x - \frac{x^3}{6} + \frac{x^4}{8} - \frac{x^5}{20} + O(x^7)$ suggests natural upper bounds such as $\ln\big(1+x+\frac{x^2}{2}\big) \le x$ and $\ln\big(1+x+\frac{x^2}{2}\big) \le x - \frac{x^3}{6} + \frac{x^4}{8}$.
    Unfortuantely, neithor of these bounds proves Eqn.~\eqref{eqn:multi-way-idealized-transform} for all $y \ge 0$ and all $0 < \delta < 1$.
    In particular, the former fails when $\delta = y$, and the latter leaves a degree-$4$ term that cannot be bounded by the right-hand-side of Eqn.~\eqref{eqn:multi-way-idealized-transform} for sufficiently large $y$.
    Instead, we shall consider three polynomial upper bounds of $\ln \big(1+x+\frac{x^2}{2}\big)$ depending on the range of $x$ and together they cover all the cases.
    These upper bounds are from the next lemma, whose proof is deferred to the end of the subsection.

    \begin{lemma}
        \label{lem:multi-way-idealized-relaxation}
        The function:
        \[
            \frac{x - \ln (1+x+\frac{x^2}{2})}{x^3}
            ~.
        \]
        is decreasing for $x > 0$
    \end{lemma}

    \paragraph{Case 1: Use $\ln \big( 1+x+\frac{x^2}{2} \big) \le x$ for $x \ge 0$, when $y \ge (1+\sqrt[3]{2}) \delta$.}
    In other words, by inequality $\ln \big( 1 + (y+\delta) + \frac{(y+\delta)^2}{2} \big) \le y+\delta$ and $y > \delta$, Eqn.~\eqref{eqn:multi-way-idealized-transform} reduces to:
    \[
        \frac{y^3}{6} - \frac{y^2\delta}{2} + \frac{y\delta^2}{2} - \frac{\delta^3}{2} \ge 0
        ~,
    \]
    or equivalently:
    \[
        \Big( \frac{y}{\delta} \Big)^3 - 3 \Big( \frac{y}{\delta} \Big)^2 + 3 \frac{y}{\delta} - 3 = \Big( \frac{y}{\delta} - 1 \Big)^3 - 2 \ge 0
        ~.
    \]
    
    We remark that this approach can also prove Eqn.~\eqref{eqn:multi-way-idealized-transform} for $y \le c \delta$ for $c \approx 0.83$.
    We skip it since it is covered by the other cases.

    \paragraph{Case 2: Use $\ln \big( 1 + x + \frac{x^2}{2} \big) \le x - \frac{x^3}{24}$ for $0 \le x \le \frac{7}{3}$, if $y + \delta \le \frac{7}{3}$.}
    The stated inequality follows by Lemma~\ref{lem:multi-way-idealized-relaxation} and that for $x = \frac{7}{3}$:
    \[
        \frac{x - \ln (1+x+\frac{x^2}{2})}{x^3} \approx 0.0419 > \frac{1}{24}
        ~.
    \]

    By $\ln \big( 1 + (y+\delta) + \frac{(y+\delta)^2}{2} \big) \le y+\delta - \frac{(y+\delta)^3}{24}$, Eqn.~\eqref{eqn:multi-way-idealized-transform} reduces to:
    \[
        \big|y-\delta\big|\delta^2 + \frac{y^3}{3} + \frac{2\delta^3}{3} - \frac{(y+\delta)^3}{8} \ge 0
        ~.
    \]

    If $y \ge \delta$, the left-hand-side equals:
    \[
        \frac{5}{24} y^3 - \frac{3}{8} y^2 \delta + \frac{5}{8} y \delta^2 - \frac{11}{24} \delta^3 = \frac{1}{24} \big( y-\delta \big) \big( 5y^2 - 4y\delta + 11 \delta^2 \big) \ge 0
        ~.
    \]

    If $y < \delta$, the left-hand-side equals:
    \[
        \frac{5}{24} y^3 - \frac{3}{8} y^2 \delta - \frac{11}{8} y \delta^2 + \frac{37}{24} \delta^3 = \frac{1}{24} \big(\delta-y\big) \big( 37 \delta^2 + 4\delta y - 5y^2 \big) \ge 0
        ~.
    \]

    \paragraph{Case 3: Use $\ln \big( 1+x+\frac{x^2}{2} \big) \le x - \frac{x^3}{36}$ for $0 \le x \le 2+\sqrt[3]{2}$, if $\frac{7}{3} \le y+\delta \le 2+\sqrt[3]{2}$.}
    The stated inequality follows by Lemma~\ref{lem:multi-way-idealized-relaxation} and that for $x = 2+\sqrt[3]{2}$:
    \[
        \frac{x - \ln (1+x+\frac{x^2}{2})}{x^3} \approx 0.028 > \frac{1}{36}
        ~.
    \]

    The argument for this case combines the assumption of $y+\delta > \frac{7}{3}$ and that $0 < \delta < 1$ to derive:
    \[
        y > \frac{4}{3} \delta
        ~.
    \]

    By $\ln \big( 1 + (y+\delta) + \frac{(y+\delta)^2}{2} \big) \le y+\delta - \frac{(y+\delta)^3}{36}$ and by $y > \delta$, Eqn.~\eqref{eqn:multi-way-idealized-transform} reduces to:
    \[
        (y-\delta)\delta^2 + \frac{y^3}{3} + \frac{2\delta^3}{3} - \frac{5(y+\delta)^3}{36} \ge 0
        ~.
    \]

    Rearranging terms, the left-hand-side equals:
    \[
        \frac{7}{36} y^3 - \frac{5}{12} y^2\delta + \frac{7}{12} y\delta^2 - \frac{17}{36} \delta^3 = \frac{1}{144} \big(5y-6\delta\big)^2 y + \frac{1}{48} \big(y^2-\delta^2\big)y + \frac{17}{48} \delta^2 \big( y - \frac{4}{3} \delta \big) \ge 0
        ~.
    \]

    \bigskip
    
    The three cases cover all $y \ge 0$ and $0 < \delta < 1$ since the last two cases prove the lemma for any $y, \delta$ that satisfies $y + \delta \le 2 + \sqrt[3]{2}$.
    For any $y, \delta$ with $y + \delta > 2 + \sqrt[3]{2}$ (and $0 < \delta < 1$) must satisfy $y > (1+\sqrt[3]{2}) \delta$ and therefore is covered by the first case. 
\end{proof}

\begin{proof}[Proof of Lemma~\ref{lem:multi-way-idealized-relaxation}]
    We shall prove that the derivative is non-positive, i.e.:
    \[
        \frac{d}{dx} \frac{x-\ln(1+x+\frac{x^2}{2})}{x^3} = \frac{3}{x^4} \Big( \ln \big(1+x+\frac{x^2}{2}\big) - x + \frac{x^3}{6(1+x+\frac{x^2}{2})} \Big) \le 0
            ~.
    \]

    Equivalently, we need to show that:
    \[
        \ln \big(1+x+\frac{x^2}{2}\big) - x + \frac{x^3}{6(1+x+\frac{x^2}{2})} \le 0
        ~.
    \]
    
    Since this holds with equality at $x = 0$, it suffices to prove that its derivative is non-positive.
    This follows by:
    \[
        \frac{d}{dx} \Big( \ln \big(1+x+\frac{x^2}{2}\big) - x + \frac{x^3}{6(1+x+\frac{x^2}{2})} \Big) = - \frac{x^3(1+x)}{6(1+x+\frac{x^2}{2})^2} \le 0
        ~.
    \]
\end{proof}

\section{Missing Proofs in Section~\ref{sec:ocs}}
\label{app:ocs}

\subsection{Proof of Lemma~\ref{lem:pseudo-path}}
It suffices to show that the pseudo-paths are pairwise disjoint. Consider an arc $(p,c)$ and a pseudo-path $P$ that involves it. It is clear that $(p,c)$ can be adjacent with at most two arcs: the other in-arc of $c$ (if exists) and the other out-arc of $p$ (if exists). As $P$ is maximal, if $c$ has another in-arc in $G^{\text{ex-ante}}$, it should be adjacent to $(p,c)$ in $P$. Similarly, if $p$ has another out-arc $(p,c')$ in $G^{\text{ex-ante}}$ such that rounds $p,c,c'$ have a common element, it should also be adjacent to $(p,c)$ in $P$. Therefore, for any arc $(p,c)$, the set of its adjacent arcs in any $P$ is fixed according to the ex-ante dependence graph. Hence, for any arc $(p,c)$, there is a unique pseudo-path in the collection that involves it. 

\subsection{Proof of Lemma~\ref{lem:pseudo-matching}}

By the definition of good forests, pseudo-paths and pseudo-matchings, the following statements are equivalent:
\begin{enumerate}
    \item A subgraph of the ex-ante dependence graph is a good forest.
    \item A subgraph satisfies: each node has at most one in-arc; and there is no node $p$ with two out-arcs $(p,c)$ and $(p,c')$ such that rounds $p,c,c'$ have a common element.
    \item A subgraph is a union of pseudo-matchings, one for each pseudo-path.
\end{enumerate}

In particular, the equivalence between the first two follows by the definition of good forests.
The equivalence between the last two follows by the definitions of pseudo-paths and pseudo-matchings.

\subsection{Proof of Lemma~\ref{lem:pseudo-path-arrival}}

\begin{lemma}
\label{lem:app-forest-constructor-common-element}
For any three different rounds $1\leq p<c<c'\leq T$ such that $(p,c),(p,c')\in E^{\text{ex-ante}}$ and all rounds have a common element, rounds $p,c'$ have the same set of elements, i.e., $\eset^p=\eset^{c'}$. 
\end{lemma}
\begin{proof}
Suppose $\eset^p=\{a,b\}$ and the subscript of $(p,c)$ is $a$.
Then the subscript of $(p,c')$ is $b$.
Since $c'$ is first round involving element $b$ after round $p$, we have $b\notin \eset^{c}$. Further by that $p,c,c'$ have a common element, the common element can only be $a$. In sum, $\eset^{c'}=\{a,b\}=\eset^p$. 
\end{proof}

We shall prove the lemma by contradiction.
Suppose for contrary that there is a pseudo-path $P$ and $1 \leq t \leq T$ such that the subset of arcs in $P$ that arrive in the first $t$ rounds is not a sub-pseudo-path.
Since all arcs in $P$ arrive after round $T$ and form a single pseudo-path, there must be a round $t'>t$ such that in-arcs of $t'$ concatenate two sub-pseudo-paths.

Because of the definition of pseudo-paths and Lemma~\ref{lem:pseudo-path}, for each round $t$, the in-arcs of $t$ are simultaneously added into a pseudo-path such that they are adjacent in the pseudo-path. As a result, there are only two possibilities of the concatenation:
\begin{enumerate}
    \item If $t'$ has only one in-arc, this in-arc is next to only at most one other arc in the pseudo-path by definition.
        Hence, it cannot concatenate two sub-pseudo-paths.
    \item If $t'$ has two in-arcs $(p,t'), (p',t')$ and it concatenates two sub-pseudo-paths together in round $t'$, there must be $c,c'<t'$ such that $(p,c), (p',c') \in P$, and further rounds $p,c,t'$ must have a common element and rounds $p',c',t'$ must have a common element.
        By Lemma~\ref{lem:app-forest-constructor-common-element}, we have $\eset^p=\eset^t=\eset^{p'}$.
        If $p\neq p'$, the in-arc from the earlier one, e.g., $(p,t')$ shall not exist in the ex-ante dependence graph by definition.
        If $p=p'$, on the other hand, arcs $(p,c)$ $(p',c)$ do not exist in the ex-ante dependence graph.
        In fact, they shall be two parallel arcs $(p, t')$ that form a pseudo-path on their own, like the right-most pseudo-path in Figure~\ref{fig:ocs-forest-constructor-pseudo-path}.
\end{enumerate}

In sum, there is always a contraction in all cases.

\subsection{Proof of Lemma~\ref{lem:forest-constructor-automata-stationary}}

With the rows and columns in the order of $\unmatched, \ready, \matched$, the transition matrices $P^+$ of $\sigma^+$ and $P^-$ of $\sigma^-$ are as follows:
\[
		P^+
        =
        \begin{bmatrix}
            0 & 1-p & p \\
            0 & 0 & 1 \\
            1 & 0 & 0
        \end{bmatrix}
       	~,
        \qquad
       	P^-
       	=
       	\begin{bmatrix}
       		0 & 0 & 1\\
       		1 & 0 & 0\\
       		p & 1-p & 0
       	\end{bmatrix}
       	~.
\]

Then, the proof of the stationary distribution follows from the next two equations:
\begin{align*}
\vec{\pi} P^+ &= \bigg( \frac{1}{3-p}, (1-p)\cdot \frac{1}{3-p}, p\cdot \frac{1}{3-p}+\frac{1-p}{3-p} \bigg) = \bigg( \frac{1}{3-p}, \frac{1-p}{3-p}, \frac{1}{3-p} \bigg) = \vec{\pi}~,\\
\vec{\pi} P^- &= \bigg( \frac{1-p}{3-p}+p\cdot \frac{1}{3-p}, (1-p)\cdot \frac{1}{3-p}, \frac{1}{3-p} \bigg) = \bigg( \frac{1}{3-p}, \frac{1-p}{3-p}, \frac{1}{3-p} \bigg) = \vec{\pi}~.
\end{align*}

\subsection{Proof of Lemma \ref{lem:forest-constructor-automata-inverse}}

By the definitions of $\sigma^+$ and $\sigma^-$, the state sequences determine the corresponding choice sequence, i.e. $\state^i=\matched$ if and only if $\choice_i=\yes$ and $\hat{\state}^i=\matched$ if and only if $\hat{\choice}_i=\yes$.
Therefore, it suffices to show that the distributions of the state sequences $(\state^i)_{0 \le i \le \ell}$ and $(\hat{\state}^i)_{0 \le i \le \ell}$ are the same.

With the chain rule: 
\begin{align*}
    \Pr[\state^0, \dots, \state^\ell ] &= \prod_{i=0}^{\ell}\Pr[\state^i \big| \state^0, \dots, \state^{i-1}]~,\\
    \Pr[\hat{\state}^0, \dots, \hat{\state}^{\ell} ] &= \prod_{i=0}^{\ell}\Pr[\hat{\state}^i \big| \hat{\state}^0, \dots, \hat{\state}^{i-1}]~.
\end{align*}

It suffices to show that for any $1\leq i\leq \ell$:
\begin{align*}
\Pr[\state^i \big| \state^0, \dots, \state^{i-1}]
=\Pr[\hat{\state}^i \big| \hat{\state}^0, \dots, \hat{\state}^{i-1}]~.
\end{align*}

The case when $i=0$ follows by Lemma \ref{lem:forest-constructor-automata-stationary}, i.e., by that states $\state^0,\hat{\state}^0$ both follow the common stationary distribution of $\sigma^+, \sigma^-$.
Next, we consider the other conditional probabilities.

For $i_0\leq i\leq \ell$, it follows directly from the memoryless property of probabilistic automata:
\begin{align*}
&\Pr[\state^i \big| \state^0, \dots, \state^{i-1}] = \Pr[\state^i \big| \state^{i-1}]
= \Pr[\hat{\state}^i \big| \hat{\state}^{i-1}] = \Pr[\hat{\state}^i \big|  \hat{\state}^0, \dots, \hat{\state}^{i-1}]~.
\end{align*}

For $1\leq i< i_0$, the memoryless property still holds for $\state^i$. On the other hand, it can be deduced that the memoryless property also holds for $\hat{\state}^i$:
\begin{align*}
\Pr[\hat{\state}^i\big|\hat{\state}^0, \dots, \hat{\state}^{i-1}] &= \frac{\Pr[\hat{\state}^0, \dots,\hat{\state}^i]}{\Pr[ \hat{\state}^0, \dots, \hat{\state}^{i-1}]}\\
&= \frac{\Pr[\hat{\state}^0, \dots, \hat{\state}^{i-2}\big|\hat{\state}^{i-1}, \hat{\state}^i]\cdot \Pr[\hat{\state}^{i-1}, \hat{\state}^i]}{\Pr[\hat{\state}^0, \dots, \hat{\state}^{i-2}\big|\hat{\state}^{i-1}]\cdot \Pr[\hat{\state}^{i-1}]}\\[1ex]
&= \Pr[\hat{\state}^i\big|\hat{\state}^{i-1}]
~.
\end{align*}

It remains to verify that the joint distributions of state pairs $(q^i, q^{i-1})$ and $(\hat{q}^i, \hat{q}^{i_1})$ are identical,
which follows by:
\[
\text{diag}(\vec{\pi})P^+ = 
\begin{bmatrix}
0 & 1-p & p\\
0 & 0 & 1-p\\
1 & 0 & 0
\end{bmatrix}
=\Big(\text{diag}(\vec{\pi})P^-\Big)^T~.
\]

\subsection{Proof of Lemma~\ref{lem:forest-constructor-automata-selection-prob}}

By Lemma~\ref{lem:forest-constructor-automata-matched-state}, to show the first equation it suffices to show for any $i\leq j$:
\[
    \Pr[\state^j=\matched|\state^i=\matched]=f_{i-j}
    ~.
\]

The proof is an induction that corresponds to the recurrence.
First consider the base cases.
The case when $j=i$ is trivial.
The case when $j=i+1$ holds because $\state^i=\matched$ implies $\state^{i+1}=\unmatched$ by the definition of $\sigma^+$.
The case when $j=i+2$ holds because $\state^i=\matched$ implies $\state^{i+1}=\unmatched$, from which $\sigma^+$ transition to $\state^{i+2}=\matched$ with probability $p$ by definition.

Finally consider the case when $j \ge i+3$.
There are only two possibilities in the two rounds after $i$.
The first case is $\state^{i+1} = \unmatched$ and $\state^{i+2} = \matched$, which happens with with probability $p$, and after which $\state^j = \matched$ with probability $f_{j-i-2}$ by the inductive hypothesis.
The other case is $\state^{i+1} = \unmatched$, $\state^{i+2} = \ready$, and $\state^{i+3} = \matched$, which happens with with probability $1-p$, and after which $\state^j = \matched$ with probability $f_{j-i-3}$ by the inductive hypothesis.
Putting together:
\begin{align*}
    \Pr[\state^j=\matched | \state^i=\matched]
    &
    =
    p \cdot f_{j-i-2} + (1-p) \cdot f_{j-i-3}
    \\[1ex]
    &
    =f_{j-i}
    ~.
\end{align*}

Further, according to the definition of $\sigma^+$,  $\state^i=\ready$ if and only if $\state^{i+1}=\matched$ while $\state^i=\unmatched$ if and only if $\state^{i-1}=\matched$. Thus, it follows that:
\begin{align*}
    \Pr \big[ \choice_j = \yes \mid \state^i = \ready \big]
        &
        =
        \Pr \big[ \choice_j = \yes \mid \state^{i+1} = \matched \big]
        =
        f_{j-i-1}
        ~,
        \\
        \Pr \big[ \choice_j = \yes \mid \state^i = \unmatched \big]
        &
        =
        \Pr \big[ \choice_j = \yes \mid \state^{i-1} = \matched \big]
        =
        f_{j-i+1}
        ~.
\end{align*}

\subsection{Proof of Lemma~\ref{lem:forest-constructor-automata-selection-prob-bound}}
By the recurrence, the first seven terms of the sequence $\{f_k\}$ are:
\begin{align*}
f_0&=1~, & f_1&=0~, &  f_2&=p~, & f_3&=1-p~,\\ 
f_4&=p^2~, & f_5&=2p(1-p)~, & f_6&=p^3+(1-p)^2~. & & 
\end{align*}

As a result of $p\geq \frac{\sqrt{5}-1}{2}$, $f_2=p,f_4=p^2\geq f_3=1-p$. For any $k\geq 5$, as $f_k=pf_{k-2}+(1-p)f_{k-3}$, it is easy to see that $f_k\geq f_3=1-p$ by induction.

On the other hand, as $f_6=pf_4+(1-p)f_3$ and $f_4\geq f_3$, $f_6\leq f_4$. Further, as $p\leq \frac{2}{3}$, $f_5=2p(1-p)\geq p^2=f_4\geq f_6$. For any $k\geq 7$, as $f_k=pf_{k-2}+(1-p)f_{k-3}$, it is easy to see that $f_k\geq f_6=p^3+(1-p)^2$ by induction.

\subsection{Proof of Lemma~\ref{lem:pseudo-path-induced}}
Consider any element $e$, any subset of nodes $U\subseteq V$ involving $e$ and any pseudo-path $P=\big((t_i,t'_i)_{e_i}\big)_{1\leq i\leq \ell}$. Consider the subset that is the union of $E^{\text{ex-ante}}_{U,e}\cap P$ and the subset of arcs in $P$ that are adjacent to two distinct arcs of $E^{\text{ex-ante}}_{U,e}\cap P$. Simply by definition, the first statement holds for this subset. Note that any two adjacent arcs in $P$ are either two in-arcs or two out-arcs of a same node, any two adjacent arcs in $P$ have different subscripts. Therefore any two arcs in $E^{\text{ex-ante}}_{U,e}\cap P$ are not adjacent in $P$ and thus each maximal sub-pseudo-path satisfies the fourth statement, i.e. it alternates between arcs with subscript $e$ and arcs with other subscripts, and the second statement, i.e. it is odd-length. Next, we shall show the third statement for this subset. 

By definition of the subset and the fourth statement, it suffices to show that any two arcs with subscript $e$ cannot be two arcs apart in $P$. Since in-arcs of a node are added simultaneously into the same pseudo-path in its corresponding round, the first arcs added into the pseudo-path should be the out-arcs of a node, say the initiative node of $P$. Let $i_0$ be the minimum index among the first arcs of $P$. The initiative node is then $t_{i_0}'$. Let $\{a,b\}$ be the set of elements of initiative node, i.e $\eset^{t_{i_0}'}=\{a,b\}$, and $a$ be the subscript of the $i_0$-th arc, i.e. $e_{i_0}=a$. Then, the subscripts of arcs in $P$ on the pseudo-path can be characterized by the following lemma. 

\begin{lemma} 
The subscripts of arcs in $P$ satisfies:
\begin{enumerate}
	\item For any $1\leq i\leq i_0$, the subscript of the $i$-th arc is $a$ if and only if $i_0-i$ is even; and 
	\item For any $i_0< i\leq \ell$, the subscript of the $i$-th arc is $b$ if and only if $i-i_0$ is odd. 
\end{enumerate}
\end{lemma}
\begin{proof}
For $1\leq i\leq i_0$, we shall prove some stronger results:
\begin{enumerate}
\item For any $1\leq i\leq i_0$, the subscript of the $i$-th arc is $a$ if and only if $i_0-i$ is even;
\item For any $1\leq i\leq i_0$, $t_i'$ involves element $a$;
\item For any $2\leq i\leq i_0$, if $i-i_0$ is odd, $t_{i-1}=t_{i}$ and otherwise $t_{i-1}'=t_i'$.
\end{enumerate}
We shall prove it by induction. The base case is that the subscript of the $i_0$-th arc is $a$, $t_{i_0}'$ involves $a$ and that $t_{i_0-1}=t_{i_0}$. For any $1\leq i\leq i_0-1$, given that the subscript of the $i$-th arc is $a$, $t_i'$ involves $a$ and $t_{i-1}=t_i$, the subscript of the $(i-1)$-th arc can't be $a$, as adjacent arcs have different subscripts, and $t_{i-2}\neq t_{i-1}$ (if $i\geq 3$), as each node has at most 2 our-arcs. Then it is clear $t_i'=t_{i-1}'$ and round $t_i=t_{i-1}, t_i',t_{i-1}'$ have a common element. By Lemma~\ref{lem:app-forest-constructor-common-element}, rounds $t_{i-1}, t_{i-1}'$ have the same set of elements. Therefore $t_{i-1}'$ involves $a$. On the other hand, given that the subscript of the $i$-th arc is not $a$, $t_i'$ involves $a$ and $t_{i-1}'=t_i'$, the subscript of the $(i-1)$-th arc is $a$, as $t_i'$ has two out-arcs and one of them has subscript $a$, and $t_{i-2}'\neq t_{i-1}'$ (if $i\geq 3$). Therefore, it is clear that $t_{i-1}'$ involves $a$ and $t_{i-2}=t_{i-1}$.

If $i_0<\ell$, i.e. there is another in-arc of the initiative node, for any $i_0<i\leq \ell$, it follows the symmetry with $1\leq i\leq i_0$.
\end{proof}

With the characterization of the subscripts, it is clear that for any $i\neq i_0-1$, the subscripts of $i$-th arc and $(i+3)$-th arc are different. The last possible violation of the third statement is that the $(i_0-1)$-th arc of $P$ and the $(i_0+2)$-th arc of $P$ have a common element. Note that these two arcs are out-arcs of the origins of in-arcs of the initiative node because of the characterization, this violation can be ruled out by the following lemma. 

\begin{lemma}
If the two origins of in-arcs of the initiative node have a same element, i.e. $\eset^{t_{i_0}}\cap \eset^{t_{i_0+1}}\neq \emptyset$, at least one of them doesn't extend another out-arc in later rounds. 
\end{lemma}
\begin{proof}
If the two origins are the same, it is clear that the origin has already extended two out-arcs and it can't further extend. Otherwise, suppose the element is $c$. One of the origins have already extended an out-arc with subscript $c$ before round $t_{i_0}'$ and it can't further extend.
\end{proof}

\subsection{Proof of Lemma~\ref{lem:ocs-unmatched-ready-ratio}}

We shall prove a stronger claim that lets the left-hand-side probability be further conditioned on an arbitrary realization of $\state^{i_{m-1}-4}$.
This probability equals:
\[
    \Pr \Big[ \state^{i_{m-1}} = \unmatched \mid X_{i_1}, \dots, X_{i_{m-1}} = 0, \state^{i_{m-1}-4} \Big]
    =
    \Pr \Big[ \state^{i_{m-1}} = \unmatched \mid X_{i_{m-1}} = 0, \state^{i_{m-1}-4} \Big]
    ~.
\]

By Bayes' rule, it is further equal to:
\[
    \frac{\Pr \big[ \state^{i_{m-1}} = \unmatched \mid \state^{i_{m-1}-4} \big]}{\Pr \big[ \state^{i_{m-1}} = \unmatched \mid \state^{i_{m-1}-4} \big] + \Pr \big[ \state^{i_{m-1}} = \ready \mid \state^{i_{m-1}-4} \big]}
    ~.
\]

It remains to show that:
\begin{equation}
    \label{eqn:ocs-unmatched-ready-ratio}
    p \cdot \Pr \big[ \state^{i_{m-1}} = \unmatched \mid \state^{i_{m-1}-4} \big]
    \ge
    \Pr \big[ \state^{i_{m-1}} = \ready \mid \state^{i_{m-1}-4} \big]
    ~.
\end{equation}

Consider the transition matrix $P^+$ of the automaton $\sigma^+$, with columns and rows in the order of $\unmatched$, $\ready$, and $\matched$:
\[
    P^+
    =
    \begin{bmatrix}
        0 & 1-p & p \\
        0 & 0 & 1 \\
        1 & 0 & 0
    \end{bmatrix}
    ~.
\]

The transition after four steps (from $i_{m-1}-4$ to $i_{m-1}$) is:
\[
    \big(P^+\big)^4
    =
    \begin{bmatrix}
        p^2 & (1-p)^2 & 2p(1-p) \\
        p & 0 & 1-p \\
        1-p & p(1-p) & p^2
    \end{bmatrix}
    ~.
\]

Since the first column multplied by $p$ dominates the second column in every entry, we prove Eqn.~\eqref{eqn:ocs-unmatched-ready-ratio}, and thus the lemma.

\subsection{Proof of Lemma~\ref{lem:forest-ocs-structural}}
Consider the collection of all maximal tree-paths consisting of nodes in $U$. 

Note that in the good forest, for any node $p$ with two children $c$ and $c'$, the corresponding rounds have no common element. Since every node in $U$ involves element $e$, for any node $p\in U$, at most one of its children is in $U$. Since the collection consists of maximal tree-paths consisting of nodes in $U$, for any $p,c\in U$ such that $p$ is the parent of $c$, $p$ is on one tree-path in the collection if and only if $c$ is on the tree-path. Therefore, for any node $p$ and any tree-path in the collection, the neighbors (i.e., parent or children) of $p$ in the path is fixed. It is clear that any two distinct tree-paths are disjoint. Moreover, if there is an arc $(p,c)$ between nodes in distinct tree-paths, there should be two tree-paths involving node $p$, which contradicts to the fact the tree-paths are pairwise disjoint.

Consider one path in the collection consisting of nodes $t_1,t_2,\cdot,t_k\in U$. For each $2\leq i\leq k$, if arc $(t_{i-1},t_i)$ have subscript $e$, it is clear that the algorithm sets $\ell_{t_i}(e)=\ell_{t_{i-1}}(e)$. Otherwise, as both node $t_{i-1}$ and $t_i$ involve $e$ and the element corresponding to subscript of the arc, the label of the other element in these nodes clearly imply $\ell_{t_i}(e)=\ell_{t_{i-1}}(e)$. Therefore, element $e$ has the same label in each tree-path.

\subsection{Proof of Lemma~\ref{lem:forest-ocs-from-origin}}
    By symmetry, consider label $\ell = \head$ without loss of generality.
    The lemma holds vacuously for $k \ge 4$ since by design the automaton never selects the same label four times in a roll.%
    \footnote{The longest identical sections are selecting $\head$ in three consecutive rounds from state $\tail^2$.}
    Next we prove the cases of $k = 1, 2, 3$.
    The remaining argument lets $q^j$ denote the state after round $j$, and lets $s^j$ denote the selected label in round $j$.

    By the symmetry of automaton $\sigma^*$, for any round $j$:
    \begin{equation}
        \label{eqn:forest-ocs-automaton-symmetry}
        \Pr \big[ q^j = q_\head \big] = \Pr \big[ q^j = q_\tail \big]
        ~,
        \qquad
        \Pr \big[ q^j = q_{\head^2} \big] = \Pr \big[ q^j = q_{\tail^2} \big]
        ~.
    \end{equation}

    Further, by the above symmetry and by:
    \begin{align*}
        \Pr \big[ q^j = q_{\head^2} \big]
        &
        = \Pr \big[ q^{j-1} = q_\head \big] \Pr \big[ s^j = \head \mid q^{j-1} = q_\head \big] = \Pr \big[ q^{j-1} = q_\head \big] \frac{1-\beta}{2}
        ~,
        \\
        \Pr \big[ q^j = q_\head \big]
        &
        \ge \Pr \big[ q^{j-1} = q_\tail \big] \Pr \big[ s^j = \head \mid q^{j-1} = q_\tail \big] = \Pr \big[ q^{j-1} = q_\tail \big] \frac{1+\beta}{2}
        ~,
    \end{align*}
    we have:
    \begin{equation}
        \label{eqn:forest-ocs-automaton-state-bias}
        \frac{\Pr [ q^j = q_\head ]}{\Pr [ q^j = q_{\head^2} ]}
        \ge
        \frac{1+\beta}{1-\beta}
        ~.
    \end{equation}

    If $k = 1$, the symmetry implies that the marginal probability of selecting each label in round $i$ equals $\frac{1}{2}$. 
    
    If $k = 2$, the probabilities of selecting $\tail$ two consecutive times from each of the states are:
    \begin{align*}
        &
        \Pr \big[ s^i = s^{i+1} = \tail \mid q^{i-1} = q_\origin \big]
        = \Pr \big[ s^i = \tail \mid q^{i-1} = q_\origin \big] \Pr \big[ s^{i+1} = \tail \mid q^i = q_\tail \big]
        = \frac{1-\beta}{4}
        ~,
        \\
        &
        \Pr \big[ s^i = s^{i+1} = \tail \mid q^{i-1} = q_\head \big]
        = \Pr \big[ s^i = \tail \mid q^{i-1} = q_\head \big] \Pr \big[ s^{i+1} = \tail \mid q^i = q_\tail \big]
        = \frac{1-\beta^2}{4}
        ~,
        \\[1.5ex]
        &
        \Pr \big[ s^i = s^{i+1} = \tail \mid q^{i-1} = q_\tail \big]
        = \Pr \big[ s^i = \tail \mid q^{i-1} = q_\tail \big] \Pr \big[ s^{i+1} = \tail \mid q^i = q_{\tail^2} \big]
        = 0
        ~,
        \\[1.5ex]
        &
        \Pr \big[ s^i = s^{i+1} = \tail \mid q^{i-1} = q_{\head^2} \big]
        = \Pr \big[ s^i = \tail \mid q^{i-1} = q_{\head^2} \big] \Pr \big[ s^{i+1} = \tail \mid q^i = q_\origin \big]
        = \frac{1}{2}
        ~,
        \\[1.5ex]
        &
        \Pr \big[ s^i = s^{i+1} = \tail \mid q^{i-1} = q_{\tail^2} \big]
        = 0
        ~.
    \end{align*}

    Hence:
    \[
        \Pr \big[ s^i = s^{i+1} = \tail \big]
        =
        \Pr \big[ q^{i-1} = q_\origin \big] \frac{1-\beta}{4}
        +
        \Pr \big[ q^{i-1} = q_\head \big] \frac{1-\beta^2}{4}
        +
        \Pr \big[ q^{i-1} = q_{\head^2} \big] \frac{1}{2}
        ~.
    \]

    By Eqn.~\eqref{eqn:forest-ocs-automaton-state-bias}, we further get that:
    \begin{align*}
        \Pr \big[ s^i = s^{i+1} = \tail \big]
        &
        \ge
        \Pr \big[ q^{i-1} = q_\origin \big] \frac{1-\beta}{4}
        + 
        \Pr \big[ q^{i-1} = q_\head \text{ or } q_{\head^2} \big]
        \Big( \frac{1+\beta}{2} \frac{1-\beta^2}{4} + \frac{1-\beta}{2} \frac{1}{2} \Big)
        \\
        &
        =
        \Pr \big[ q^{i-1} = q_\origin \big] \frac{1-\beta}{4}
        + 
        \Pr \big[ q^{i-1} = q_\head \text{ or } q_{\head^2} \big] \Big( \frac{(1+\beta)^2}{2} + 1 \Big) \frac{1-\beta}{4}
        \\
        &
        = 
        \Big( \Pr \big[ q^{i-1} = q_\origin \big] + 2 \Pr \big[ q^{i-1} = q_\head \text{ or } q_{\head^2} \big] \Big) \frac{1-\beta}{4}
        \tag{$\beta = \sqrt{2}-1$}
        \\
        &
        = 
        \frac{1-\beta}{4}
        ~.
        \tag{Eqn.~\eqref{eqn:forest-ocs-automaton-symmetry}}
    \end{align*}
    

    If $k = 3$, the automaton must start from $q_{\head^2}$ in order to selet $\tail$ in three consecutive rounds.
    The probability equals:
    \[
        \Pr \big[ q^{i-1} = q_{\head^2} \big]
        \Pr \big[ s^i = \tail \mid q^{i-1} = q_{\head^2} \big]
        \Pr \big[ s^{i+1} = \tail \mid q^i = q_\origin \big]
        \Pr \big[ s^{i+2} = \tail \mid q^{i-1} = q_\tail \big]
        ~.
    \]

    The first term is at most $\frac{1-\beta}{4}$ by Equations~\eqref{eqn:forest-ocs-automaton-symmetry} and \eqref{eqn:forest-ocs-automaton-state-bias}.
    The last three equal $1$, $\frac{1}{2}$, and $\frac{1-\beta}{2}$ respectively.
    Hence:
    \[
        \Pr \big[ s^i = s^{i+1} = s^{i+2} = \tail \big] = \frac{(1-\beta)^2}{16} < \frac{(1-\beta)^2}{8}
        ~.
    \]

\subsection{Proof of Lemma~\ref{lem:forest-ocs-from-other}}
    It suffices to prove it for $i = 2$, since otherwise it reduces to the case of $i = 2$ by conditioning on the state after round $i-2$.
    Further, if the automaton starts from the original state $q_\origin$, it follows from Lemma~\ref{lem:forest-ocs-from-origin}.
    If the automaton starts from $q_{\head^2}$ or $q_{\tail^2}$, it resets back to the original state $q_\origin$ after the first round and once again the lemma reduces to Lemma~\ref{lem:forest-ocs-from-origin}.
    Finally the lemma holds vacuously for $k \ge 4$ since automaton $\sigma^*$ never selects the same label in four consecutive rounds.

    The remaining proof consider starting from $q_\head$ and $q_\tail$ and $k \in \{ 1, 2, 3 \}$.
    By symmetry, we consider label $\ell = \head$ without loss of generality.
    If the automaton starts from $q_\head$:
    \begin{align*}
        \Pr \big[ s^2 = \tail \mid q^0 = q_\head \big]
        &
        = \Pr \big[ s^1 = \head \mid q^0 = q_\head \big] \Pr \big[ s^2 = \tail \mid q^1 = q_{\head^2} \big] \\
        &
        \quad
        + \Pr \big[ s^1 = \tail \mid q^0 = q_\head \big] \Pr \big[ s^2 = \tail \mid q^1 = q_\tail \big]
        \\
        &
        =
    \frac{1-\beta}{2} \cdot 1 + \frac{1+\beta}{2} \cdot \frac{1-\beta}{2}
        \\[.5ex]
        &
        = \frac{1}{2}
        ~;
        \tag{$\beta = \sqrt{2}-1$}
        \\[2ex]
        \Pr \big[ s^2 = s^3 = \tail \mid q^0 = q_\head \big]
        &
        = \Pr \big[ s^1 = \head \mid q^0 = q_\head \big] \Pr \big[ s^2 = \tail \mid q^1 = q_{\head^2} \big] \Pr \big[ s^3 = \tail \mid q^2 = q_\origin \big] \\
        &
        \quad
        + \Pr \big[ s^1 = \tail \mid q^0 = q_\head \big] \Pr \big[ s^2 = \tail \mid q^1 = q_\tail \big] \Pr \big[ s^3 = \tail \mid q^2 = q_{\tail^2} \big]
        \\
        &
        = \frac{1-\beta}{2} \cdot 1 \cdot \frac{1}{2} + \frac{1+\beta}{2} \cdot \frac{1-\beta}{2} \cdot 0
        \\
        &
        = \frac{1-\beta}{4}
        ~;
        \\[2ex]
        \Pr \big[ s^2 = s^3 = s^4 = \tail \mid q^0 = q_\head \big]
        &
        = \Pr \big[ s^1 = \head \mid q^0 = q_\head \big] \Pr \big[ s^2 = \tail \mid q^1 = q_{\head^2} \big]
        \\
        &
        \qquad
        \Pr \big[ s^3 = \tail \mid q^2 = q_\origin \big] \Pr \big[ s^4 = \tail \mid q^3 = q_\head \big]
        \\
        &
        = \frac{1-\beta}{2} \cdot 1 \cdot \frac{1}{2} \cdot \frac{1-\beta}{2}
        \\
        &
        = \frac{(1-\beta)^2}{8}
        ~.
    \end{align*}

    The last case omits the $s^1 = \tail$ option because the automaton by design never selects $\tail$ in four consecutive rounds.

    If the automaton starts from $q_\tail$:
    \begin{align*}
        \Pr \big[ s^2 = \tail \mid q^0 = q_\tail \big]
        &
        = \Pr \big[ s^1 = \head \mid q^0 = q_\tail \big] \Pr \big[ s^2 = \tail \mid q^1 = q_\head \big] \\
        &
        \quad
        + \Pr \big[ s^1 = \tail \mid q^0 = q_\tail \big] \Pr \big[ s^2 = \tail \mid q^1 = q_{\tail^2} \big]
        \\
        &
        =
        \frac{1+\beta}{2} \cdot \frac{1+\beta}{2} + \frac{1-\beta}{2} \cdot 0 
        \\[.5ex]
        &
        = \frac{1}{2}
        ~;
        \tag{$\beta = \sqrt{2}-1$}
        \\[2ex]
        \Pr \big[ s^2 = s^3 = \tail \mid q^0 = q_\tail \big]
        &
        = \Pr \big[ s^1 = \head \mid q^0 = q_\tail \big] \Pr \big[ s^2 = \tail \mid q^1 = q_\head \big] \Pr \big[ s^3 = \tail \mid q^2 = q_\tail \big] \\[1ex]
        &
        = \frac{1+\beta}{2} \cdot \frac{1+\beta}{2} \cdot \frac{1-\beta}{2}
        \\
        &
        = \frac{1-\beta}{4}
        ~.
        \tag{$\beta = \sqrt{2}-1$}
    \end{align*}

    The second case omits the $s^1 = \tail$ option since automaton $\sigma^*$ cannot select $\tail$ in three consecutive rounds starting from state $q_\tail$.
    
    Finally, it is impossible to have $s^2 = s^3 = s^4 = \tail$ starting from state $q^0 = q_\tail$ because from here we cannot have $q^1 = q_{\head^2}$, the only state of automaton $\sigma^*$ that could lead to selecting $\tail$ in the next three rounds.

\subsection{Proof of Lemma~\ref{lem:forest-ocs-fork}}
    Let $q^j, s^j$ denote the states and selected labels of the first copy, and let $\hat{q}^j, \hat{s}^j$ denote those of the second copy.
    If the initial state is the original state $q_\origin$, it follows by Lemma~\ref{lem:forest-ocs-from-origin}.
    If the initial state is $q_{\head^2}$ or $q_{\tail^2}$, the lemma holds because the first selections in the two copies are the same.

    By symmetry, we next without loss of generality that the initial state is $q^0 = \hat{q}^0 = q_\head$.
    From $q_\head$ it is impossible to select $\head$ in the next two rounds, or to select $\tail$ in the next three rounds.
    Hence, the lemma follows if $k \ge 2$ or $\hat{k} \ge 3$.
    It remains to consider $k = 1$, and $\hat{k} = 1, 2$.

    If $k = \hat{k} = 1$:
    \begin{align*}
        \Pr \big[ s^1 = \head, \hat{s}^1 = \tail \mid q^0 = \hat{q}^0 = q_\head \big]
        &
        = 
        \Pr \big[ s^1 = \head  \mid q^0 = q_\head \big]
        \Pr \big[ \hat{s}^1 = \tail \mid \hat{q}^0 = q_\head \big]
        \\[1ex]
        &
        = \frac{1-\beta}{2} \frac{1+\beta}{2}
        \\
        &
        < \frac{1}{4}
        ~.
    \end{align*}

    If $k = 1$ and $\hat{k} = 2$:
    \begin{align*}
        \Pr \big[ s^1 = \head, \hat{s}^1 = \hat{s}^2 = \tail \mid q^0 = \hat{q}^0 = q_\head \big]
        &
        = 
        \Pr \big[ s^1 = \head \mid q^0 = q_\head \big]
        \Pr \big[ \hat{s}^1 = \hat{s}^2 = \tail \mid \hat{q}^0 = q_\head \big]
        \\[1ex]
        &
        = \frac{1-\beta}{2} \frac{1-\beta^2}{4}
        \\
        &
        < \frac{1-\beta}{8}
        ~.
    \end{align*}

\section{Missing Proofs in Section~\ref{sec:matching}}

\subsection{Proof of Theorem~\ref{thm:lp-solution}}

\label{app:lp-solution}

\begin{proof}
    We first verify the \emph{feasibility} of the stated solution.
    Constraint~\eqref{eqn:two-way-lp-gain-split} holds with equality by the definitions of $a(k)$ and $b(k)$.

    Constraint~\eqref{eqn:approximate-dual-feasible} also holds with equality.
    When $k = 0$, it follows by:
    \begin{align}
        \Gamma
        &
        = p(0) - \sum_{i=0}^{\infty} \Big(\Big(\frac{2}{3}\Big)^i - \Big(\frac{2}{3}\Big)^{i+1}\Big) p(i)
        \tag{Definition of $\Gamma$, $p(0) = 1$}
        \\
        &
        = \sum_{i=0}^\infty \Big(\frac{2}{3}\Big)^{i+1} p(i) - \sum_{i=1}^\infty \Big(\frac{2}{3}\Big)^i p(i)
        \notag
        \\
        &
        =
        \sum_{i=0}^\infty \Big(\frac{2}{3}\Big)^{i+1} \big( p(i)-p(i+1) \big) 
        ~.
        \label{eqn:two-way-lp-alternative-Gamma}
    \end{align}

    This equals $2 b(0)$ by definition.
    Then, it further holds inductively for $k \ge 1$ because:
    \begin{align*}
        b(k)
        &
        = \frac{3}{2} b(k-1) - \frac{1}{2} \big( p(k-1) - p(k) \big)
        \tag{Definition of $b(k), b(k-1)$}
        \\
        &
        = b(k-1) - \frac{1}{2} a(k-1)
        ~.
        \tag{Definition of $a(k-1)$}
    \end{align*}

    That is, the left-hand-side of Constraint~\eqref{eqn:approximate-dual-feasible} stays the same from $k-1$ to $k$.
    Since the above equation $b(k) = b(k-1) - \frac{1}{2} a(k-1)$ would also imply Constraint~\eqref{eqn:two-way-lp-monotone} provided that $a(k-1) \ge 0$, it remains to verify that $a(k)$ and $b(k)$ are nonnegative.
    By its definition and by $p(k+1) \le \frac{2}{3} p(k)$, we get that $b(k) \ge 0$.
    The non-negativity of $a(k)$ follows by its definition and by:
    \begin{align*}
        b(k)
        &
        = \frac{1}{3} p(k) - \sum_{i=k+1}^\infty \Big( \Big(\frac{2}{3}\Big)^{i-k-1} - \Big(\frac{2}{3}\Big)^{i-k} \Big) p(i) 
        \tag{Definition of $b(k)$}
        \\
        &
        \le \frac{1}{3} p(k)
        \\[2ex]
        &
        \le p(k) - p(k+1)
        ~.
        \tag{$p(k+1) \le \frac{2}{3} p(k)$}
    \end{align*}

    Next we establish its \emph{optimality}.
    Multiplying Constraint~\eqref{eqn:approximate-dual-feasible} by $\big(\frac{2}{3}\big)^k$ and summing over $k \ge 0$:
    \[
        \sum_{k=0}^\infty \Big( \frac{2}{3} \Big)^k \Big( \sum_{i=0}^{k-1} a(i)+2 b(k) \Big) \geq \sum_{k=0}^\infty \Big( \frac{2}{3} \Big)^k \Gamma
        ~.
    \]
    
    Grouping terms on the left and dividing both sides by $3$, this is:
    \[
        \sum_{i=0}^\infty \Big(\frac{2}{3}\Big)^{i+1} a(i) + \sum_{i=0}^\infty \Big(\frac{2}{3}\Big)^{i+1} b(i) \ge \Gamma
        ~.
    \]

    Further by Constraint~\eqref{eqn:two-way-lp-gain-split}, the left-hand-side is at most:
    \[
        \sum_{i=0}^\infty \Big(\frac{2}{3}\Big)^{i+1} \big( p(i)-p(i+1) \big)
        ~.
    \]

    This equals the optimal $\Gamma$ in the theorem by Eqn.~\eqref{eqn:two-way-lp-alternative-Gamma}.
\end{proof}

\subsection{Proof of Theorem~\ref{thm:multi-way-lp-solution}}
\label{app:multi-way-lp-solution}

\begin{proof}
    We first verify its \emph{feasibility}.
    Constraint Eqn.~\eqref{eqn:multi-way-lp-gain-split} holds with equality by definition, i.e.:
    \begin{equation}
        \label{eqn:multi-way-lp-solution-gain-split}
        a(y) + b(y) = -p'(y)
        ~,
        \qquad\qquad
        \forall y \ge 0
        ~.
    \end{equation}

    Constraint Eqn.~\eqref{eqn:multi-way-approximate-dual-feasible} holds for $y = 0$ from integration by parts:
    \[
        b(0) = - \int_0^\infty p'(z) e^{-z} dz = \int_{0}^\infty e^{-z} \big( 1 - p(y) \big) dz = \Gamma
        ~.
    \]

    It further holds for $y > 0$ since its left-hand-side is a constant for all $y$.
    Indeed, the derivative of the left-hand-side is:
    \begin{align*}
        a(y) + b'(y)
        &
        = a(y) - e^y \int_y^\infty p'(z)e^{-z} dz + p'(y)
        \tag{Definition of $b(y)$}
        \\
        &
        = a(y) + b(y) + p'(y)
        \tag{Definition of $b(y)$}
        \\[2ex]
        &
        = 0
        ~.
        \tag{Eqn.~\eqref{eqn:multi-way-lp-solution-gain-split}}
    \end{align*}
    
    Constraint~\ref{eqn:multi-way-lp-monotone} holds, i.e., $b(y)$ is decreasing because:
    \begin{align*}
        b'(y)
        &
        = - e^y \int_y^\infty p'(z)e^{-z} dz + p'(y)
        \tag{Definition of $b(y)$}
        \\
        &
        = - e^y \int_y^\infty \big( p'(y) - p'(z) \big) e^{-z} dz
        \\[.5ex]
        &
        \le 0
        ~.
        \tag{Convexity of $p$}
    \end{align*}

    Finally, $b(y)$ is non-negative by definition and by that $p$ is decreasing.
    The non-negativity of $a(y)$ follows by:
    \begin{align*}
        a(y)
        &
        = -p'(y) + e^y \int_y^\infty p'(z) e^{-z} dz
        \tag{Definitions of $a(y), b(y)$}
        \\
        &
        = e^y \int_{y}^\infty e^{-z} p''(z)dz
        \tag{Integration by parts}
        \\[1ex]
        &
        \ge 0
        ~.
        \tag{Convexity of $p$}
    \end{align*}
\end{proof}

\end{document}